\documentclass[letterpaper,12pt,reqno]{article}

\usepackage[letterpaper, left=2.5cm,right=2.5cm,top=2.5cm,bottom=2.5cm]{geometry}
\usepackage{setspace}

\usepackage[T1]{fontenc}
\usepackage{newtxtext}
\usepackage{microtype}

\usepackage[hyphens]{url}
\usepackage[hyphenbreaks]{breakurl}
\usepackage{hyperref}
\hypersetup{colorlinks=true, linkcolor=blue, urlcolor  = blue,  citecolor=black}

\usepackage{titlesec}
\usepackage[utf8]{inputenc}
\usepackage{graphicx}
\usepackage{amsmath}
\usepackage{MnSymbol}%
\usepackage{wasysym}%
\graphicspath{{Images/}}
\usepackage{amsfonts}
\usepackage{commath}
\usepackage{algpseudocode}
\usepackage{enumitem} 
\usepackage{xcolor}
\usepackage[english]{babel}
\usepackage{amsthm}
\usepackage{tocloft}
\usepackage{bbm}
\usepackage{mathtools}
\usepackage{tikz}
\usepackage{comment}
\usetikzlibrary{decorations.pathreplacing}
\usetikzlibrary{positioning,quotes,calc}
\usepackage{tikz-3dplot}
\usepackage{subcaption}
\usetikzlibrary{calc,intersections}

\definecolor{niceblue}{RGB}{0,104,178}
\definecolor{darkgreen}{RGB}{3, 131, 127}
\definecolor{orange}{RGB}{213,94,0}
\definecolor{darkorange}{RGB}{191,84,0}
\definecolor{paleniceblue}{RGB}{118,171,208}
\definecolor{camniceblue}{RGB}{145, 182, 177}
\definecolor{grey}{RGB}{60, 62, 61}
\definecolor{purple}{RGB}{153, 51, 255}
\definecolor{red}{RGB}{255, 51, 51}
\definecolor{darkdarkgreen}{RGB}{10, 105, 102}
\definecolor{darkred}{RGB}{148, 47, 47}

\definecolor{niceblue}{RGB}{0,104,178}
\definecolor{darkgrey}{RGB}{3, 131, 127}
\definecolor{orange}{RGB}{213,94,0}
\definecolor{darkorange}{RGB}{191,84,0}
\definecolor{paleniceblue}{RGB}{118,171,208}
\definecolor{camniceblue}{RGB}{145, 182, 177}
\definecolor{grey}{RGB}{60, 62, 61}
\definecolor{purple}{RGB}{153, 51, 255}
\definecolor{red}{RGB}{255, 51, 51}
\definecolor{darkdarkgrey}{RGB}{10, 105, 102}
\definecolor{darkred}{RGB}{148, 47, 47}

\definecolor{teal}{RGB}{0, 150, 136}        
\definecolor{peach}{RGB}{255, 160, 122}      
\definecolor{lavender}{RGB}{178, 102, 255}  
\definecolor{midnightniceblue}{RGB}{25, 25, 112} 
\definecolor{olive}{RGB}{128, 128, 0}       
\definecolor{mustard}{RGB}{204, 153, 0}      
\definecolor{rose}{RGB}{255, 102, 102}       
\definecolor{plum}{RGB}{102, 0, 102}

\newtheorem*{theorem*}{Theorem}
\newtheorem*{corollary*}{Corollary}
\newtheorem*{proposition*}{Proposition}
\newtheorem*{lemma*}{Lemma}
\newtheorem*{fact*}{Fact}
\newtheorem*{definition*}{Definition}
\newtheorem*{conjecture*}{Conjecture}
\newtheorem{theorem}{Theorem}
\newtheorem{example}{Example}
\newtheorem{corollary}{Corollary}
\newtheorem{proposition}{Proposition}
\newtheorem{assumption}{Assumption}
\newtheorem{definition}{Definition}
\newtheorem{lemma}{Lemma}
\newtheorem{problem}{Problem}
\newtheorem{remark}{Remark}
\newtheorem{fact}{Fact}

\DeclarePairedDelimiterX{\inp}[2]{\langle}{\rangle}{#1, #2}

\newcommand{\ool}[1]{\overline{#1}}
\newcommand{\uul}[1]{\underline{#1}}
\DeclareMathOperator*{\argmax}{arg\,max}

\usepackage[font=small, width=0.8\textwidth]{caption}

\titleformat{\section}
		{\large\bfseries\center}
         {\thesection}
        {0.5em}
        {}
        []
\titleformat{\subsection}
        {\normalfont\bfseries}
        {\thesubsection}
        {0.5em}
        {}
\titleformat{\subsubsection}[runin]
        {\normalfont\bfseries}
        {\thesubsubsection}
        {0.5em}
        {}
        [.]

\AtBeginDocument{
  \setlength{\abovedisplayskip}{9pt}
  \setlength{\belowdisplayskip}{9pt}
  \setlength{\abovedisplayshortskip}{1pt}
  \setlength{\belowdisplayshortskip}{6pt}
  \setlength{\jot}{5pt}
}

\usepackage[authoryear]{natbib}
\usepackage{soul}
\title{Targeting Without Transfers\thanks{I am grateful to Itai Ashlagi, Ben Brooks, Eric Budish, Laura Doval, Piotr Dworczak, Joey Feffer, Jason Hartline, Ben Hébert, Jan Kożuszek, Jacob Leshno, Shengwu Li, Axel Niemeyer, Michael Ostrovsky, Ilya Segal, Andrzej Skrzypacz, Takuo Sugaya, Vasiliki Skreta, Frank Yang, and Leeat Yariv for their helpful comments and suggestions.}}
\author{Filip Tokarski \\ Stanford GSB}

\begin{document}
\date{\today} 
\maketitle
\vspace*{-1cm} 
\begin{abstract}
I study the welfare-maximizing allocation of heterogeneous goods when monetary transfers are prohibited. Agents have private values, and the designer chooses a mechanism subject to incentive compatibility and aggregate supply constraints. I characterize when the optimal mechanism takes the form of a simple menu, where each option offers some amount of one kind of good and none of the others. When this is the case, it can be implemented as a competitive equilibrium with equal incomes. However, this mechanism can be suboptimal when narrow preference margins between pure options are sufficiently predictive of higher values. In such cases, the designer can improve welfare by offering mixed bundles, targeting high-value agents through their willingness to accept mixing.
\end{abstract}

\section{Introduction}

When designing mechanisms without transfers, it is often natural to evaluate them using criteria that avoid interpersonal utility comparisons. This approach is especially appealing when the policymaker has explicitly non-welfarist goals, such as fairness, or when participants' valuations for the allocated goods are plausibly similar. Indeed, the literature on mechanisms without money has largely focused on notions based on Pareto efficiency and ordinal welfare rankings.\footnote{See, among others, \cite{hylland1979efficient,abdulkadirouglu1998random, bogomolnaia2001new, budish2011combinatorial, abdulkadirouglu2013matching}.} Nevertheless, criteria agnostic to cardinal values are less fitting for settings like social programs, where policymakers view applicants as differing sharply in terms of need and aim to target those for whom receiving the goods has the greatest social value. For instance, affordable housing programs in many European countries serve a broad population, including families facing eviction as well as middle-class households with stable employment \citep{whitehead2007social}. In the U.S. context, \cite{cook2023build} find that affordable housing recipients differ substantially in various measures of need, and that this heterogeneity persists even after conditioning on observables.

Motivated by such settings, I consider a mechanism design problem without transfers where the designer allocates heterogeneous goods to maximize cardinal welfare. Importantly, agents' valuations are their private information; this prevents the designer from simply giving the available supply to those who need it most. The main lesson from my analysis is that the optimal mechanism depends on the statistical relationship between agents' \emph{absolute} level of need and the strength of their \emph{relative} preferences across options. When higher-value agents tend to be \emph{more picky}, that is, have stronger preferences between the offered goods, the optimal mechanism often takes the form of a simple menu of \emph{pure options}: each option offers some amount of one kind of good and none of the others. This mechanism can be implemented in several natural ways, including as a competitive equilibrium with equal incomes---a mechanism often proposed for its desirable fairness and Pareto-efficiency properties \citep{varian1973equity,hylland1979efficient,budish2011combinatorial,azevedo2019strategy}. However, when higher-value agents tend to be \emph{less picky}, the designer can sometimes improve upon this mechanism by letting agents choose larger mixed bundles alongside smaller pure options. Intuitively, bundles reward agents willing to give up choice with a larger total allocation. When such willingness signals greater need, this directs more resources toward higher-value agents in an incentive-compatible way.

Both patterns of correlation between need and selectivity are plausible in different settings. For instance, \citet{cook2023build} find that lower-income households are less selective when applying for affordable housing: they are more willing to trade off assignment to a preferred unit for a higher probability of receiving an offer \emph{somewhere}.  In other settings, however, the correlation may go in the opposite direction. For example, consider school choice environments with specialized curricula, such as dual-language immersion programs. Families who place disproportionate weight on admission to such programs may do so because of a child's idiosyncratic needs, aptitudes, or interests. In such cases, an intense \emph{relative} preference for a particular option may instead signal a higher \emph{absolute} value for receiving it.

I establish three sets of results. First, I characterize when the mechanism offering only pure options is optimal. The condition is stated in terms of monotone transports on an appropriately constructed signed measure, which captures the designer's marginal value of increasing different types' allocations. I then derive a simpler sufficient condition in the special case of symmetric goods: there, the pure option mechanism is optimal if agents whose cardinal values for their favorite goods are higher tend to be more selective, in a precise stochastic sense.

Second, I fully characterize the welfare-maximizing mechanism when there are two kinds of goods. The mechanism always offers two pure options and sometimes also includes a
single mixed bundle. In the symmetric case, I show that introducing the mixed bundle is welfare-improving precisely when weak relative preferences are sufficiently predictive of higher total value, so that the informational gain from targeting outweighs the allocative inefficiency from mixing. 

Third, I extend the model by allowing agents to differ in characteristics that the designer can observe. Indeed, targeting based on such verifiable tags is common in social programs, where priority often depends on earnings or family size. I show that when the designer can condition the allocation on observables, she effectively faces a set of separate screening problems---one for each tag---linked through a common supply constraint. Her problem can then be decomposed into two stages: an \emph{inner stage}, in which goods are allocated optimally within each group, and an \emph{outer stage}, in which the total supply is divided across groups. This supply split has a particular structure: each group is either \emph{active} and receives positive amounts of every good, or \emph{excluded} and receives nothing. Intuitively, a group is active when it has a ``comparative advantage'' in one of the goods, that is, when agents with the corresponding tag tend to value that kind of good significantly more than agents in other groups do. I also show that observable heterogeneity cannot generally be accommodated by simply assigning each group a different budget in a common pseudomarket, an approach used, for example, by Feeding America to allocate donations across food banks. This is because, generically, the designer can do better by tailoring each group's supply to its comparative advantages.

I then discuss the implications of my results for market design in settings such as allocating donations to food banks and affordable housing. In these settings, existing mechanisms can be improved when the agents the designer wants to prioritize are also those willing to accept less tailored allocations. For food banks, the results suggest a natural modification of the CEEI mechanism currently used to allocate donations: the system could offer discounts for bundles of food, allowing food banks that receive little support from their local donors to get more food cheaply. In the case of affordable housing, the model speaks to the design of waitlists and repeated lotteries. Indeed, I show in Appendix~\ref{app:waitlist} that the baseline model can be interpreted as a reduced-form description of such mechanisms, once allocations are redefined appropriately. There, the program could target ``less picky'' applicants by offering a priority waitlist that requires them to accept any suitable unit, alongside longer waitlists for specific units or locations.

The rest of the paper is structured as follows. Section~\ref{sec:lit} discusses the related literature. Section~\ref{sec:model} presents the model, and Section~\ref{sec:examples} illustrates the core intuitions through two-good examples. Section~\ref{sec:renormalization} formally separates absolute and relative values, and Section~\ref{sec:pureoption} introduces the pure option mechanism and its implementations. Section~\ref{sec:CEEI} characterizes when this mechanism is optimal, and Section~\ref{sec:2good} fully characterizes the optimal mechanism with two kinds of goods. Section~\ref{sec:observable-groups} extends the model to observable groups, Section~\ref{sec:discussion} discusses implications for market design, and Section~\ref{sec:proofdetails} outlines the proof of the main characterization.

\section{Related literature}\label{sec:lit}

My paper contributes to the literature on allocating heterogeneous goods without transfers, and connects most directly to work on pseudomarkets and competitive equilibrium with equal incomes (CEEI). \citet{hylland1979efficient} introduce CEEI in an assignment setting by giving agents equal budgets of artificial currency and letting them purchase probability shares in goods; the resulting lotteries are ex ante Pareto-efficient. Subsequent work has developed pseudomarket mechanisms for richer environments, including approximate CEEI for combinatorial assignment \citep{budish2011combinatorial} and pseudomarkets with priorities and other constraints \citep{he2018pseudo}. A related literature also studies mechanisms that elicit only ordinal rankings over options. Random serial dictatorship lets agents choose sequentially in random order \citep{abdulkadirouglu1998random}, while probabilistic serial assigns agents fractional shares by letting them continuously claim their most-preferred available objects, yielding ordinally efficient outcomes \citep{bogomolnaia2001new}. Other papers analyze the large-market properties of these mechanisms: \citet{azevedo2019strategy} formalize a notion of strategy-proofness in the large for assignment and related mechanisms, while \citet{che2010asymptotic} show that random serial dictatorship and probabilistic serial become asymptotically equivalent.

While the study of allocating heterogeneous goods without transfers has focused mainly on criteria that avoid interpersonal utility comparisons, a smaller body of work allows for \emph{cardinal} objectives and studies mechanisms that maximize them \citep{miralles2012cardinal,CHAKRAVARTY20131,ashlagi2016optimal,dogan2020welfare,akyol2025allocation}. My paper is closest to \citet{miralles2012cardinal}, who studies welfare-maximizing mechanisms with two symmetric goods and finitely many symmetric agents. He characterizes the optimal mechanism under regularity conditions on the value distribution, showing that it combines lotteries, auctions, and insurance, and that under these conditions the optimum approaches a pseudomarket as the number of agents grows large. I instead focus on large markets directly and consider an arbitrary number of goods without symmetry restrictions. I characterize when the pseudomarket is welfare-maximizing and, when these conditions fail, show how the designer can use mixed bundles to screen agents by exploiting the relationship between the strength of their relative preferences and their absolute values.

A related literature studies eliciting preference intensities—information about how strongly agents prefer some options over others. In school choice, \citet{bostonreconsidered} observe that the Boston mechanism can elicit the \emph{extent} to which families prefer certain schools---a property that deferred acceptance does not have. \citet{abebe2020truthful} show that, in one-sided matching, a
truthful mechanism eliciting cardinal values can substantially outperform any
ordinal mechanism in approximating the Nash bargaining solution. In a paper closely related to mine, \cite{ortoleva2021cares} consider optimal mechanisms in a setting without transfers where agents have a common ranking over goods but differ in their sensitivity to quality. My paper, by contrast, does not impose such structure and considers heterogeneous goods. This leads to different and complementary results. Indeed, the authors show that the first-best allocations may offer lotteries between qualities, and that second-best allocations always involve lotteries and may involve free disposal; neither of these results holds in my setting. Like me, they show that CEEI allocations, despite being Pareto-efficient, do not always maximize welfare. 

Finally, my work builds on methods developed in the multidimensional screening literature. To obtain a characterization of the optimality of the pure option mechanism, I invoke ideas used in the study of the multi-product monopoly problem \citep{armstrong1996multiproduct,rochet1998ironing,manelli2006bundling,kleiner2019strong}. In particular, my optimality conditions rely on stochastic dominance and transport arguments related to those in \cite{daskalakis2013mechanism,daskalakis}.

\section{Model}\label{sec:model}

The designer has $N$ different kinds of goods indexed by $i\in \{1,\dots,N\}$ with $N\geq 2$. She possesses a fixed mass of each, with the supplies given by $s=(s_1,s_2,\dots,s_N)> 0$. There is a unit mass of agents, each of whom has a profile of values $v= (v_1,v_2,\dots,v_N)$ for the goods; the values are private information and are distributed according to $F$ with support on a bounded set $\mathcal{V}\subset \mathbb{R}^N_{+}$. This distribution puts no mass on \(\mathbf 0\) and assigns positive mass to agents with \(v_i>0\) for every good \(i\). The designer chooses an allocation rule for the goods, $y=(y_1,y_2,\dots,y_N):\mathcal{V}\to \mathbb{R}_+^N$, to maximize utilitarian welfare:
\begin{equation}\label{eq:obj}
\int_{\mathcal{V}} v\cdot y(v) d F(v). \tag{O}
\end{equation}
She faces incentive compatibility and supply constraints:
    \begin{equation}\label{eq:IC}
      v\cdot y(v) \geq v\cdot y(v') \quad \text{for all }v, v', \tag{IC}
    \end{equation}
    \begin{equation}\label{eq:S}
      \int y(v) d F(v) \leq s. \tag{S}
     \end{equation}
An allocation rule $y:\mathcal{V}\to \mathbb{R}_+^N$ is \emph{feasible} if it satisfies \eqref{eq:IC} and \eqref{eq:S}.

\subsection{Discussion of the model}\label{rem:model-dis}

I briefly discuss the interpretation of the model. First, it assumes that the designer's preferences align with those of agents: if an agent chooses one allocation over another, the designer also considers giving her that allocation more socially valuable.

Second, agents' values \(v\) admit several interpretations. For instance, \(v_i\) may represent an agent's latent willingness to pay for a unit of good \(i\). Although such values cannot be directly elicited without money, they remain meaningful for the designer's welfare objective. More generally, \(v\) may represent the designer's assessment of the social value of allocating goods to different agents. She may, for example, place greater weight on individuals with greater need or vulnerability and believe that these characteristics are correlated with the preferences agents reveal over the available goods.

It is also useful to clarify how my environment differs from a standard mechanism-design problem with monetary transfers. The key difference is that agents cannot bring \emph{their own money} into the system and therefore are unable to signal higher values for the offered goods through their willingness to pay. Formally, this is captured by the restriction that all allocations must be nonnegative. That said, one of the allocated goods can still be interpreted as money. For example, my model can describe a housing program that allocates both a fixed stock of housing and a fixed budget for housing vouchers, or a university that distributes several pools of funding, earmarked for different purposes, across laboratories or researchers.

Finally, note that at baseline we interpret \(y_i(v)\) as the \emph{amount} of good \(i\) given to type \(v\). However, some settings of interest feature unit demand, so applying the model to them requires interpreting allocations as \emph{probabilities} of receiving a certain good. This in turn requires imposing an additional constraint that $\sum y_i(v)\le 1$ for all $v$. While this restriction changes the problem in general, the following result shows that imposing it is without loss when supply is sufficiently scarce relative to the population.

\begin{proposition}\label{prop:P_redundant}
Consider the model augmented with the probability constraint
\begin{equation}\label{eq:P}
  \sum y_i(v) \le 1 \quad \text{for all } v\in\mathcal V. \tag{P}
\end{equation}
There exists $\bar\eta>0$ such that for every $\eta\in(0,\bar\eta]$ and every allocation rule
$y:\mathcal V\to\mathbb R_+^N$ that is feasible with supplies $\eta s$, constraint
\eqref{eq:P} is slack for all $v\in\mathcal V$.
\end{proposition}

Intuitively, this is because when overall supply is sufficiently small, the designer cannot afford to offer any option that delivers some good with certainty. If she did, the mass of agents requesting such an option would be so large that supply constraints would be violated. This kind of extreme mismatch between demand and supply is plausible in settings like public housing lotteries, where units are exceptionally scarce relative to the number of applicants.\footnote{Under this unit-demand interpretation, \(y_i(v)\) is the probability that a type-\(v\) agent receives good \(i\). Given appropriate technical conditions, the exact law of large numbers allows such interim probabilities to be implemented by independent individual lotteries that satisfy exact aggregate feasibility; see \citet{sun2006exact}.\label{foot:sun}}

Proposition~\ref{prop:P_redundant} also implies that the set of feasible allocation rules is uniformly bounded for any supply $s$. A standard compactness argument then shows that an optimal mechanism exists.

\section{Examples}\label{sec:examples}

To preview the paper's intuitions, I begin with illustrative examples
featuring just two goods.

\begin{example}\label{example:ex3}
Fix any supplies $s_1,s_2>0$ and let values follow an absolutely
continuous distribution with full support on $[0,1]^2$.
Consider the following mechanism, which offers agents two options:
\[
\left\{\text{$q_1$ of good $1$}\right\}, \quad
\left\{\text{$q_2$ of good $2$}\right\}.
\]
The quantities $q_1,q_2$ are chosen so that the supply constraint
holds with equality when all agents pick their preferred option.
These quantities are uniquely pinned down by the value distribution
and the supplies.
\end{example}

\begin{figure}[h!]
  \centering
  \vspace*{-5.2mm}
  \begin{tikzpicture}[scale=0.75]
    \def\L{5}
    \def\slope{1.4}
    \pgfmathsetmacro{\xint}{\L/\slope}

    \fill[niceblue, opacity=0.3]
      (0,0) -- (0,\L) -- (\xint,\L) -- cycle;

    \fill[mustard, opacity=0.3]
      (0,0) -- (\xint,\L) -- (\L,\L) -- (\L,0) -- cycle;

    \draw[->] (0,0) -- (\L+0.1,0) node[right] {$v_1$};
    \draw[->] (0,0) -- (0,\L+0.1) node[above] {$v_2$};

    \pgfmathsetmacro{\Bx}{\xint/3}
    \pgfmathsetmacro{\By}{2*\L/3}
    \node[darkdarkgreen] at (\Bx,\By)
      {\fontsize{15}{15}\selectfont $q_2$ of $2$};

    \pgfmathsetmacro{\Ax}
      {(\L*\L - \xint*\xint/3)/(2*\L - \xint)}
    \pgfmathsetmacro{\Ay}
      {\L*(3*\L - 2*\xint)/(3*(2*\L - \xint))}
    \node[darkdarkgreen] at (\Ax,\Ay)
      {\fontsize{15}{15}\selectfont $q_1$ of $1$};
  \end{tikzpicture}
  \caption{Allocation under the mechanism in
    Example~\ref{example:ex3}.}
  \label{fg:ceeiexample2}
\end{figure}

Under this mechanism, agents for whom $q_1v_1>q_2v_2$ select the
former option, while those for whom $q_1v_1<q_2v_2$ select the
latter. As shown in Figure~\ref{fg:ceeiexample2}, these two sets of
types are separated by a ray from the origin defined by
\begin{equation}\label{eq:rayexample1}
\frac{v_1}{v_2}=\frac{q_2}{q_1}.
\end{equation}

Two features of this allocation are important. First, it is
Pareto-efficient among all allocation rules satisfying the supply
constraint~\eqref{eq:S}. Second, it depends only on the
\emph{ratio} of agents' values for goods $1$ and $2$, but not on
how large $v_1$ and $v_2$ are in absolute terms. Indeed, while the
designer would like to allocate the goods only to agents whose
$v_1$ and $v_2$ are high, she cannot do so in an
incentive-compatible way. She is therefore restricted to
conditioning the allocation on how much agents like one good
\emph{relative to the other}. This highlights a useful
distinction: an agent's \emph{absolute values}, $(v_1,v_2)$,
capture the overall intensity of need for the goods, while her
\emph{relative values}, which I define as
$(\tfrac{v_1}{v_1+v_2},\tfrac{v_2}{v_1+v_2})$, capture the
strength of her preference between the goods. Since all agents
with the same relative values always rank all offered options
the same way, no incentive-compatible mechanism can meaningfully
elicit absolute values among agents with the same profile of
relative values.

In fact, as I show in the course of this paper, this simple mechanism is optimal for a broad
class of value distributions, including the uniform distribution
on $[0,1]^2$. However, this is not always the case: for other value distributions, the designer
may want to offer a richer menu.

\begin{example}\label{example:ex2}
Let \(s_1=s_2\) and fix \(\varepsilon\in(0,1)\). An \(\varepsilon\)-share of agents have values distributed uniformly on \([0,1]^2\). A \((1-\varepsilon)/2\)-share of agents have values distributed uniformly on
\[
[1-\varepsilon,1]^2.
\]
Finally, two groups, each of mass \((1-\varepsilon)/4\), have values distributed uniformly on
\[
[0,1]\times[0,\varepsilon] \qquad\text{and}\qquad [0,\varepsilon]\times[0,1].
\]
Consider a mechanism that offers agents three options:
\[
\left\{\text{\(q_L\) of good \(1\)}\right\}, \qquad
\left\{\text{\(q_L\) of good \(2\)}\right\}, \qquad
\left\{\text{\(q_B\) of good \(1\) and \(q_B\) of good \(2\)}\right\},
\]
where \(q_L/2<q_B<q_L\). The quantities are scaled so that both supply constraints hold with equality when all agents pick their preferred option.
\end{example}

Under this mechanism, each agent can pick between a small amount of her favorite good and a bundle of the two goods that offers a larger total amount. Agents with strong relative preferences between the two goods pick the pure allocations, while those with narrow preference margins choose the bundle.

\begin{figure}[h!]
    \centering
    \begin{subfigure}[b]{0.4\linewidth}
        \centering
        \begin{tikzpicture}[scale=0.75]
            \pgfmathsetmacro{\L}{5}
            \pgfmathsetmacro{\eps}{0.65}
            \pgfmathsetmacro{\H}{\L-\eps}
            \fill[gray!20] (0,0) rectangle (\L,\L);
            \fill[brown, opacity=0.45] (0,0) rectangle (\L,\eps);
            \fill[brown, opacity=0.45] (0,0) rectangle (\eps,\L);
            \fill[darkdarkgreen!70, opacity=0.55] (\H,\H) rectangle ({\H+\eps},{\H+\eps});
            \draw[->] (0,0) -- (\L+0.1,0) node[right] {$v_1$};
            \draw[->] (0,0) -- (0,\L+0.1) node[above] {$v_2$};
        \end{tikzpicture}
        \caption{Value distribution in Example~\ref{example:ex2}.}
        \label{fig:example1}
    \end{subfigure}
    \hspace{1.3cm}
    \begin{subfigure}[b]{0.4\linewidth}
        \centering
        \begin{tikzpicture}[scale=0.75]
            \draw[->] (0,0) -- (5.1,0) node[right] {$v_1$};
            \draw[->] (0,0) -- (0,5.1) node[above] {$v_2$};
            \fill[niceblue, opacity=0.3] (0,5) -- (0,0) -- (2.5,5) -- cycle;
            \fill[mustard, opacity=0.3] (5,0) -- (5,2.5) -- (0,0) -- cycle;
            \fill[plum, opacity=0.3] (0,0) -- (2.5,5) -- (5,5) -- (5,2.5) -- cycle;
            \node[darkdarkgreen] at (0.95,4.0453) {\fontsize{13}{13}\selectfont $q_L$ of $2$};
            \node[darkdarkgreen] at (3.5333,0.7333) {\fontsize{13}{13}\selectfont $q_L$ of $1$};
            \node[darkdarkgreen, align=center] at (3.6,3.6) {\fontsize{14}{14}\selectfont $q_B$ of $2$ \\ $\text{and}$ \\ \fontsize{14}{14}\selectfont $q_B$ of $1$};
        \end{tikzpicture}
        \caption{Allocation in Example~\ref{example:ex2}.}
        \label{fg:betterexample}
    \end{subfigure}
\end{figure}

As I show in Appendix~\ref{sec:verifyex3}, for all sufficiently small \(\varepsilon>0\), the quantities \(q_L\) and \(q_B\) can be chosen so that this three-option menu is optimal. To see why this is the case, note that, for small \(\varepsilon\), the value distribution illustrated in Figure~\ref{fig:example1} concentrates its mass on two kinds of agents. The first are \emph{common-need} agents, who are close to indifferent about which good they get but have high values for both. The second are \emph{specialized-need} agents, who tend to value one good substantially more than the other. However, their highest values tend to be lower than those of the common-need agents.

Here too, all agents with the same relative values $(\tfrac{v_1}{v_1+v_2},\tfrac{v_2}{v_1+v_2})$ receive the same allocation. However, agents whose relative values are close together choose the bundle and thus receive higher total allocations. Since these agents also tend to have higher \emph{absolute} values $(v_1,v_2)$, offering a bundle gives the designer an incentive-compatible way of directing more goods to agents in greater need. More generally, doing so can help the designer if relative and absolute values are statistically related. In such cases, she can sometimes proxy for high absolute values by offering more attractive options to agents with certain relative preferences.

Note, however, that this allocation is not Pareto-efficient among all allocation rules satisfying the supply constraint~\eqref{eq:S}.\footnote{For the optimal choice of quantities, the allocation is nevertheless Pareto-efficient within the class of allocation rules satisfying both~\eqref{eq:S} and~\eqref{eq:IC}.} Indeed, agents who receive the bundle could profitably trade among themselves so that types above and below the 45-degree line in Figure~\ref{fg:betterexample} each receive only the good they prefer. This illustrates the designer's tradeoff: bundling reduces allocative efficiency, but helps screen agents by the strength of their relative preferences, allowing scarce resources to be targeted toward those with greater need.

\section{Absolute and relative values}\label{sec:renormalization}

Motivated by the preceding examples, I now formally separate absolute and relative values. Since the designer cannot elicit absolute values of agents who share the same profile of relative values, we can, without loss, identify types with the latter. To that end, let $\Gamma$ be the $(N-1)$-simplex of relative-value profiles:
\[
\Gamma \;:=\; \bigl\{\theta\in\mathbb{R}^N_{+}:\ \sum \theta_i=1\bigr\}.
\]
Let \(V\sim F\) denote the random value vector of an agent, and define the \(\Gamma\)-valued random variable \(\Theta:=V/\sum_j V_j\).\footnote{We can exclude the type \(V=\mathbf 0\) without loss of generality.} We map all types that are identical up to scaling to the same renormalized type $\theta\in \Gamma$. Let $G$ denote the distribution of $\Theta$, which is the push-forward of $F$ under the map $v\mapsto v/\sum_j v_j$. Throughout, we impose the following assumption on \(G\):

\begin{assumption}\label{ass:regong}
$G$ admits a density $g$ that is strictly positive a.e. on $\Gamma$. Moreover, \(g\) and its first weak derivatives along \(\Gamma\) are square-integrable.
\end{assumption}

\begin{figure}[h!]
    \centering
    \begin{subfigure}[b]{0.45\linewidth}
        \centering

\tdplotsetmaincoords{70}{30}
\begin{tikzpicture}[scale=0.7, tdplot_main_coords]

  \def\Lx{4}
  \def\Ly{4.7}
  \def\Lz{4}

  \def\Ax{\Lx+0.7}
  \def\Ay{\Ly+0.7}
  \def\Az{\Lz+0.7}

  \coordinate (B) at (1.3,1.1,\Lz);
\coordinate (R) at (3,3.1,\Lz);
  \coordinate (G) at (\Lx,1.0,0.5);

  \fill[gray!30, fill opacity=0.18] (0,0,0) -- (\Lx,0,0) -- (\Lx,\Ly,0) -- (0,\Ly,0) -- cycle;
  \fill[gray!30, fill opacity=0.18] (0,0,\Lz) -- (\Lx,0,\Lz) -- (\Lx,\Ly,\Lz) -- (0,\Ly,\Lz) -- cycle;
  \fill[gray!30, fill opacity=0.18] (0,0,0) -- (\Lx,0,0) -- (\Lx,0,\Lz) -- (0,0,\Lz) -- cycle;
  \fill[gray!30, fill opacity=0.18] (0,\Ly,0) -- (\Lx,\Ly,0) -- (\Lx,\Ly,\Lz) -- (0,\Ly,\Lz) -- cycle;
  \fill[gray!30, fill opacity=0.18] (0,0,0) -- (0,\Ly,0) -- (0,\Ly,\Lz) -- (0,0,\Lz) -- cycle;
  \fill[gray!30, fill opacity=0.18] (\Lx,0,0) -- (\Lx,\Ly,0) -- (\Lx,\Ly,\Lz) -- (\Lx,0,\Lz) -- cycle;

  \draw[thin] (0,0,0) -- (\Lx,0,0);
  \draw[thin] (\Lx,0,0) -- (\Lx,\Ly,0);
  \draw[thin] (\Lx,\Ly,0) -- (0,\Ly,0);
  \draw[thin] (0,\Ly,0) -- (0,0,0);

  \draw[thin] (0,0,\Lz) -- (\Lx,0,\Lz);
  \draw[thin] (\Lx,0,\Lz) -- (\Lx,\Ly,\Lz);
  \draw[thin] (\Lx,\Ly,\Lz) -- (0,\Ly,\Lz);
  \draw[thin] (0,\Ly,\Lz) -- (0,0,\Lz);

  \draw[thin] (0,0,0) -- (0,0,\Lz);
  \draw[thin] (\Lx,0,0) -- (\Lx,0,\Lz);
  \draw[thin] (\Lx,\Ly,0) -- (\Lx,\Ly,\Lz);
  \draw[thin] (0,\Ly,0) -- (0,\Ly,\Lz);

  \draw[->, thick] (0,0,0) -- (\Ax,0,0) node[below left] {$v_1$};
\draw[->, thick] (0,0,0) -- (0,\Ay,0) node[below, yshift=-4pt] {$v_2$};
  \draw[->, thick] (0,0,0) -- (0,0,\Az) node[above] {$v_3$};

  \draw[blue, very thick] (0,0,0) -- (B);
  \draw[red, very thick] (0,0,0) -- (R);
  \draw[green!60!black, very thick] (0,0,0) -- (G);

  \fill[blue] (B) circle (2.2pt);
  \fill[red] (R) circle (2.2pt);
  \fill[green!60!black] (G) circle (2.2pt);

\end{tikzpicture}
    \end{subfigure}
    \begin{subfigure}[b]{0.45\linewidth}
        \centering
\begin{tikzpicture}[scale=2.3, line join=round, line cap=round]
  \coordinate (e1) at (90:1);
  \coordinate (e2) at (210:1);
  \coordinate (e3) at (330:1);

  \path[use as bounding box] (-1.25,-0.95) rectangle (1.25,1.25);

  \fill[gray!30, opacity=0.7] (e1) -- (e2) -- (e3) -- cycle;

  \draw (e1) -- (e2) -- (e3) -- cycle;

  \node[above]       at (e1) {$e_3$};
  \node[below left]  at (e2) {$e_1$};
  \node[below right] at (e3) {$e_2$};

  \fill[blue]            (0.00, 0.45) circle (0.05);
  \fill[red]             (-0.05,0.12) circle (0.05);
  \fill[green!60!black]  (-0.45,-0.18) circle (0.05);
\end{tikzpicture}
 \vspace{-0.3cm}
    \end{subfigure}
\caption{The original value space \(\mathcal V\) and the renormalized type space \(\Gamma\) in the case with three types of goods, $N=3$. The renormalization \(v\mapsto v/\sum v_i\) collapses each ray in \(\mathcal V\) to a single relative-value profile in \(\Gamma\).}
\label{fig:renormex}
\end{figure}

While the designer cannot elicit absolute values, they are still important for her objective. We therefore let $\lambda:\Gamma\to\mathbb R_+$ be the conditional expectation of total value given the renormalized type:
\[
\lambda(\Theta)
=
\mathbb E \Big[\sum_j V_j \,\Big|\, \Theta\Big].
\]
Intuitively, $\lambda(\theta)$ is the expected total value associated with agents
whose types $v$ were mapped to $\theta$.\footnote{Related renormalizations appear, for example, in \cite{weitzman1977price} and \cite{dworczak2021redistribution}. There, values are normalized by agents' marginal utility of money, so the analogue of \(\lambda\) is the expected marginal utility of money among agents with a given willingness to pay for a good.} Since $\mathcal V$ is bounded, $\lambda$ is also bounded on all of $\Gamma$. 

The renormalization has a clear economic interpretation. The type $\theta$ contains the minimal information needed to describe behavior and is the object that can be empirically identified from choices. By contrast, $\lambda(\theta)$ captures the expected \emph{scale} of values conditional on $\theta$, and therefore affects the problem only through the designer's objective. Economically, $\lambda$ encodes the designer's prior about how need (i.e., absolute value) varies across preference profiles, and is relevant only for the normative ranking of feasible allocations. We can use it to rewrite the designer's problem as follows:
\begin{problem}\label{prob:renormalized}
Choose an allocation rule $x:\Gamma\to \mathbb{R}^N_+$ to maximize weighted welfare
\begin{equation}\label{eq:Oprime}
\int_{\Gamma} \lambda(\theta)\,U(\theta)\,dG(\theta), \tag{O'}
\end{equation}
where $U(\theta)=x(\theta)\cdot \theta$, subject to:
\begin{equation}\label{eq:ICprime}
\theta\cdot x(\theta)\ \ge\ \theta\cdot x(\theta')
\ \ \  \text{for all } \  \ \theta,\theta'\in\Gamma,\tag{IC'}
\end{equation}
\begin{equation}\label{eq:supply'}
\int_{\Gamma} x(\theta)\,dG(\theta) \ \le\ s.\tag{S'}
\end{equation}
\end{problem}
Indeed, Problem \ref{prob:renormalized} is equivalent to the designer's original problem in the following sense:

\begin{lemma}\label{lem:renorm-wlog} For any feasible allocation rule $y:\mathcal V\to\mathbb R_+^N$, there exists $x:\Gamma\to \mathbb{R}_+^N$ satisfying 
\begin{equation}\label{eq:constructx}
x(\theta)\;=\;\mathbb{E}\left[\,y(V)\mid \Theta=\theta\,\right] \quad \text{ a.e.}
\end{equation}
that is feasible in Problem~\ref{prob:renormalized}. Moreover, welfare from \(y\) equals the renormalized welfare from \(x\):
\begin{equation}\label{eq:welfarerenromequal}
\int_{\mathcal V} v\cdot y(v)\,dF(v)
\;=\;
\int_\Gamma \lambda(\theta)\, \theta\cdot x(\theta) \,dG(\theta).
\end{equation}
Conversely, for any feasible $x$ in Problem~\ref{prob:renormalized}, the allocation rule given by \(y(\mathbf 0)=0\) and \(y(v):=x(v/\sum v_i)\) for \(v\neq\mathbf 0\) is feasible for the original problem, and the two allocation rules satisfy \eqref{eq:welfarerenromequal}.
\end{lemma}

\section{The pure option mechanism}\label{sec:pureoption}

As Example~\ref{example:ex3} shows, the optimal mechanism sometimes takes a particularly simple form: it offers one pure option for each good, and each agent chooses her favorite. In this section, I formally define this kind of mechanism.
\begin{definition}
An allocation rule $x:\Gamma\to\mathbb R_+^N$ is a \textbf{pure option mechanism} with a vector of \textbf{pure options} $q=(q_1,\dots,q_N)\in\mathbb R_{++}^N$ if for every $\theta\in\Gamma$,
\[
x(\theta)=q_i e_i
\ \  \text{for some } 
i\in\argmax_{j} \theta_j q_j,\quad 
\text{and} \qquad
\int_\Gamma x(\theta)\,dG(\theta)=s.
\]
\end{definition}

That is, agents face a menu with $N$ options,
\[
\left\{\text{$q_1$ of good $1$}\right\}, \ \ \ \left\{\text{$q_2$ of good $2$}\right\}, \ \ \ \dots, \ \ \ 
\left\{\text{$q_N$ of good $N$}\right\},
\]
where the quantities $q_i$ are chosen so that all the supply constraints bind when everyone chooses their favorite one. In the appendix, I show that such a market-clearing menu is unique:
\begin{fact}\label{fact:puremechs}
The pure option mechanism exists and is unique up to tie-breaking for a null set of types. 
\end{fact}

I now describe the allocation induced by this mechanism, that is, I characterize which types choose which pure option. To do so, I first introduce a
partial order comparing how strongly different types favor a particular good
relative to the other goods.

\begin{definition}\label{def:partialorder}
Take $\theta, \theta' \in \Gamma$ with $\theta_i,\theta_i'>0$. We say $\theta$ is closer to vertex $e_i$ than $\theta'$, denoted by $\theta \succ_i \theta'$, if for all $k\neq i$:
\[
\frac{\theta_k}{\theta_i} \leq \frac{\theta_k'}{\theta_i'}.
\]
\end{definition}

Intuitively, $\theta \succ_i \theta'$ means that $\theta$ values good $i$ relatively more than does $\theta'$, compared to every other good.

\begin{figure}[h!]
        \centering
  \begin{tikzpicture}[scale=2.5, line join=round, line cap=round]
    \coordinate (ei) at (90:1);
    \coordinate (ej) at (210:1);
    \coordinate (ek) at (330:1);

    \path[use as bounding box] (-1.25,-0.95) rectangle (1.25,1.25);

    \coordinate (th)  at (0,0.35);
    \coordinate (thp) at (0.20,-0.05);

    \path[name path=sideIEJ] (ei)--(ej);
    \path[name path=sideIEK] (ei)--(ek);

    \path[name path=rayJ] (ej)--($(ej)!3!(thp)$);
    \path[name path=rayK] (ek)--($(ek)!3!(thp)$);

    \path[name intersections={of=rayJ and sideIEK, by={Pj}}];
    \path[name intersections={of=rayK and sideIEJ, by={Pk}}];
    \fill[niceblue, opacity=0.35] (ei) -- (Pk) -- (thp) -- (Pj) -- cycle;

    \node[above]       at (ei) {$e_i$};
    \node[below left]  at (ej) {$e_j$};
    \node[below right] at (ek) {$e_k$};
    \draw[darkred, thick] (ej) -- (Pj);
    \draw[darkred, thick] (ek) -- (Pk);
    \fill (th)  circle (0.02) node[below] {$\theta$};
    \fill (thp) circle (0.02) node[below] {$\theta'$};
    \draw[thick] (ei)--(ej)--(ek)--cycle;
  \end{tikzpicture}
            \vspace*{-4.2mm} 
            \caption{Types in the shaded area are closer to $e_i$ than $\theta'$, i.e., $\theta \succ_i \theta'$.} \label{fig:preciorder}
    \end{figure}

\begin{fact}
Let $q=(q_1,\dots,q_N)$ be the unique vector of pure options in the pure option mechanism, and let $\theta^0\in\Gamma^\circ$ denote the type who is indifferent among all of them:
\[
\theta^0
\;:=\;
\left(
\frac{1/q_1}{\sum_{k=1}^N 1/q_k},\ 
\frac{1/q_2}{\sum_{k=1}^N 1/q_k},\ 
\dots,
\frac{1/q_N}{\sum_{k=1}^N 1/q_k}
\right).
\]
Define
\[
\Gamma_i:=\{\theta:\ \theta\succ_i\theta^0\}.
\]
Then every type $\theta\in\Gamma_i^\circ$ receives $x(\theta)=q_i e_i.$
\end{fact}

\begin{proof}
Let $q=(q_1,\dots,q_N)$ be the unique vector of pure options. By definition of $\theta^0$, we have $\theta_i^0 q_i=\theta_k^0 q_k$ for all $i,k$. Then ${\theta_k^0}/{\theta_i^0}={q_i}/{q_k}$ for all $i,k$. Now, fix $\theta\in\Gamma_i^\circ$. Since $\theta\succ_i\theta^0$, for every $k\neq i$, $\frac{\theta_k}{\theta_i}<\frac{\theta_k^0}{\theta_i^0}=\frac{q_i}{q_k}.$ Multiplying through by $\theta_i q_k$ gives $\theta_k q_k<\theta_i q_i$ for all $k\neq i$, so option $i$ uniquely maximizes $\theta_j q_j$ among all pure options. Thus $x(\theta)=q_i e_i$ in the pure option mechanism.
\end{proof}

That is, the mechanism's allocation can be described as follows. Let $q$ be the unique market-clearing vector of pure-option quantities. These quantities determine a unique type $\theta^0$ that is indifferent among all options. Up to tie-breaking on a null set, option~$i$ is then chosen exactly by the types who like good $i$ more than $\theta^0$, in the sense that they lie in the direction of vertex~$e_i$ from $\theta^0$ according to the partial order defined above. Denote this region by $\Gamma_i$. Since every type lies in the direction of some vertex from $\theta^0$, the regions $\{\Gamma_i\}_{i=1}^N$ then form a partition of $\Gamma$, up to the null set of types who are indifferent between two or more pure options (Figure~\ref{fg:CEEIsplit}).

\begin{figure}[h!]
              \vspace*{-4.2mm}
  \centering
\begin{tikzpicture}[scale=2.5, line join=round, line cap=round]
  \coordinate (ei) at (90:1);
  \coordinate (ej) at (210:1);
  \coordinate (ek) at (330:1);

  \path[use as bounding box] (-1.25,-0.95) rectangle (1.25,1.25);

  \coordinate (th)  at (0,0.35);
  \coordinate (thp) at (0.20,-0.05);

  \path[name path=sideIEJ] (ei)--(ej);
  \path[name path=sideIEK] (ei)--(ek);
  \path[name path=sideJEK] (ej)--(ek);

  \path[name path=rayJ] (ej)--($(ej)!3!(thp)$);
  \path[name path=rayK] (ek)--($(ek)!3!(thp)$);
  \path[name path=rayI] (ei)--($(ei)!3!(thp)$);

  \path[name intersections={of=rayJ and sideIEK, by={Pj}}];
  \path[name intersections={of=rayK and sideIEJ, by={Pk}}];
  \path[name intersections={of=rayI and sideJEK, by={Pi}}];

  \fill[darkgreen, opacity=0.35] (ei) -- (Pk) -- (thp) -- (Pj) -- cycle;
  \fill[mustard,   opacity=0.25] (ej) -- (Pi) -- (thp) -- (Pk) -- cycle;
  \fill[plum, opacity=0.20] (ek) -- (Pj) -- (thp) -- (Pi) -- cycle;

  \node[above]       at (ei) {$e_1$};
  \node[below left]  at (ej) {$e_2$};
  \node[below right] at (ek) {$e_3$};

  \draw[grey, thick] (ej) -- (Pj);
  \draw[grey, thick] (ek) -- (Pk);
  \draw[grey, thick] (ei) -- (Pi);

  \fill (thp) circle (0.03) node[below,left] {};
    \fill (0.24,-0.18) node[below,left] {$\theta^0$};

    \fill (-0.65,-0) node[below,left, font=\large, mustard] {$\Gamma_2$};
    \fill (0.6,0.6) node[below,left, font=\large, darkgreen] {$\Gamma_1$};
    \fill (1,-0.05) node[below,left, font=\large, plum] {$\Gamma_3$};

  \draw[thick] (ei)--(ej)--(ek)--cycle;
\end{tikzpicture}
            \vspace*{-4.2mm}
            \captionsetup{width=0.6\linewidth}
            \caption{Each region $\Gamma_i$ contains types who receive $q_i$ of good $i$ under the pure option mechanism.}
            \label{fg:CEEIsplit}
\end{figure}

Throughout the paper, the indifferent type \(\theta^0\) will serve as a
benchmark for the strength of agents' relative preferences. Thus, an agent's
``pickiness'' should generally be understood in terms of the strength of her
preferences among the market-clearing pure options \(q_i e_i\). When supplies
and value distributions are symmetric, all market-clearing pure options offer
the same quantity, so pickiness is captured directly by how unequal an agent's
per-unit values are across goods. In asymmetric environments, however, this
comparison must account for differences in supply and demand: if one good is
substantially scarcer or more desirable, its market-clearing option will be smaller, so an agent must place a correspondingly higher per-unit value on
that good to be indifferent between all pure options.

Although the pure option mechanism is most easily described through the menu introduced above, it also admits several alternative implementations. These connect my analysis to mechanisms studied in the literature and illustrate how it could be implemented in practice.

\begin{definition}\label{def:pure_implementations}
We define three implementations.
\begin{enumerate}[label=(\roman*)]
\item  A \textbf{competitive equilibrium with equal incomes} (CEEI) is a
vector of prices \(p\in\mathbb R_+^N\) and an allocation rule
\(x:\Gamma\to\mathbb R_+^N\) such that the supply constraints
\eqref{eq:supply'} bind for all goods and all types choose utility-maximizing
allocations subject to their unit budget constraint:
\[
  \text{for all }\theta \in \Gamma, \quad
  x(\theta)\in \argmax_{z \in \mathbb{R}^N_+}
  \left\{ \theta\cdot z : \ \  z\cdot p \leq 1  \right\}.
\]
Intuitively, a CEEI gives every agent one unit of artificial currency, posts
per-unit market-clearing prices \(p\), and lets everyone buy their favorite
bundle \(z\).

\item In the \textbf{representative endowment economy}, every agent is endowed
with an equal share \(s\) of the total supply. \textbf{A Walrasian equilibrium} of this economy is a pair
\((p,x)\), where \(p\in\mathbb R_+^N\) is a price vector and
\(x:\Gamma\to\mathbb R_+^N\) is an allocation rule, such that markets clear and
every type chooses a utility-maximizing bundle whose cost does not exceed the
market value of her endowment:
\[
x(\theta)\in
\argmax_{z\in\mathbb R_+^N}
\left\{
\theta\cdot z:\ p\cdot z\le p\cdot s
\right\}.
\]

\item A \textbf{choice-based lottery} is a game where each type \(\theta\in\Gamma\) chooses a
good. The supply of each good is then allocated uniformly at random among the
agents who selected it. Thus, each agent who chooses \(i\)
receives
\[
x(\theta)\;=\;\frac{s_i}{m_i}\,e_i,
\]
with \(s_i/m_i=+\infty\) if \(m_i=0\), where \(m_i\) is the mass of agents choosing good \(i\).
\end{enumerate}
\end{definition}

\begin{fact}\label{fact:pure_implementations}
The pure option mechanism can be implemented as a CEEI, as a Walrasian
equilibrium of the representative endowment economy, and, under the unit-demand
interpretation, as a pure-strategy Nash equilibrium of the choice-based lottery.
Conversely, each of these implementations induces the pure option mechanism, up to tie-breaking for a null set of types.
\end{fact}

Intuitively, the CEEI and representative-endowment implementations give all agents identical budgets to spend on goods: in the former, this budget is assigned directly, while in the latter it is the market value of their endowment \(s\). Each agent then spends her whole budget on the good with the highest ``bang per buck,'' \(\theta_i/p_i\), and so both implementations induce pure allocations. Since these are also incentive compatible and market clearing, Fact~\ref{fact:puremechs} implies that they coincide with the pure option mechanism. The choice-based lottery is similar: any pure equilibrium induces a pure, incentive-compatible, market-clearing allocation and therefore implements the pure option mechanism by the same fact.

\section{Optimality of the pure option mechanism}\label{sec:CEEI}

I now characterize when the pure option mechanism is welfare-maximizing; I focus on presenting the result and its high-level intuitions; the technical details underlying the proof are discussed in Section~\ref{sec:proofdetails}. To state the result, I first introduce several objects.

\paragraph{Supply costs.} These values capture the marginal welfare gain of relaxing each supply constraint at the pure option allocation. To construct them, define
\[
M_i:=\int_{\Gamma_i}  g(\theta)\, d\theta, \quad A_i:=\int_{\Gamma_i} \theta_i  g(\theta)\lambda(\theta)\, d\theta,
\]
so that $M_i$ is the mass of agents choosing option $i$ and $A_i$ is the designer's total value of giving each of them a unit of good $i$. For $i\neq j$, define
\[
T_{ij}:=\int_{\Gamma_i\cap\Gamma_j}  g(\theta)\,\theta_i\, d\sigma(\theta) \Big/ \sqrt{q_i^2+q_j^2-\tfrac1N(q_i-q_j)^2},
\]
where $d\sigma$ denotes $(N-2)$-dimensional Hausdorff measure on $\Gamma_i\cap\Gamma_j$. Note that $M_i,A_i>0$ and $T_{ij}\geq0$ for every $i\neq j$. We use these objects to construct the following matrix $J\in\mathbb{R}^{N\times N}$:
\[
J_{ii}=M_i+q_i\sum_{j\neq i}T_{ij},
\qquad
J_{ij}=-\,q_j\,T_{ij}\ \ (i\neq j).
\]

\begin{definition}\label{def:supply_costs}
The vector of \textbf{supply costs} $c = (c_1, c_2, \dots, c_N)$ is given by
\[
c \;=\; J^{-1}A, \quad \text{where} \quad A:=\begin{pmatrix}A_1 & \cdots & A_N\end{pmatrix}^{\top}.
\]
\end{definition}
\begin{fact}\label{fact:supply_cost_properties}
The supply costs $c=J^{-1}A$ are well-defined and strictly positive.
\end{fact}
I discuss the intuition for this construction in Section~\ref{sec:proofdetails}. Note, however, that $c$ is pinned down entirely by the primitives $(\lambda,g,s)$: directly through the distribution $g$ and the weight function $\lambda$, and indirectly through the market-clearing pure options $q$, which are themselves determined by $g$ and $s$.

\paragraph{Chosen-good component.} For each type~\(\theta\), define its chosen-good component \(\theta_*(\theta)\) as the coordinate corresponding to its preferred pure option:
\[
\theta_*(\theta):=\theta_i
\quad\text{for }\theta\in\Gamma_i^\circ.
\]
On indifference boundaries, use any measurable tie-breaking rule.

\paragraph{Rent measure.} Define the following signed measure on the type space $\Gamma$:
\begin{equation}\label{eq:signedmeas}
\mu(A) = \int_{A\cap \Gamma} \theta_*\ \Bigl[
(\lambda-\textstyle\sum_j c_j) \, g
+
\operatorname{div}\Bigl(\bigl( c-\theta \textstyle\sum_j c_j \bigr)\, g \Bigr)
\Bigr]\,d\theta
-\int_{A\cap \partial\Gamma}\theta_* \,
\bigl( c-\theta \textstyle\sum_j c_j \bigr)\, g \cdot\nu \,d\sigma,
\end{equation}

where $\nu$ is the outward unit conormal to $\partial\Gamma$ in the hyperplane containing $\Gamma$; the divergence is also taken within this hyperplane. 

Intuitively, $\mu$ captures the marginal welfare gain from increasing the size of a type's preferred option by a unit. It incorporates several effects: the direct utility to the recipient weighted by $\lambda$, the supply cost of the additional resources required, and the effects arising from tightening local incentive-compatibility constraints. Thus, $\mu$ places negative mass on types whose allocations the designer would want to \emph{decrease} after accounting for all these effects, and positive mass on those whose allocations she would want to \emph{increase}. In this sense, the rent measure plays a role analogous to virtual values in one-dimensional screening.

\paragraph{Monotone transport plan.} The optimality conditions will be stated in terms of the shape of this rent measure $\mu$. To formalize them, we use the following concept:

\begin{definition}\label{def:monotone-transport}
Let $\rho,\tau$ be measures on some $\Omega \subset \mathbb{R}^N$ with
$\rho(\Omega)=\tau(\Omega)$, and let $\succeq$ be a partial order on $\Omega$.
A $\succeq$-\textbf{monotone transport plan} from $\rho$ to $\tau$ is a finite
nonnegative Borel measure $\pi$ on $\Omega\times\Omega$, supported on
$\{(x,y):x\preceq y\}$, with first marginal $\rho$ and second marginal $\tau$;
that is,
\[
\pi(A\times \Omega)=\rho(A),\qquad
\pi(\Omega\times A)=\tau(A)
\quad \text{for all Borel }A\subseteq \Omega.
\]
\end{definition}

Intuitively, a transport plan specifies how one measure can be moved onto another, while allowing mass to move only upward according to the partial order \(\succeq\).

We can now state the main result of this section:

\begin{theorem}
\label{thm:rent-measure-characterization}
The pure option mechanism is optimal if and only if there exists a tuple
\((\tilde\mu_1,\dots,\tilde\mu_N)\) of finite signed Borel measures, with
\(\tilde\mu_i\) defined on \(\Gamma_i\), such that the following two
conditions hold.

\begin{enumerate}
\item \textbf{Preprocessing.}
For every continuous, convex \(U:\Gamma\to\mathbb R_+\),
\[
\int_{\Gamma}\frac{U(\theta)}{\theta_*(\theta)}\,d\mu(\theta)
\le
\sum_{i=1}^N
\int_{\Gamma_i}\frac{U(\theta)}{\theta_i}\,d\tilde\mu_i(\theta).
\]

\item \textbf{Regionwise transport.}
For each \(i\), let \(\tilde\mu_i^+\) and \(\tilde\mu_i^-\) denote the
positive and negative parts of \(\tilde\mu_i\). Then there exists a
\(\succ_i\)-monotone transport plan from \(\tilde\mu_i^-\) to
\(\tilde\mu_i^+\).
\end{enumerate}
\end{theorem}

Thus, the pure option mechanism is optimal if and only if the rent measure
can be transformed into a collection of region-specific auxiliary measures
whose positive and negative parts are suitably positioned relative to each
other. The preprocessing step corresponds to ironing in one-dimensional
mechanism design (see \cite{rochet1998ironing} and \cite{daskalakis} for
related discussions). Once the rent measure has been preprocessed into
\((\tilde\mu_i)_{i=1}^N\), the second condition requires that, on each region
\(\Gamma_i\), the negative part of \(\tilde\mu_i\) can be transported onto
its positive part by moving mass only in the direction toward vertex \(e_i\),
in the sense of the partial order \(\succ_i\).

\begin{figure}[h!]
  \centering
\begin{tikzpicture}[scale=2.5, line join=round, line cap=round]

  \coordinate (ei) at (90:1);
  \coordinate (ej) at (210:1);
  \coordinate (ek) at (330:1);

  \path[use as bounding box] (-1.25,-0.95) rectangle (1.25,1.25);

  \coordinate (thp) at (0.20,-0.05); 

  \path[name path=sideIEJ] (ei)--(ej);
  \path[name path=sideIEK] (ei)--(ek);
  \path[name path=sideJEK] (ej)--(ek);

  \path[name path=rayJ] (ej)--($(ej)!3!(thp)$);
  \path[name path=rayI] (ei)--($(ei)!3!(thp)$);
  \path[name path=rayK] (ek)--($(ek)!3!(thp)$);

  \path[name intersections={of=rayJ and sideIEK, by={Pj}}];
  \path[name intersections={of=rayI and sideJEK, by={Pi}}];
  \path[name intersections={of=rayK and sideIEJ, by={Pk}}];

  \def\blobplot{
    plot[smooth cycle, tension=0.95] coordinates {
      (0.33, 0.08)
      (0.36,-0.10)
      (0.16,-0.30)
      (-0.10,-0.32)
      (-0.38,-0.22)
      (-0.66, 0.10)
      (-0.48, 0.36)
      (-0.14, 0.38)
      (0.18, 0.32)
    }
  }

  \begin{scope}

    \begin{scope}
      \clip (ei)--(Pk)--(thp)--(Pj)--cycle;

      \path[draw=none, fill=darkgreen, opacity=0.22, even odd rule]
        (-2,-2) rectangle (2,2)
        \blobplot;

      \path[draw=none, fill=darkgreen, opacity=0.50] \blobplot;
    \end{scope}

    \begin{scope}
      \clip (ej)--(Pi)--(thp)--(Pk)--cycle;

      \path[draw=none, fill=mustard, opacity=0.20, even odd rule]
        (-2,-2) rectangle (2,2)
        \blobplot;

      \path[draw=none, fill=mustard, opacity=0.48] \blobplot;
    \end{scope}

    \begin{scope}
      \clip (ek)--(Pj)--(thp)--(Pi)--cycle;

      \path[draw=none, fill=plum, opacity=0.18, even odd rule]
        (-2,-2) rectangle (2,2)
        \blobplot;

      \path[draw=none, fill=plum, opacity=0.46] \blobplot;
    \end{scope}

  \end{scope}

  \draw[grey, thick] (Pi) -- (thp);
  \draw[grey, thick] (Pj) -- (thp);
  \draw[grey, thick] (Pk) -- (thp);

  \node[above]       at (ei) {$e_1$};
  \node[below left]  at (ej) {$e_2$};
  \node[below right] at (ek) {$e_3$};

  \fill (thp) circle (0.03);

  \draw[thick] (ei)--(ej)--(ek)--cycle;
      \fill (0.24,-0.18) node[below,left] {$\theta^0$};

    \fill (-0.65,-0) node[below,left, font=\large, mustard] {$\tilde \mu_2$};
    \fill (0.6,0.6) node[below,left, font=\large, darkgreen] {$\tilde \mu_1$};
    \fill (1,-0.05) node[below,left, font=\large, plum] {$\tilde \mu_3$};

\end{tikzpicture}

            \vspace*{-4.2mm} 
            \captionsetup{width=0.7\linewidth}
            \caption{An example where the regionwise transport condition is satisfied. Within each region $\Gamma_i$, the negative part of $\tilde\mu_i$ (darker shading) can be transported onto the positive part (lighter shading) by moving mass toward vertex $e_i$.}
            \label{fg:measureexample}
\end{figure}

This happens when the types the designer would like to reward more are precisely the more selective ones---those with stronger relative preferences for their chosen pure option. This foreshadows the sufficient conditions in the next subsection: when higher total values $\lambda(\theta)$ are associated with more extreme relative preferences, the transport condition is easier to satisfy, and the pure option mechanism is more likely to be optimal.

\subsection{Sufficient conditions for optimality}\label{sec:verifiable}

We now derive a sufficient condition from Theorem~\ref{thm:rent-measure-characterization} for the case of equal supplies and exchangeable distributions. To state it, we first introduce a stochastic monotonicity notion.

\begin{definition}\label{def:stoch_mono_cond}
Let $X$ be an $\mathcal X$-valued random variable and $Y$ be a real-valued random variable. Let $\succeq$ be a partial order on $\mathcal X$.
Fix an event $E$ with $\mathbb P(E)>0$. For any $t$ with $\mathbb P(Y\ge t,\ E)>0$, let $\mathcal L(X\mid Y\ge t,\ E)$ denote the conditional law of $X$ given $\{Y\ge t\}\cap E$.

Then $X$ is $\succeq$-\textbf{stochastically decreasing in $Y$ conditional on $E$} if for all
such $t,t'$ for which $t>t'$,
\[
\mathcal L(X\mid Y\ge t',\ E)\ \quad \text{$\succeq$-stochastically dominates}\quad
\mathcal L(X\mid Y\ge t,\ E).
\]
\end{definition}
That is, conditional on \(E\), higher values of \(Y\) are associated with conditional distributions of \(X\) that are lower in the sense of first-order stochastic dominance with respect to \(\succeq\).

\begin{proposition}\label{prop:exchcondition}
Assume \(s_1=...=s_N\) and let \(g\) and \(\lambda\) be exchangeable. Then the pure option mechanism is optimal if, for every \(i\), the random vector
\[
\left(\frac{\Theta_1}{\Theta_i},\dots,\frac{\Theta_N}{\Theta_i}\right)
\]
is \(\ge\)-stochastically decreasing in \(\lambda(\Theta)\Theta_i\) conditional on $\Theta_i\ge \Theta_j$ for all $j\neq i$.
\end{proposition}

Intuitively, the condition says that the pure option mechanism is optimal when agents with higher values for their favorite good tend to be \emph{more selective}: conditional on good $i$ being their favorite, higher realizations of the weighted value $\lambda(\Theta)\Theta_i$ are associated with smaller ratios $(\Theta_j/\Theta_i)_{j\neq i}$ in the sense of stochastic dominance. This echoes the intuition from Example~\ref{example:ex2}. There, distorting the pure option menu by introducing mixtures was beneficial precisely because \emph{less} selective agents had higher cardinal values. Under the condition in Proposition~\ref{prop:exchcondition}, the opposite is true, and such distortions are counterproductive. Note that symmetry makes the notion of selectivity especially direct: all market-clearing pure options offer the same quantity, so relative preferences are captured simply by the ratios of values.

\begin{remark}
The stochastic monotonicity condition in Proposition~\ref{prop:exchcondition}
is related to those in \cite{haghpanah2021pure} and
\cite{yang2025costly}, who study a monopolist with access to what is essentially a wasteful screening instrument. Under a similar condition, this instrument is less costly to high-value agents, and thus can only serve to prevent misreporting from low- to high-value types. At the optimum, however, incentive constraints bind in the opposite direction: high-value types want to imitate low-value types, and thus the instrument is counterproductive and should not be used. While the proof technique used to show my result is different, the condition arises for a similar economic reason: bundling, as opposed to letting agents choose which good they want, acts as a wasteful screening instrument, and the designer is concerned about low-value types pretending to have high values. Under the right stochastic monotonicity condition, however, high-value agents tend to be pickier, so bundling is less attractive to them than to the agents who try to imitate them. Consequently, this instrument screens ``in the wrong direction''.
\end{remark}

Proposition~\ref{prop:exchcondition} also yields the following corollary, which provides a closed-form condition that can be verified directly from the marginal distribution of values:

\begin{corollary}\label{cor:iid_reverse_hazard_pure}
Assume \(s_1=...=s_N\), and suppose that the unnormalized values \(V_1,\dots,V_N\) are i.i.d.\ with common distribution
\(F_M\). Suppose $F_M$ is supported on \([0,\bar v]\) and admits a density
\(f_M\) that is positive on \((0,\bar v)\). If
\begin{equation}\label{eq:reverse_hazard_elasticity}
x\,\frac{f_M(x)}{F_M(x)}
\quad\text{is non-increasing on }(0,\bar v),
\end{equation}
then the pure option mechanism is optimal.
\end{corollary}

The condition is satisfied for several standard distribution families:
\begin{example}\label{ex:iid_reverse_hazard_examples}
Assume \(s_1=...=s_N\), and suppose \(V_1,\dots,V_N\) are i.i.d. on
\([0,\bar v]\). Then the pure option mechanism is optimal if the common
distribution \(F_M\) is any of the following:
\begin{enumerate}
\item the uniform distribution, \(F_M(x)=\frac{x}{\bar v}\);
\item a power distribution, \(F_M(x)=\left(\frac{x}{\bar v}\right)^\alpha\) with \(\alpha> \frac32\);
\item a truncated exponential distribution,
\(F_M(x)=\frac{1-\exp(-\beta x)}{1-\exp(-\beta \bar v)}\) with \(\beta>0\).
\end{enumerate}
\end{example}

\subsection{Stability of the optimality of the pure option mechanism}\label{sub:whyoftenopt}

The preceding results show that the pure option mechanism is optimal for a
non-trivial set of primitives. In the symmetric i.i.d.\ case, for example, the
reverse-hazard-rate condition in
Corollary~\ref{cor:iid_reverse_hazard_pure} is relatively permissive. While it
is harder to state equally simple sufficient conditions in the general
asymmetric case, the next result shows that the optimality of the pure option mechanism is not a
knife-edge implication of symmetry. To make this statement precise, I first define a notion of stability.

\begin{definition}\label{def:stability}
Fix baseline primitives $(\lambda^*,g^*,s^*)$. A property of the
primitives is \textbf{stable at \((\lambda^*,g^*,s^*)\)} if there exists
\(\varepsilon>0\) such that the property holds for every admissible primitive
vector \((\lambda,g,s)\) satisfying
\[
\lambda=e^\ell\lambda^*,
\qquad
g=\frac{e^h g^*}{\int_\Gamma e^h g^*\,dm},
\qquad\text{and}\qquad
\|\ell\|_{C^1(\Gamma)}
+
\|h\|_{C^2(\Gamma)}
+
\|s-s^*\|
<\varepsilon.
\]
\end{definition}
The log-scale parametrization preserves the positivity of \(\lambda\) and \(g\), while the normalization ensures that \(g\) remains a
probability density.

I now show that if the optimality conditions underlying
Theorem~\ref{thm:rent-measure-characterization} hold with appropriately defined \emph{strict slack}, then the optimality of the pure option mechanism is stable. To formulate this notion of slack, we introduce the following definition:

\begin{definition}\label{def:upper-set-slack}
Fix primitives \((\lambda,g,s)\), and let \(\Gamma_i\) and \(\mu\) denote
the associated pure-option regions and rent measure. A set
\(C\subseteq\Gamma_i\) is an \(\succ_i\)-\textbf{upper set} if $\theta\in C$, $\theta\prec_i\theta'$, and $\theta'\in\Gamma_i$ imply that $\theta'\in C$.

The primitives \((\lambda,g,s)\) have \textbf{upper-set slack} if there
exists \(\eta>0\) such that, for every \(i\) and every closed
\(\succ_i\)-upper set \(C\subseteq\Gamma_i\),
\begin{equation}\label{eq:BUS}
\mu(C)
\ge
\eta
\min\left\{
\alpha(C),
\alpha\bigl(\Gamma_i\setminus C\bigr)
\right\},
\tag{BUS}
\end{equation}
where \(\alpha:=m+\sigma\), \(m\) is the
\((N-1)\)-dimensional Lebesgue measure on \(\Gamma\), and \(\sigma\) is the
\((N-2)\)-dimensional Hausdorff measure on \(\partial\Gamma\).
\end{definition}

This upper-set formulation arises from Strassen's theorem, which links the
existence of monotone transports to positivity on all upper sets. We then get
the following result:

\begin{proposition}\label{prop:neighbourhood-perturb}
Fix baseline primitives \((\lambda^*,g^*,s^*)\), where \(g^*\) satisfies
Assumption~\ref{ass:regong}, \(\lambda^*\in C^1(\Gamma)\),
\(g^*\in C^2(\Gamma)\), and \(s^*\gg0\). If
\((\lambda^*,g^*,s^*)\) has upper-set slack, then the optimality of the pure
option mechanism is stable at \((\lambda^*,g^*,s^*)\).
\end{proposition}

Note that, as the primitives are
perturbed, the market-clearing pure options $q_i$ and the regions \(\Gamma_i\)
generally change. The result then says that the corresponding pure option mechanism,
with its perturbed quantities and regions, is optimal.

It is also worth noting that the upper-set slack condition is not excessively demanding. For example, it holds in the uniform symmetric environment:
\begin{example}\label{ex:uniform-cube-BUS}
Assume \(s_1=...=s_N\), and suppose unnormalized values $V$ are distributed uniformly on $[0,1]^N$. Then the associated primitives $(\lambda,g,s)$ have upper-set slack.
\end{example}

\subsection{Why is the pure option mechanism often optimal?}

The preceding results show that the pure option mechanism is optimal for a
broad class of primitives and that this optimality is robust to perturbations. This raises the question: why is such a simple mechanism often exactly optimal? To see this, note first that
it is essentially the only mechanism that is both Pareto-efficient, given the
available supply, and incentive compatible.

\begin{proposition}\label{prop:PE_impl_implies_PO}
Suppose the allocation rule \(x\) is Pareto-efficient subject to the supply constraint \eqref{eq:supply'}, meaning that there is no allocation rule \(\tilde x\) satisfying \eqref{eq:supply'} such that
\[
\theta\cdot\tilde x(\theta)
\ge
\theta\cdot x(\theta)
\quad\text{for all }\theta,
\]
with strict inequality for a positive mass of types. If \(x\) also satisfies \eqref{eq:ICprime}, then it coincides almost everywhere with the pure option mechanism. That is, if \(q\) is the unique market-clearing vector of pure options, then
\[
x(\theta)=q_i e_i
\qquad
\text{for almost every }\theta\in\Gamma_i.
\]
\end{proposition}
To see why, first note that Pareto efficiency requires exhausting supply: any leftover amount could be distributed uniformly without violating incentives. Second, assigning bundles to a positive mass of types cannot be Pareto-efficient, since their recipients could trade so that each good goes to those who value it relatively more. Any Pareto-efficient allocation must therefore be pure almost everywhere, which pins it down as the pure option mechanism.

Any welfare improvement over the pure option mechanism must consequently involve Pareto-inefficient distortions. However, all such distortions (other than discarding supply) rely on mixed allocations and thus reward agents who are \emph{less picky}. If these agents tend to have \emph{lower} values, directing rents toward them is counterproductive, and the designer does better by simply adhering to the Pareto-efficient outcome. The pure option mechanism is therefore optimal for an entire ``direction'' of Pareto weights \(\lambda\), as illustrated in Figure~\ref{fig:direction}.

\begin{figure}[h!]
    \centering
    \begin{subfigure}[b]{0.45\linewidth}
        \centering
\begin{tikzpicture}[scale=4.2]
  \coordinate (O) at (0,0);
  \coordinate (B) at (0,0.3);
  \coordinate (K) at (0.5,0.5);
  \coordinate (D) at (0.7,0);

  \path[fill=niceblue, fill opacity=0.4]
    (O) -- (B) -- (K)
    to[out=-45,in=90] (D)
    -- cycle;

  \draw[orange, line width=1.2pt] (0,1) -- (1,0);
  \node[text=orange, align=center, font=\small] at (0.28,1)
    {Pareto\\frontier};

  \draw[darkred, line width=1.2pt]
    (0,0.7) -- (1.02,0.29);

  \fill[darkdarkgreen] (K) circle(0.7pt);

  \node[text=darkgreen, align=center, font=\small] at (0.75,0.59)
    {Pure option\\allocation};

  \node[align=center, font=\small] at (0.3,0.17)
    {Feasible\\region};

  \draw[->] (O) -- (1.1,0) node[right] {$U_{\text{flexible}}$};
  \draw[->] (O) -- (0,1.1) node[above] {$U_{\text{picky}}$};
\end{tikzpicture}
        \caption{Pareto weights skew towards picky agents.}
        \label{fig:dirofintegr}
    \end{subfigure}
    \hfill
    \begin{subfigure}[b]{0.45\linewidth}
        \centering
\begin{tikzpicture}[scale=4.2]
  \coordinate (O) at (0,0);
  \coordinate (B) at (0,0.3);
  \coordinate (K) at (0.5,0.5);
  \coordinate (D) at (0.7,0);

  \path[fill=niceblue, fill opacity=0.4]
    (O) -- (B) -- (K)
    to[out=-45,in=90] (D)
    -- cycle;

  \draw[orange, line width=1.2pt] (0,1) -- (1,0);
  \node[text=orange, align=center, font=\small] at (0.26,1)
    {Pareto\\frontier};

  \draw[darkred, line width=1.2pt]
    (0.42,1.05) -- (0.75,0);

  \fill[darkred] (0.684,0.2) circle(0.7pt);
  \fill[darkdarkgreen] (K) circle(0.7pt);

  \node[text=darkgreen, align=center, font=\small] at (0.82,0.59)
    {Pure option\\allocation};

  \node[align=center, font=\small] at (0.3,0.17)
    {Feasible\\region};

  \draw[->] (O) -- (1.1,0) node[right] {$U_{\text{flexible}}$};
  \draw[->] (O) -- (0,1.1) node[above] {$U_{\text{picky}}$};
\end{tikzpicture}
        \caption{Pareto weights skew towards flexible agents.}
        \label{fig:bdry}
    \end{subfigure}
    \caption{Moving away from the pure option mechanism can reward flexible agents, but not picky ones. The designer therefore prefers the Pareto-efficient allocation when her Pareto weights favor picky agents, but may distort it when they favor flexible agents.}
    \label{fig:direction}
\end{figure}

More broadly, the pure option mechanism is often optimal because the designer has very limited instruments for reallocating rents across agents; this distinguishes my model from environments with transfers. In the optimal taxation model of \cite{mirrlees1971exploration}, for example, a Pareto-efficient allocation is optimal only under knife-edge social preferences because the designer can generally distort it to direct rents in her preferred direction: if she wants to help low-productivity agents, she taxes labor and redistributes; if she wants to help high-productivity agents, she imposes lump-sum taxes and subsidizes labor. Here, the designer has no such freedom. Without transfers, she can only redistribute indirectly, by offering bundles that target flexible types, and has no analogous device for redistributing rents in the opposite direction.

\section{Two kinds of goods}\label{sec:2good}
I now characterize the optimal mechanism with two kinds of goods. Under the reparametrization from Section~\ref{sec:renormalization}, each type \(\theta=(1-\theta_2,\theta_2)\) is fully described by the scalar \(\theta_2\), making the type space one-dimensional. This yields a particularly simple representation of incentive-compatible allocation rules: every allocation rule satisfying \eqref{eq:ICprime} can be decomposed into a nonnegative combination of binary threshold rules \(x^z\), indexed by \(z\in[0,1]\):
\begin{equation}\label{eq:thresholdrulee}
x^z(1-t,t)=
\begin{cases}
(z,0), & t<z,\\[3pt]
(0,1-z), & t\geq z.
\end{cases}
\end{equation}
That is, \(x^z\) is induced by a two-option menu offering either \(z\) units of good \(1\) or \(1-z\) units of good \(2\); the cutoff type \(\theta=(1-z,z)\) is indifferent between these options. Define
\[
\varphi(z)
:=
\left( z\,\mathbb P(\Theta_2\le z), \
(1-z)\,\mathbb P(\Theta_2\ge z) \right),
\quad \text{and} \quad
w(z)
:=
\mathbb E\bigl[\lambda(\Theta)\max\{z\Theta_1,(1-z)\Theta_2\}\bigr].
\]
Then \(\varphi(z)\) is the supply vector used by the threshold rule \(x^z\), and \(w(z)\) is the welfare it generates. We then get the following characterization of the optimum:
\begin{theorem}\label{th:2good_asym_menu}
For any $a,b\in (0,1)$ such that \(a<\theta_2^0<b\), let $m_a,m_b>0$ be the unique weights satisfying
\begin{equation}\label{eq:asym_supply_clearing_weights}
m_a\varphi(a)+m_b\varphi(b)=s.
\end{equation}
The pure option mechanism is optimal if and only if, for every such $a,b$, we have:
\begin{equation}\label{eq:asym_no_profitable_bundle}
m_a w(a)+m_b w(b)
\le
\frac{s_1}{\varphi_1(\theta_2^0)}w(\theta_2^0).
\end{equation}
Otherwise, there is an optimal mechanism that offers two pure options and one bundle:
\begin{equation}\label{eq:menumixed}
\left\{\text{$r_1$ of good $1$}\right\}, \qquad
\left\{\text{$r_2$ of good $2$}\right\}, \qquad
\left\{\text{$b_1$ of good $1$ and $b_2$ of good $2$}\right\},
\end{equation}
where the quantities are chosen so that supply constraints bind when agents choose their favorite option.
\end{theorem}
Thus, the optimal mechanism can take one of two forms: it is either the pure
option mechanism, or it offers two pure options together with a mixed bundle,
as in Example~\ref{example:ex2}. 

While the proof is in the appendix, I explain the core logic behind the result. The threshold-rule decomposition turns the designer's problem into a linear
program over nonnegative measures on \(z\in[0,1]\). Each cutoff uses supplies \(\varphi(z)\) and generates welfare \(w(z)\). Since there are only two supply constraints, there always exists a solution that places weight on at most two cutoffs, say \(a\le b\); the resulting allocation rule then has at most three allocation levels: a pure good-\(1\) option, a bundle for the types between the cutoffs, and a
pure good-\(2\) option. 

The pure option mechanism corresponds to placing all the weight on a single cutoff
\(z=\theta_2^0\), and generates welfare equal to the right-hand side of
\eqref{eq:asym_no_profitable_bundle}. As explained above, there always exists an optimal menu that exhausts both supplies and is either the pure option mechanism or of the form in \eqref{eq:menumixed}. Thus, to check whether the pure option mechanism is optimal, it is enough to compare it with supply-exhausting menus of the latter form. Each such menu is pinned down by a pair of types \((1-a,a)\) and \((1-b,b)\) with \(a<\theta_2^0<b\), as in Figure~\ref{fg:2dbundle}; these types are indifferent between the bundle and one of the pure options. In the threshold-rule representation, the menu combines the binary threshold rules \(x^a\) and \(x^b\) with the weights \(m_a\) and \(m_b\) defined in \eqref{eq:asym_supply_clearing_weights}, and generates welfare \(m_aw(a)+m_bw(b)\).\footnote{The ordering \(a<\theta_2^0<b\) is equivalent to the existence of such weights. Intuitively, one threshold rule uses more of good~\(1\) than the pure option mechanism, while the other uses relatively more of good~\(2\), allowing their combination to exhaust both supplies.} Condition~\eqref{eq:asym_no_profitable_bundle} therefore ensures that no such menu improves upon the pure option mechanism.

\begin{figure}[h!]
        \centering
        \begin{subfigure}[b]{0.4\linewidth}
        \centering
\begin{tikzpicture}[scale=0.8]

  \def\L{6.2}     
  \def\xm{2*\L/5} 
  \def\hA{2.03}   
  \def\hB{2.56}   

  \fill[plum, opacity=0.35]
    (0,0) rectangle ({\xm},{\hA});

  \fill[darkdarkgreen!70, opacity=0.45]
    ({\xm},0) rectangle (\L,{\hB});

  \draw[dashed, thick] ({\xm},0) -- ({\xm},{\hA});

  \draw[thick] (0,0) -- (\L,0);

  \draw[thick] (0,-0.18) -- (0,0.18);
  \draw[thick] (\L,-0.18) -- (\L,0.18);

  \draw[thick] ({\xm},-0.15) -- ({\xm},0.15);

  \node[below] at (0,-0.18) {$(1,0)$};
  \node[below] at (\L,-0.18) {$(0,1)$};
  \node[below] at ({\xm},-0.15) {$(\theta_1^0,\theta_2^0)$};

\end{tikzpicture}
            \caption{Figure \ref{fig:2dpure}: The pure option mechanism}
            \label{fig:2dpure}
        \end{subfigure}
        \hspace{1cm}
        \begin{subfigure}[b]{0.4\linewidth}
            \centering
\begin{tikzpicture}[scale=0.8]

  \def\L{6.2}         
  \def\xm{2*\L/5}     
  \def\xa{\L/5}       
  \def\xb{7*\L/10}    
  \def\hA{1.70}       
  \def\hB{2.35}       
  \def\hM{2.95}       

  \fill[plum, opacity=0.35]
    (0,0) rectangle ({\xa},{\hA});

  \fill[plum, opacity=0.35]
    ({\xa},0) rectangle ({\xb},{\hM/3});

  \fill[darkdarkgreen!70, opacity=0.45]
    ({\xa},{\hM/3}) rectangle ({\xb},{\hM});

  \fill[darkdarkgreen!70, opacity=0.45]
    ({\xb},0) rectangle (\L,{\hB});

  \draw[dashed, thick] ({\xa},0) -- ({\xa},{\hA});
  \draw[dashed, thick] ({\xb},0) -- ({\xb},{\hB});

  \draw[thick] (0,0) -- (\L,0);

  \draw[thick] (0,-0.18) -- (0,0.18);
  \draw[thick] (\L,-0.18) -- (\L,0.18);

  \draw[thick] ({\xm},-0.15) -- ({\xm},0.15);

  \draw[thick] ({\xa},-0.15) -- ({\xa},0.15);
  \draw[thick] ({\xb},-0.15) -- ({\xb},0.15);

  \node[below] at ({\xa},-0.15) {$(1-a,a)$};
  \node[below] at ({\xb},-0.15) {$(1-b,b)$};

\end{tikzpicture}

            \caption{Figure \ref{fg:2dbundle}: A bundling mechanism }
            \label{fg:2dbundle}
        \end{subfigure}
    \end{figure}

When the two cutoffs satisfy \(a<1/2<b\), the mixed bundle delivers more in
total than either pure option, as in Example~\ref{example:ex2}. This is the
natural case when imbalances in the supply and demand for the two goods are
limited. More substantial imbalances can shift the market-clearing indifferent
type \(\theta_2^0\) far enough from \(1/2\) that both cutoffs lie on the same
side of \(1/2\). The bundle is then larger than one pure option but smaller
than the other. In either case, the cutoffs lie on opposite sides of
\(\theta_2^0\), so the bundle targets agents with narrow preference margins
after correcting for the relative scarcity of the two goods.

The optimality condition for the pure option mechanism simplifies further in
the symmetric case. There, it is without loss of generality to consider symmetric bundling mechanisms whose two cutoffs are mirror images of one another.

\begin{corollary}\label{cor:2good_symmetric_two_options_theta}
Suppose \(g\) and \(\lambda\) are exchangeable and \(s_1=s_2\). Then the pure option mechanism is optimal if and only if, for every
\(z\in[0,1/2)\),
\begin{equation}\label{eq:foralltau_cor}
\mathbb E\left[
\lambda(\Theta)\bigl(\min_i\Theta_i-z\bigr)
\ \middle|\ 
\min_i\Theta_i\ge z
\right]
\le
(1-2z)\,
\mathbb E\left[
\lambda(\Theta)\max_i\Theta_i
\right].
\end{equation}
In particular, this holds if \(t\,\lambda(1-t,t)\) is non-decreasing in
\(t\in[1/2,1]\).
\end{corollary}

To interpret the condition, begin with the symmetric pure option mechanism in Figure~\ref{fig:2dsymmpure} and consider the following modification. Fix symmetric cutoff types \((1-z,z)\) and \((z,1-z)\), and offer every type between them a bundle containing \(2s(1-z)\) of each good, as in Figure~\ref{fig:2dsymmbundle}. The cutoff types are indifferent between the bundle and their corresponding pure options. For a type in the middle region, this modification raises utility from \(2s\max_i\Theta_i\) to \(2s(1-z)\), a gain of $2s(\min_i\Theta_i-z).$

\begin{figure}[h!]
    \centering
    \begin{subfigure}[b]{0.42\linewidth}
        \centering
\begin{tikzpicture}[xscale=0.7, yscale=0.8]

  \def\L{6.2}      
  \def\xm{\L/2}    
  \def\hP{2.30}    

  \fill[plum, opacity=0.35]
    (0,0) rectangle ({\xm},{\hP});

  \fill[darkdarkgreen!70, opacity=0.45]
    ({\xm},0) rectangle (\L,{\hP});

  \draw[dashed, thick] ({\xm},0) -- ({\xm},{\hP});

  \draw[decorate, decoration={brace, amplitude=6pt}, thick]
    (-0.35,0) -- (-0.35,\hP)
    node[midway,left=8pt] {$2s$};

  \draw[thick] (0,0) -- (\L,0);

  \draw[thick] (0,-0.18) -- (0,0.18);
  \draw[thick] (\L,-0.18) -- (\L,0.18);

  \draw[thick] ({\xm},-0.15) -- ({\xm},0.15);

  \node[below] at (0,-0.18) {$(1,0)$};
  \node[below] at (\L,-0.18) {$(0,1)$};
  \node[below] at ({\xm},-0.15) {$(\theta_1^0,\theta_2^0)$};

\end{tikzpicture}
        \caption{Figure \ref{fig:2dsymmpure}: Symmetric pure option mechanism}
        \label{fig:2dsymmpure}
    \end{subfigure}
    \hspace{0.7cm}
    \begin{subfigure}[b]{0.42\linewidth}
      \centering
\begin{tikzpicture}[xscale=0.7, yscale=0.8]

  \def\L{6.2}         
  \def\xm{\L/2}       
  \def\xa{\L/4}       
  \def\xb{3*\L/4}     
  \def\hP{2.30}       
  \def\hM{3.31}       

  \fill[plum, opacity=0.35]
    (0,0) rectangle ({\xa},{\hP});

  \fill[plum, opacity=0.35]
    ({\xa},0) rectangle ({\xb},{\hM/2});

  \fill[darkdarkgreen!70, opacity=0.45]
    ({\xa},{\hM/2}) rectangle ({\xb},{\hM});

  \fill[darkdarkgreen!70, opacity=0.45]
    ({\xb},0) rectangle (\L,{\hP});

  \draw[dashed, thick] ({\xa},0) -- ({\xa},{\hP});
  \draw[dashed, thick] ({\xb},0) -- ({\xb},{\hP});

  \draw[decorate, decoration={brace, amplitude=6pt}, thick]
    (-0.35,0) -- (-0.35,\hP)
    node[midway,left=8pt] {$2s$};

  \draw[decorate, decoration={brace, mirror, amplitude=6pt}, thick]
    (\L+0.35,0) -- (\L+0.35,{\hM/2})
    node[midway,right=8pt] {$2s(1-z)$};

  \draw[thick] (0,0) -- (\L,0);

  \draw[thick] (0,-0.18) -- (0,0.18);
  \draw[thick] (\L,-0.18) -- (\L,0.18);

  \draw[thick] ({\xm},-0.15) -- ({\xm},0.15);

  \draw[thick] ({\xa},-0.15) -- ({\xa},0.15);
  \draw[thick] ({\xb},-0.15) -- ({\xb},0.15);

  \node[below] at ({\xa},-0.15) {$(1-z,z)$};
  \node[below] at ({\xb},-0.15) {$(z,1-z)$};

\end{tikzpicture}
        \caption{Figure \ref{fig:2dsymmbundle}: Symmetric bundling mechanism}
        \label{fig:2dsymmbundle}
    \end{subfigure}
    \label{fig:2dsymmetric}
\end{figure}
The average gain conditional on belonging to this region is therefore the left-hand side of \eqref{eq:foralltau_cor}, multiplied by \(2s\). The modification, however, uses \(4s(1-z)\) units in total, rather than \(2s\), and thus requires an additional $2s(1-2z)$ units. Under the pure option mechanism, the shadow value of one unit of either good is $\mathbb E\![\lambda(\Theta)\max_i\Theta_i],$
since it can be used to increase the corresponding pure option. The average cost of the additional supply is therefore the right-hand side of \eqref{eq:foralltau_cor}, multiplied by \(2s\). Hence, the condition says precisely that, for every pair of symmetric cutoffs, the benefit of introducing the bundle does not exceed its additional supply cost.

The sufficient monotonicity condition has a natural interpretation. As in Example~\ref{example:ex2}, the bundle directs rewards toward agents with less extreme relative preferences. If these agents have lower values for their preferred good, as measured by \(t\lambda(1-t,t)\), doing so is counterproductive. Importantly, however, the opposite monotonicity of \(t\lambda(1-t,t)\) is \emph{not} sufficient to conclude that the designer should introduce the bundle. This is because bundling is an intrinsically distortionary screening device: to direct rents toward less extreme types, the mechanism must give them some of the good they value less and finance this by reducing other agents' allocations of their preferred good. Thus, even if less extreme relative preferences predict higher expected total values, this relationship must be strong enough to compensate for the resulting inefficiency.

\section{Observable labels}
\label{sec:observable-groups}

The baseline model assumes that agents differ only in their privately known
values. In many applications, however, the designer also observes
characteristics that are relevant for targeting, such as income, household
size, age, or disability status. I now extend the model
by allowing allocations to depend on such observable labels.

Suppose that each agent belongs to an observable group
\(r\in\{1,\dots,R\}\), $R\geq 2$, where group \(r\) has population share
\(\pi_r>0\) with $\sum_{r=1}^R \pi_r=1.$ Conditional on belonging to group \(r\), the agent's renormalized type
\(\theta\in\Gamma\) is distributed according to \(G_r\). Let
\(\lambda_r:\Gamma\to\mathbb R_+\) denote the expected scale of values
conditional on the observable group and renormalized type. Throughout, I impose
the following groupwise analogue of the regularity assumptions in the
baseline model.

\begin{assumption}
\label{ass:observable-groups}
For every group \(r\) and every nonempty \(B\subseteq\{1,\dots,N\}\), define
\[
\Theta^B:=\frac{(\Theta_i)_{i\in B}}{\sum_{i\in B}\Theta_i},
\qquad
\lambda_{r,B}(\Theta^B)
:=
\mathbb E\left[
\lambda_r(\Theta)\sum_{i\in B}\Theta_i
\ \middle|\ \Theta^B,r
\right].
\]
The distribution of \(\Theta^B\) admits a density \(g_{r,B}\) that is
strictly positive a.e., with \(g_{r,B}\) and its first weak
derivatives square-integrable. Moreover, \(\lambda_{r,B}\) is measurable,
bounded, and strictly positive a.e.
\end{assumption}

Because the group label is observable, the designer may offer a different
allocation rule to each group. Let \(x_r:\Gamma\to\mathbb R_+^N\) denote the
allocation rule offered to group \(r\). The designer's problem is therefore as follows.
\begin{problem}
\label{prob:observable-groups}
Choose allocation rules
\((x_1,\dots,x_R)\) to maximize
\begin{equation}
\label{eq:observable-group-objective}
\sum_{r=1}^R
\pi_r
\int_\Gamma
\lambda_r(\theta)\,
\theta\cdot x_r(\theta)\,
dG_r(\theta),
\tag{L-O}
\end{equation}
subject to within-group incentive compatibility,
\begin{equation}
\label{eq:observable-group-IC}
\theta\cdot x_r(\theta)
\ge
\theta\cdot x_r(\theta')
\qquad
\text{for all }\theta,\theta'\in\Gamma
\text{ and every }r,
\tag{L-IC}
\end{equation}
and the common supply constraint
\begin{equation}
\label{eq:observable-group-supply}
\sum_{r=1}^R
\pi_r
\int_\Gamma x_r(\theta)\,dG_r(\theta)
\le s.
\tag{L-S}
\end{equation}
\end{problem}

\paragraph{Relation to the baseline problem.}
The observable-group problem can be decomposed into two stages, in the spirit
of \citet{akbarpour2024redistributive}. In the \emph{inner stage}, for each
group \(r\), the designer fixes the total supply vector
\(s^r\in\mathbb R_+^N\) assigned to that group and allocates it optimally
among its members, subject to within-group incentive compatibility
\eqref{eq:observable-group-IC}. Indeed, let $W(\lambda,g,s)$ denote the value of the baseline problem when the weight function is \(\lambda\), the type density is \(g\), and the supply vector is
\(s\).\footnote{By Proposition~\ref{prop:P_redundant}, this problem's value is attained
when \(s\gg0\); if some coordinates are zero, Assumption~\ref{ass:observable-groups} ensures the same argument applies after restricting
the problem to the goods with strictly positive supplies.} Then the aggregate welfare obtained by optimally allocating the supply \(s^r\) within group \(r\) is given by
\begin{equation}
\label{eq:group-value-function}
W_r(s^r)
:=
\pi_r
W\left(
\lambda_r,
g_r,
s^r/\pi_r
\right).
\end{equation}
In the \emph{outer stage}, the designer splits the total supply across
groups, anticipating that each group's assigned supply will then be allocated
optimally within that group. The outer problem is therefore
\begin{equation}
\label{eq:supply-splitting-problem}
\max_{s^1,\dots,s^R\in\mathbb R_+^N}
\sum_{r=1}^R W_r(s^r)
\quad
\text{subject to}
\quad
\sum_{r=1}^R s^r\leq s.
\end{equation}

\paragraph{Optimal supply split.}
Since each inner problem is an instance of the baseline model, we now focus
on the outer problem of splitting supply across groups. For every group \(r\)
and good \(i\), define
\[
\eta_{ri}
:=
\int_\Gamma
\lambda_r(\theta)\theta_i\,dG_r(\theta)
=
\mathbb E[V_i\mid r],
\qquad \text{and} \qquad
\overline v_{ri}
:=
\sup_{\theta\in\Gamma}
\lambda_r(\theta)\theta_i
=\sup_{\theta\in\Gamma}
\mathbb E[V_i\mid r,\Theta=\theta].
\]
Thus, \(\eta_{ri}\) is group \(r\)'s average value for good \(i\), while
\(\overline v_{ri}\) is the highest conditional expected value for good \(i\)
among the relative-value profiles in group \(r\). We then get the following result:

\begin{theorem}\label{thm:3}
Let \((s^1,\dots,s^R)\) solve \eqref{eq:supply-splitting-problem}. Then 
every group \(r\) is either \textbf{active} and gets \(s^r\gg0\), or \textbf{excluded} and gets \(s^r=0\). Moreover:
\begin{enumerate}
\item[(i)] group $r$ is excluded if $\max_{k\neq r}\eta_{ki}
\geq
\overline v_{ri},$ for every good $i$, with at least one strict inequality,
\item[(ii)] group $r$ is active if there exists a nonempty set of
goods $B\subseteq\{1,\dots,N\}$ such that
\begin{equation}
\label{eq:comparative-advantage-participation-difference}
\sum_{i\in B}
s_i
\left(
\eta_{ri}
-
\max_{l\neq r}\overline v_{l i}
\right)
>
\sum_{i\notin B}
s_i
\left(
\max_{l\neq r}\overline v_{l i}
-
\max_{k\neq r}\eta_{ki}
\right).
\end{equation}
\end{enumerate}
\end{theorem}

Thus, an observable group may be excluded altogether, but any group that is served receives a positive amount of every good. To see why, suppose a group receives some goods but none of good \(j\). Because the group contains agents who value almost exclusively good \(j\), these agents must instead receive bundles of goods they value very little. Even a small allocation of good \(j\) could then replace those bundles and free substantial quantities of the other goods.

Part~$(i)$ gives a sufficient condition for exclusion: for every good, even the highest value in group \(r\) is no greater than the average value in some other group. Any supply assigned to \(r\) can then be reallocated, good by good, to other groups without reducing welfare, making it strictly suboptimal to serve \(r\). The following example illustrates
when this can occur.
\begin{example}
\label{ex:specialists-exclude-generalist}
Consider \(N\) goods and \(N+1\) groups, indexed by
\(0,1,\dots,N\). Fix \(H>L>0\), \(\delta>0\), and a small
\(\varepsilon>0\), and suppose that $H-L>\frac{\delta}{2}.$ In every group, an \(\varepsilon\)-share of agents have values distributed
uniformly on $[0,L]^N.$ For the remaining \(1-\varepsilon\) share, values are independent. In group \(0\), they satisfy
\[
V_j\mid r=0
\sim
\operatorname{Unif}[L,L+\delta]
\qquad
\text{for every }j,
\]
whereas in each group \(i=1,\dots,N\), they satisfy
\[
V_i\mid r=i
\sim
\operatorname{Unif}[H,H+\delta],
\qquad
V_j\mid r=i
\sim
\operatorname{Unif}[L,L+\delta]
\quad
\text{for every }j\neq i.
\]
Then, for all sufficiently small \(\varepsilon>0\), group \(0\) is excluded.
\end{example}
The small uniform component ensures that every group satisfies Assumption~\ref{ass:observable-groups}. Note also that the
highest value in group \(0\) is \(L+\delta\), while the average value in
specialist group \(i\) converges to \(H+\delta/2>L+\delta\) as
\(\varepsilon\downarrow0\). Thus, for all sufficiently small
\(\varepsilon>0\), the specialist groups dominate group \(0\) good by good.

Part~$(ii)$, in turn, gives a sufficient condition for activity. A group is guaranteed to be active when it has a sufficiently strong \emph{comparative advantage} in some set of goods \(B\). The left-hand side of \eqref{eq:comparative-advantage-participation-difference} measures the value of this advantage, while the right-hand side bounds the cost of supplying the group with the remaining goods. If the former exceeds the latter, exclusion cannot be optimal.

\begin{example}
\label{ex:parametric-comparative-advantage}
Suppose there are \(N\) groups and \(N\) goods. Fix \(H>L>0\),
\(\delta>0\), and a small \(\varepsilon>0\). In each group \(r\), an
\(\varepsilon\)-share of agents have values distributed uniformly on $[0,L]^N.$ Among the remaining agents, values are independently distributed according to
\[
V_r\mid r
\sim
\operatorname{Unif}[H,H+\delta],
\qquad
V_i\mid r
\sim
\operatorname{Unif}[L,L+\delta]
\quad
\text{for every }i\neq r.
\]
Suppose, moreover, that
\begin{equation}
\label{eq:parametric-comparative-advantage-gap}
\frac{\delta}{2}\sum_{i=1}^N s_i
<
(H-L)\min_r s_r.
\end{equation}
Then, for all sufficiently small \(\varepsilon>0\), every group is active.
\end{example}

When \(\varepsilon=0\), each group \(r\) has a
comparative advantage in good \(r\). Taking \(B=\{r\}\), condition
\eqref{eq:comparative-advantage-participation-difference} reduces to $(H-L)s_r
>
\frac{\delta}{2}\sum_{i=1}^N s_i,$ which follows from
\eqref{eq:parametric-comparative-advantage-gap}. Since the inequality is
strict, it continues to hold for all sufficiently small
\(\varepsilon>0\). Thus, each group's advantage in its specialized good is
large enough to justify serving it.

\paragraph{Competitive equilibrium with \emph{unequal} incomes.}
One natural way to incorporate observables into the mechanism is to adapt the
CEEI mechanism from the baseline model by adjusting agents' budgets of
artificial currency according to their labels. This produces a competitive
equilibrium with \emph{unequal} incomes, where all agents can use their
group-specific budgets to purchase goods at common, market-clearing prices.
Indeed, such mechanisms have been used in practice: for instance, Feeding
America's Choice System lets food banks bid for food donations using artificial
currency, with banks serving larger populations receiving higher budgets
\citep{prendergast2017food,prendergast2022allocation}. I now define the CEUI
mechanism formally:

\begin{definition}
\label{def:CEUI}
A \textbf{competitive equilibrium with unequal incomes} (CEUI) consists of a
common price vector \(p\in\mathbb R_{++}^N\), group-specific budgets
\(b_1,\dots,b_R\geq0\), and group-specific allocation rules
\(x_r:\Gamma\to\mathbb R_+^N\) such that markets clear,
\[
\sum_{r=1}^R
\pi_r
\int_\Gamma x_r(\theta)\,dG_r(\theta)
=
s,
\]
and every type chooses an optimal bundle subject to her group's budget:
\[
x_r(\theta)
\in
\argmax_{z\in\mathbb R_+^N}
\left\{
\theta\cdot z:
p\cdot z\leq b_r
\right\}
\qquad
\text{for every }\theta\in\Gamma
\text{ and every }r.
\]
\end{definition}
However, I show that this way of accounting for observables is generically suboptimal. To state this
result, I use a notion of genericity analogous to Definition~\ref{def:stability}.
\begin{definition}
A property holds \textbf{generically at \(((\lambda_r^*,g_r^*)_{r=1}^R,s^*)\)} if there exists \(\varepsilon>0\) such that the set of perturbations $
\bigl((l_r,h_r)_{r=1}^R,\Delta s\bigr)$ satisfying
\[
\lambda_r=e^{l_r}\lambda_r^*,
\quad
g_r=
\frac{e^{h_r}g_r^*}
{\int_\Gamma e^{h_r}g_r^*\,dm},
\quad
s=s^*+\Delta s>0,
 \quad \text{and} \quad
\sum_{r=1}^R
\left(
\|l_r\|_{C^1(\Gamma)}
+
\|h_r\|_{C^2(\Gamma)}
\right)
+
\|\Delta s\|
<\varepsilon
\]
and inducing admissible primitive vectors for which the property holds contains an open and dense subset of this neighborhood.\footnote{Openness and density are understood with respect to the product topology on $[C^1(\Gamma)\times C^2(\Gamma)]^R\times\mathbb R^N$.}
\end{definition}

\begin{proposition}\label{prop:generic-common-pseudomarket}
A CEUI mechanism (with at least two active groups) is not generically
optimal at any primitives \(((\lambda_r^*,g_r^*)_{r=1}^R,
s^*)\) satisfying Assumption~\ref{ass:observable-groups}, with
\(g_r^*\in C^2(\Gamma)\), \(\lambda_r^*\in C^1(\Gamma)\),
and \(s^*\gg 0\).
\end{proposition}

Importantly, the result allows both prices and group-specific budgets to vary
with the perturbation. Nevertheless, it shows that the optimality of CEUI mechanisms is not
robust to small changes in the primitives. Intuitively, this is because the designer generally benefits from tailoring the relative quantities of different goods offered to each group to its comparative advantages. Group-specific budgets allow her to scale all quantities offered to a group, but common prices force the vectors of quantities offered to be proportional across groups.

\section{Discussion}\label{sec:discussion}

\paragraph{Implications for market design.}
The main lesson of my paper is that, in settings without transfers, the
welfare-maximizing mechanism depends on how agents' absolute and relative
values covary. When higher-value agents tend to be more selective, mechanisms
offering only pure options are likely to be optimal. When they tend to be less
selective, the designer can sometimes improve upon these mechanisms by offering
mixed allocations to reward agents willing to give up choice.

This logic informs design choices in a range of settings. Consider, for instance,
the allocation of food donations across food banks, studied by
\citet{prendergast2017food,prendergast2022allocation} and
\citet{altmann2023choice}: Feeding America, a national nonprofit organization whose network includes more
than 200 food banks, allocates donated food by allowing participating banks to
bid for it using an artificial currency. This currency
is distributed according to a formula based on the population each food bank
serves. However, this measure of need is imperfect because food banks also
differ in the generosity of local donors operating outside the system. My results suggest that if banks receiving fewer local donations also have weaker preferences over the
composition of the food they receive, the system could target them by letting banks bid on broad, pre-packaged assortments alongside individual donations. Since accepting a fixed composition is less
attractive than choosing items separately, these bundles would trade at a discount relative to their components. Banks with greater unmet need but weaker compositional preferences could then obtain more food with their
budgets, while more selective banks could continue to bid item by item. If
this relationship between unmet need and flexibility is weak or reversed,
however, the item-by-item pseudomarket is likely to be optimal.

The results in Section~\ref{sec:observable-groups} also
suggest that observable differences across food banks should affect more than
their artificial-currency budgets. For example, banks in agricultural regions
may receive ample fresh produce from local growers but relatively little
shelf-stable food. If this leads them to value shelf-stable items more highly
than banks in other regions, they have a ``comparative advantage'' in those items,
in the sense of Theorem~\ref{thm:3}. The optimal mechanism may then reserve
part of its stock of shelf-stable food for such banks, rather than responding
to this observable difference only by adjusting their overall budgets.

My analysis also speaks to the design of mechanisms for allocating affordable housing. This environment differs from the baseline model because eligible applicants and
vacant units arrive over time, and housing authorities often use dynamic
mechanisms such as waitlists or repeated lotteries. Nevertheless, I show in Appendix~\ref{app:waitlist} that the baseline model provides a reduced-form description of the steady states of such dynamic mechanisms when we interpret \(y_i(v)\) as an applicant's total discounted allocation. Under this interpretation, several
mechanisms used in practice correspond to the pure option mechanism. One
example is a repeated choice-based lottery, the dynamic analogue of the
lottery in Definition~\ref{def:pure_implementations}: in each period, a
separate lottery is held for each kind of development, applicants choose one,
and the available vacancies of each kind are allocated randomly among those who entered the corresponding lottery.\footnote{
Tokyo's Toei housing uses repeated choice-based lotteries, with a separate lottery for each application district; see \url{https://www.to-kousya.or.jp/kouei/toeibosyu/index_teiki.html}.}
Another example is a system with one waitlist for each development where, in equilibrium, more
oversubscribed developments come with longer waits.

Other mechanisms used in practice instead reward applicants who are willing
to accept less choice. For example, certain systems let applicants enter
several development-specific lotteries, so applicants willing to accept more
locations have more opportunities to receive housing.\footnote{For example,
the Amsterdam housing lottery allows applicants to enter two draws per week;
see
\url{https://www.wooninfo.nl/nieuws/2013/04/nieuw-een-woning-via-loting/}.}
Similarly, waitlist systems often allow applicants to reject offers, but
sometimes at the cost of losing priority or leaving the list. These mechanisms
do not necessarily implement the precise mixed option characterized in the
model, but they illustrate how willingness to accept a less tailored
allocation can be used as a screening instrument. Indeed, the literature on housing allocation has recognized the tradeoff between choice and targeting; however, it has focused on comparing extreme mechanisms assigning applicants without choice and letting them select a particular development \citep{arnosti2020design,waldinger2021targeting}. My paper instead
considers a general class of mechanisms and shows that intermediate designs
offering both routes may be optimal. For instance, public-housing systems that
use waitlists could offer both development- or location-specific lists and a
priority list on which applicants must accept any suitable unit that becomes
available; this priority list would be shorter in equilibrium. Applicants
could therefore choose between retaining control over where they are housed
and receiving housing sooner. Similarly, a repeated-lottery system could let
applicants choose between a lottery for a particular development and a pooled
lottery with better odds.

These mechanisms could be combined with
priorities based on observable characteristics, as outlined in
Section~\ref{sec:observable-groups}. Importantly, however, they would also
target applicants with high \emph{unobserved need} that is difficult to
document. For example, applicants in precarious informal living arrangements
may be unable to demonstrate the full extent of their housing insecurity but,
if this makes them more willing to give up choice, can reveal their need
through their willingness to accept any suitable unit in exchange for faster
assignment.

\paragraph{Future directions.}
This paper focuses on identifying conditions under which the pure option
mechanism is optimal. To establish these results, I develop a strong-duality
framework that could also be used to certify the optimality of other candidate
mechanisms. A natural next step would therefore be to use it to
extend the analysis beyond the pure option mechanism, for example by
characterizing the optimal mechanism within a restricted class of menus, such
as those that add a single mixed bundle to the pure options.

Another direction would be to allow agents to have heterogeneous endowments of the goods, which they can choose to bring into the mechanism and trade with others. Such a model would describe, for example, public-housing settings in which existing tenants can apply to relocate while retaining the right to remain in their current unit. Formally, this extension could be accommodated by adding agents' endowments to the designer's supply and imposing individual-rationality constraints that give each agent the option of keeping her endowment. This extension is particularly natural in the present framework. Once the participation constraints are incorporated through a Lagrangian, their shadow values simply augment the corresponding types' weights \(\lambda\). Under the appropriately chosen shadow values, the endowment problem therefore reduces to an instance of the baseline problem.

Finally, one could extend the model by relaxing the assumption of linear utility. Linearity is
natural when allocations represent probabilities or discounted assignments.
When agents receive physical quantities of several goods, however,
complementarities and substitutabilities may matter. While sufficiently strong
complementarities could overturn the precise optimality result for the pure option
mechanism, I conjecture that the broader logic behind my results would survive:
Pareto-efficient mechanisms are likely to remain optimal in a non-trivial set of
cases. Even in these richer environments, the absence of transfers leaves the
designer with only a limited set of instruments for redistributing rents: if
departures from Pareto efficiency shift rents only toward one group of agents,
the efficient mechanism remains optimal when the designer instead favors the other.

\section{Proof of Theorem \ref{thm:rent-measure-characterization}}\label{sec:proofdetails}
I now present the key steps in the proof of Theorem~\ref{thm:rent-measure-characterization} along with additional technical intuition; supporting facts and lemmas are shown in the appendix.

\paragraph{Characterizing incentive compatibility.} I first characterize incentive-compatible allocation rules in the renormalized problem using the partial order introduced in Definition \ref{def:partialorder}.

\begin{proposition}\label{thm:icchar}
A function \(U:\Gamma\to\mathbb R_+\) is the indirect utility function of some allocation rule \(x:\Gamma\to\mathbb R_+^N\) satisfying \eqref{eq:ICprime} if and only if \(U\) is convex
and satisfies the following condition:
\begin{equation}\label{eq:pseudomono}
\text{for every $i$ and every $\theta,\theta'$ in $\Gamma$ such that $\theta \succ_i \theta'$}, \quad
\frac{U(\theta')}{\theta_i'}\ \ge\ \frac{U(\theta)}{\theta_i}.
\tag{R}
\end{equation}
\end{proposition}
Incentive-compatible allocation rules must produce convex indirect utility functions $U$, as they are maxima of affine functions of $\theta$, i.e. $U(\theta)=\max_{\theta'\in\Gamma}\ \theta\cdot x(\theta').$ Condition \eqref{eq:pseudomono} additionally restricts how fast indirect utility $U(\theta)$ can grow as $\theta$ moves towards the vertex $e_i$. To understand why \eqref{eq:pseudomono} is necessary for incentive compatibility, fix any good $i$ and two types such that $\theta \succ_i \theta'$. Note that normalizing $U(\theta)$ by $\theta_i$ gives
\[
\frac{U(\theta)}{\theta_i} \ = \ \sum_{k\neq i} \frac{\theta_k}{\theta_i} \ x_k(\theta) \ + \ x_i(\theta).
\]
We can then equivalently think of type-$\theta$ agents as maximizing their scaled utilities $U(\theta)/\theta_i$. Recall also that by the definition of the $\succ_i$-order, all the ratios $\theta_k'/\theta_i'$ are higher for $\theta'$ than for $\theta$. This implies that type \(\theta'\) can always guarantee a higher scaled
indirect utility than type \(\theta\). Indeed, by reporting \(\theta\), she
can obtain \(x(\theta)\), so
\[
\frac{U(\theta')}{\theta_i'}
\ge
\sum_{k\neq i}\frac{\theta_k'}{\theta_i'}x_k(\theta)+x_i(\theta)
\ge
\sum_{k\neq i}\frac{\theta_k}{\theta_i}x_k(\theta)+x_i(\theta)
=
\frac{U(\theta)}{\theta_i},
\]
where the second inequality follows from
\(\theta_k'/\theta_i'\ge\theta_k/\theta_i\) for every \(k\neq i\). As it turns out, convexity of $U(\theta)$ and \eqref{eq:pseudomono} are also sufficient for incentive compatibility.\footnote{While working with \eqref{eq:ICprime} in the simplex representation $\Gamma$ is more analytically convenient in my setting, one can also characterize it in terms of an indirect utility $\tilde U:\mathbb R_+^N\to\mathbb R$ defined on unnormalized values $v$. In a quasilinear model with transfers, an allocation rule is incentive compatible if and only if its indirect utility $\tilde U$ is convex and non-decreasing in each coordinate \citep{ROCHET1987191}.
Without transfers, incentive compatibility additionally forces $\tilde U$ to be positively homogeneous of degree one: for all $v\in\mathbb R_+^N$ and $k>0$, $\tilde U(kv)=k\,\tilde U(v)$ \citep{lahr2024extreme}.}

\paragraph{Expressing the problem in utility space.} It will be useful to reformulate the designer's problem in terms of indirect utility functions because, as shown in Proposition~\ref{thm:icchar}, incentive compatibility can be characterized in terms of their shape. Indeed, writing the objective as an integral over weighted indirect utilities, rather than weighted allocations, is standard in the multidimensional screening literature.\footnote{See, for instance, \cite{armstrong1996multiproduct, rochet1998ironing, manelli2006bundling,daskalakis2013mechanism,daskalakis}.} The reformulation relies on the following result.

\begin{corollary}\label{lem:utility-space-corr}
Every indirect utility function \(U\) implemented by some allocation rule
satisfying \eqref{eq:ICprime} is Lipschitz continuous. Let \(\nabla_H U\) denote the gradient of \(U\) intrinsic to the hyperplane containing the simplex \(\Gamma\) and, at every
\(\theta\in\Gamma^\circ\) where \(\nabla_H U(\theta)\) exists, define
\begin{equation}\label{eq:x-from-U}
x_U(\theta)
:=
\nabla_H U(\theta)
-
\mathbf 1\bigl(\nabla_H U(\theta)\cdot\theta-U(\theta)\bigr).
\end{equation}
Then
\(x_U(\theta)\in\mathbb R_+^N\), \(x_U\) is uniformly bounded, and
\(U(\theta)=\theta\cdot x_U(\theta)\) almost everywhere. Furthermore, every
allocation rule satisfying \eqref{eq:ICprime} and implementing \(U\)
coincides with \(x_U\) a.e.
\end{corollary}

Consequently, Problem~\ref{prob:renormalized} can be written as
\begin{equation}\label{eq:U-objective}
\sup_{U\in\mathcal U}
\int_\Gamma \lambda(\theta)U(\theta)g(\theta)d\theta
\quad
\text{subject to}
\quad
\int_\Gamma x_U(\theta)g(\theta)d\theta\le s,
\end{equation}
where
\begin{equation}\label{eq:UL}
\mathcal U
:=
\left\{
U\in C(\Gamma):
U\geq 0,\ 
U\text{ is convex and satisfies }\eqref{eq:pseudomono}
\right\}.
\end{equation}
Next, I incorporate the supply constraint in \eqref{eq:U-objective} into the objective using the supply costs from Definition~\ref{def:supply_costs}, and explain the economic intuition behind their construction.

\paragraph{Supply costs: construction and intuition.} Consider an exercise in which the designer can allocate any amount of the \(N\) goods, but must pay per-unit costs \(c=(c_1,\dots,c_N)\) for them. Starting from the pure option mechanism with quantities \(q=(q_1,\dots,q_N)\), ask: what cost vector would make the designer indifferent about marginally perturbing these quantities? To pin down those costs, choose any good \(k\) and consider increasing \(q_k\) by \(\epsilon\), holding the other quantities fixed. To first order, this perturbation has two effects, illustrated in Figure~\ref{fg:multiplierpic}. First, agents in \(\Gamma_k\) who chose option \(k\) before continue to do so, but now receive a higher quantity. This raises their utility and incurs an additional cost of \(c_k\,\epsilon\) per agent. Second, the perturbation induces some agents who previously chose \(q_j\), \(j\neq k\), to switch to \(q_k\). For each such agent, the designer pays \(c_k q_k\) instead of \(c_j q_j\). The direct welfare effects of these switchers are not first-order: both their mass and their welfare change are of order \(\epsilon\). As \(\epsilon\) becomes small, the sum of these effects, divided by \(\epsilon\), converges to:
\[
A_k - c_k M_k + \sum_{j\neq k} T_{kj}\,(c_j q_j - c_k q_k).
\]
Thus, the system $Jc = A$ defining the supply costs captures precisely the first-order conditions ensuring that such perturbations are not beneficial.

\begin{figure}[h!]
  \centering
\begin{tikzpicture}[scale=2.5, line join=round, line cap=round]
  \coordinate (ei) at (90:1);
  \coordinate (ej) at (210:1);
  \coordinate (ek) at (330:1);

  \path[use as bounding box] (-1.25,-0.95) rectangle (1.25,1.25);

  \coordinate (thp) at (0.20,-0.05); 

  \path[name path=sideIEJ] (ei)--(ej);
  \path[name path=sideIEK] (ei)--(ek);
  \path[name path=sideJEK] (ej)--(ek);

  \path[name path=rayJ] (ej)--($(ej)!3!(thp)$);
  \path[name path=rayI] (ei)--($(ei)!3!(thp)$);
  \path[name path=rayK] (ek)--($(ek)!3!(thp)$);

  \path[name intersections={of=rayJ and sideIEK, by={Pj}}];
  \path[name intersections={of=rayI and sideJEK, by={Pi}}];
  \path[name intersections={of=rayK and sideIEJ, by={Pk}}];

  \coordinate (th1) at ($(ek)!1.1!(thp)$);

  \path[name path=rayI1] (ei)--($(ei)!3!(th1)$);
  \path[name path=rayJ1] (ej)--($(ej)!3!(th1)$);

  \path[name intersections={of=rayI1 and sideJEK, by={Qi}}];
  \path[name intersections={of=rayJ1 and sideIEK, by={Qj}}];

  \fill[darkgreen, opacity=0.50]
    (th1) -- (Qi) -- (Pi) -- (thp) -- (Pj) -- (Qj) -- cycle;

  \draw[darkdarkgreen, very thick] (ei) -- (Qi);
  \draw[darkdarkgreen, very thick] (ej) -- (Qj);

  \draw[black, very thick] (ei) -- (Pi);
  \draw[black, very thick] (ej) -- (Pj);

  \draw[grey, thick] (ek) -- (Pk);

  \node[above]       at (ei) {$e_i$};
  \node[below left]  at (ej) {$e_j$};
  \node[below right] at (ek) {$e_k$};

  \fill (thp) circle (0.03);
  \fill (th1) circle (0.03);

  \draw[thick] (ei)--(ej)--(ek)--cycle;
    \fill[plum, opacity=0.50] (ek) -- (Pj) -- (thp) -- (Pi) -- cycle;
\end{tikzpicture}
            \vspace*{-4.2mm} 
            \captionsetup{width=0.6\linewidth}
            \caption{First-order effects of increasing the pure option $q_k$. Agents in the violet region receive more of good $k$; agents in the green region switch from other goods to $k$.}
            \label{fg:multiplierpic}
\end{figure}

I now write $U_{\mathrm{pure}}$ for the indirect utility from the pure option mechanism. The following result shows that we can verify the optimality of $U_{\mathrm{pure}}$ using the Lagrangian relaxation of~\eqref{eq:U-objective}, with the multipliers on the supply constraints set to $c=J^{-1}A$.

\begin{proposition}\label{prop:supply_cost_characterization}
\(U_{\mathrm{pure}}\) solves \eqref{eq:U-objective} if and only if
\(U_{\mathrm{pure}}\) solves
\begin{equation}\label{eq:problem-supplycost}
  \max_{U\in\mathcal U}
\left\{
\int_\Gamma \lambda(\theta)U(\theta)g(\theta)d\theta
-
c\cdot\left(\int_\Gamma x_U(\theta)g(\theta)d\theta-s\right)
\right\},
\qquad
c=J^{-1}A.
\end{equation}
\end{proposition}

\paragraph{Measure representation of the objective.} I now rewrite the objective in \eqref{eq:problem-supplycost} using a signed measure and show that this measure nets to zero on each option-specific region \(\Gamma_i\). Since \(U_{\mathrm{pure}}(\theta)/\theta_*(\theta)\) is constant on each \(\Gamma_i\), this zero-mass property implies that the pure option mechanism has value zero in this representation.
\begin{lemma}\label{lemma:measure-objective}
Let  \(c=J^{-1}A\). Up to the constant term \(c\cdot s\), the objective in \eqref{eq:problem-supplycost} can be written as
\begin{equation}\label{eq:measureobj}
  \int_{\Gamma}\frac{U(\theta)}{\theta_*(\theta)} d\mu(\theta),
\end{equation}
where the signed measure $\mu$ is given by \eqref{eq:signedmeas}. Moreover, \(\mu\) satisfies the following balance conditions:
\begin{equation}\label{eq:gammai-balance}
\text{for all $i$,} \quad \mu(\Gamma_i)=0.
\end{equation}
Consequently,
\begin{equation}\label{eq:primalzeropure}
\int_\Gamma \frac{U_{\mathrm{pure}}(\theta)}{\theta_*(\theta)}\,d\mu(\theta)
=
\sum_{i=1}^N q_i\,\mu(\Gamma_i)=0.
\end{equation}
\end{lemma}

The rent measure \(\mu\) assigns zero mass to the indifference boundaries
between distinct pure-option regions. Hence, on each region \(\Gamma_i\),
we have \(\theta_*(\theta)=\theta_i\) for \(\mu\)-almost every \(\theta\), and
therefore
\[
\frac{U_{\mathrm{pure}}(\theta)}{\theta_*(\theta)}=q_i
\qquad
\text{for \(\mu\)-almost every }\theta\in\Gamma_i.
\]
Thus, at the pure option mechanism, the rent measure weights the quantity of
the good chosen by each type. It therefore captures the designer's marginal value of increasing each
type's chosen pure option, after accounting for supply costs and how such
changes propagate through local incentive constraints. In this sense, it plays a role
analogous to a virtual value in one-dimensional screening. For a
general mechanism, however, \(U(\theta)/\theta_i\) need not equal the amount
of good \(i\) allocated to type \(\theta\). Instead, it is the amount of good
\(i\) that would give the type the same indirect utility.\footnote{
One could ask why the objective is not written in terms of the allocation rule $x$, as is typically done in one-dimensional mechanism design. Indeed, one could integrate \eqref{eq:measureobj} by parts to obtain a representation involving the allocation rule \(x\). However, because \(x\) is a vector field, such a representation is not unique: it depends on a choice of vector-valued flows which, intuitively, correspond to sets of paths in the type space \(\Gamma\) along which one integrates by parts. Then, when optimizing over \(x\) to maximize such an expression, one implicitly accounts only for the effects of perturbing \(x\)
that propagate through local IC constraints along these paths. In general, this can lose important information about effects propagating through other local IC constraints.

Representing the designer's objective in terms of \(U\) avoids this issue: since \(U\) is a scalar potential, the objective can be rewritten in terms of \(U\) without having to select paths along which indirect utility is integrated. As a result, this representation encodes information about effects propagating through \emph{all} local IC constraints.}

\paragraph{The designer's problem as a linear program.}
We have so far established that the pure option mechanism is optimal if and only if
\(U_{\mathrm{pure}}\) solves
\[
\max_{U\in K}
\int_\Gamma \frac{U}{\theta_*}\,d\mu
\quad
\text{subject to}\quad \eqref{eq:pseudomono},
\quad \text{where}\quad
K
:=
\left\{
U\in C(\Gamma):
U\ge 0
\text{ and }
U\text{ is convex}
\right\}.
\]
To derive the transport-based optimality conditions, I reformulate this problem as an infinite-dimensional linear program and use a duality approach. To that end, define, for each $i$,
\[
R_i
:=
\left\{
(\theta,\theta')\in\Gamma\times\Gamma:
\theta_k\theta_i'\ge \theta_k'\theta_i
\text{ for every } k\neq i
\right\},
\]
which is the closed, cross-multiplied version of the order
\(\theta\prec_i\theta'\). Since \(R_i\) is a closed subset of the compact set
\(\Gamma\times\Gamma\), it is compact. Now, for \(U\in C(\Gamma)\), define
\[
(B_iU)(\theta,\theta')
:=
\theta_i U(\theta')-\theta_i' U(\theta),
\qquad
(\theta,\theta')\in R_i.
\]
The constraint \(B_iU\le0\) on \(R_i\) is the
cross-multiplied version of \eqref{eq:pseudomono}; unlike the ratio form, it
remains well-defined on the boundary of the simplex. This lets us write our
problem as
\begin{equation}\label{eq:primal}
\sup_{U\in K}
\left\{
\int_\Gamma \frac{U}{\theta_*}\,d\mu:
B_iU\le0\text{ on }R_i
\text{ for every }i
\right\}.
\end{equation}
Since \(K\) is a convex cone, this is a linear program. Then, the optimality of
\(U_{\mathrm{pure}}\) in problem \eqref{eq:primal} is equivalent to the existence of dual multipliers for the constraints~\eqref{eq:pseudomono}. 
\begin{definition}\label{def:dualcert}
A \textbf{dual certificate for \(U_{\mathrm{pure}}\)} is a tuple of finite
nonnegative Borel measures $(\gamma_1,\dots,\gamma_N)$ such that $\gamma_i\in\mathcal M_+(R_i),$ and, for every \(U\in K\),
\begin{equation}\label{eq:dualcert}
\int_\Gamma \frac{U}{\theta_*}\,d\mu
\le
\sum_{i=1}^N
\int_{R_i}
(B_iU)(\theta,\theta')\,d\gamma_i(\theta,\theta').
\end{equation}
\end{definition}

\begin{proposition}\label{prop:dualcert}
The pure option mechanism is optimal if and only if a dual certificate for
\(U_{\mathrm{pure}}\) exists.
\end{proposition}

We will soon relate these dual certificates to the transports in Theorem \ref{thm:rent-measure-characterization}. Note, however, that while these transports happen only within some region $\Gamma_i$, the dual certificates from Proposition~\ref{prop:dualcert} may place weight on pairs of types drawn from different regions $\Gamma_i$ and $\Gamma_j$. The next lemma shows that such cross-region constraints are redundant: any dual certificate can be replaced by one that operates entirely within each region.

\begin{lemma}
\label{lem:regionwise_exact_certificate}
Suppose a dual certificate for $U_{\mathrm{pure}}$ exists. Then there also exists another dual certificate $(\bar\gamma_1,\dots,\bar\gamma_N)$ such that, for every $j$, $\bar\gamma_j\in\mathcal M_+(R_j)$ and $\operatorname{supp}(\bar\gamma_j)
\subseteq
R_j\cap(\Gamma_j\times\Gamma_j).$
\end{lemma} 

\paragraph{Deriving preprocessing and transport conditions.}
The final step converts the regionwise dual certificates into the
preprocessing and transport conditions in Theorem
\ref{thm:rent-measure-characterization}.

\begin{lemma}\label{lem:certificate-transport-equivalence}
A regionwise dual certificate exists if and only if there exists a tuple
\((\tilde\mu_1,\dots,\tilde\mu_N)\) of finite signed measures, with
\(\tilde\mu_i\) defined on \(\Gamma_i\), satisfying the preprocessing and
regionwise transport conditions in
Theorem~\ref{thm:rent-measure-characterization}.
\end{lemma}

The equivalence reflects two ways of representing the same dual object. Each
measure \(\bar\gamma_i\) in a regionwise certificate is supported on binding
constraints between ordered pairs \(\theta\prec_i\theta'\) within the same
region \(\Gamma_i\). For a general \(U\), the difference associated with such
a pair is
\[
(B_iU)(\theta,\theta')
=
\theta_i\theta_i'
\left(
\frac{U(\theta')}{\theta_i'}
-
\frac{U(\theta)}{\theta_i}
\right).
\]
The multiplier measure can therefore be interpreted as moving mass from
\(\theta\) to \(\theta'\), that is, upward in the \(\succ_i\)-order.
Reweighting it by \(\theta_i\theta_i'\) ensures that the difference between
its second and first marginals multiplies the normalized utility
\(U(\theta)/\theta_i\). The resulting marginal difference defines the
region-specific auxiliary measure \(\tilde\mu_i\). Summing these regional
contributions gives the preprocessing inequality. The first and second
marginals describe the mass sent and received within each region and, after
common mass is cancelled, yield a transport from \(\tilde\mu_i^-\) to
\(\tilde\mu_i^+\).

By Proposition~\ref{prop:dualcert} and
Lemma~\ref{lem:regionwise_exact_certificate}, the pure option mechanism is
optimal if and only if a regionwise dual certificate exists. Lemma
\ref{lem:certificate-transport-equivalence} therefore proves Theorem
\ref{thm:rent-measure-characterization}.

\paragraph{Interpreting the transport condition.}
For each region \(\Gamma_i\), the positive part \(\tilde\mu_i^+\) places
weight on types whose normalized utilities the designer would like to raise,
after preprocessing and after accounting for supply costs and how such changes
propagate through local IC constraints. Conversely, \(\tilde\mu_i^-\) places
weight on types whose normalized utilities the designer would like to reduce.
The designer cannot, however, choose \(U\) freely:
Proposition~\ref{thm:icchar} shows that feasible allocation rules must
generate indirect utilities satisfying certain shape restrictions. In
particular, \eqref{eq:pseudomono} limits how quickly \(U(\theta)\) can
increase as \(\theta\) moves toward the vertices of \(\Gamma\).

The transport condition in
Theorem~\ref{thm:rent-measure-characterization} formally captures the idea
that, after possibly ironing the rent measure into the region-specific
auxiliary measures, the positive part \(\tilde\mu_i^+\) lies closer to vertex
\(e_i\) than the negative part \(\tilde\mu_i^-\) on each region
\(\Gamma_i\). When this condition holds, the designer wants to raise utility
as much as possible for types close to \(e_i\); the best feasible way to do
so is to make \(U(\theta)\) increase as rapidly as
\eqref{eq:pseudomono} permits as one moves toward that vertex. The pure
option indirect utility \(U_{\mathrm{pure}}\) is exactly the extremal
utility profile that does this: it makes the constraints
\eqref{eq:pseudomono} bind on each pure option region.

\paragraph{Relation to techniques in the literature.}
My proof builds on strong-duality methods for multidimensional screening. As
in \cite{kleiner2019strong} and \cite{kleiner2022optimal}, who study
multi-product monopoly pricing and multidimensional delegation, respectively,
I formulate the designer's problem as an infinite-dimensional linear program
and use measures as dual multipliers. The constraint \eqref{eq:pseudomono} I dualize, however, is specific to the present problem, and thus existing strong-duality results cannot be applied. In Proposition~\ref{prop:dualcert}, I therefore establish a new strong duality result. Its
proof uses the linear-programming approach in
\cite{kleiner2019strong}, but relies on a different repair
lemma, Lemma~\ref{lem:R-repair}, which shows that a convex utility profile
that approximately satisfies \eqref{eq:pseudomono} can be uniformly
approximated by one that satisfies it exactly.

The resulting condition in
Theorem~\ref{thm:rent-measure-characterization} also relates to the
stochastic-dominance certificates developed by
\cite{daskalakis2013mechanism,daskalakis} for the multi-good monopoly
problem. In particular, \cite{daskalakis2013mechanism} provide a dominance
condition for the optimality of grand bundling that is phrased in terms of a
signed measure similar to mine. In both of our approaches the objective is rewritten as an integral against a signed measure, and the optimality of an extremal indirect-utility profile is certified
through a stochastic-dominance comparison. However, several features of my environment require a different construction. First, in my problem, types live on a simplex, and
the planner maximizes weighted welfare rather than revenue. Second, feasibility is
governed by aggregate supply constraints,
so the signed measure must also incorporate the supply costs. Most
importantly, the source of extremality is different. In
\cite{daskalakis2013mechanism}, it comes from unit caps on allocations. Here,
it comes from \eqref{eq:pseudomono}, which bounds how quickly
\(U(\theta)\) can grow as \(\theta\) approaches a vertex. This is why the
objective representation in Lemma~\ref{lemma:measure-objective} involves the
transformed term \(U(\theta)/\theta_*(\theta)\), rather than \(U(\theta)\)
alone, and why the resulting transport comparison uses the problem-specific
orders \(\succ_i\).

\appendix

\titleformat{\subsection}[runin]
  {\normalfont\bfseries}
  {\thesubsection.}
  {0.5em}
  {}
  [.]
\titlespacing*{\subsection}
  {0pt}
  {1ex plus .2ex minus .1ex}
  {0.75em}

\setstretch{1.1} 

\section{Waitlist interpretation}\label{app:waitlist}

I now show that the baseline model from Section \ref{sec:model} can be interpreted as a reduced-form description of a steady-state waitlist setting.

\paragraph{Waitlist model.} Consider a stationary waitlist environment in discrete time. There are $N$ different kinds of goods indexed by $i\in \{1,\dots,N\}$ with $N\geq 2$. The designer possesses a fixed mass of each, with the supplies given by $s=(s_1,s_2,\dots,s_N)> 0$. In each period, a unit mass of new agents arrives; each of them has a profile of values $v= (v_1,v_2,\dots,v_N)$ for the goods. The values are private information and, in each arriving cohort, are distributed according to $F$ with support on a bounded set $\mathcal{V}\subset \mathbb{R}^N_{+}$. This distribution puts no mass on \(\mathbf 0\) and assigns positive mass to agents with \(v_i>0\) for every good \(i\). Agents remain in the system from one period to the next with probability
\(\delta\in[0,1)\). The departure process is exogenous and independent of
the agent's type, report, and allocation history.

Agents may enjoy a good in every period after entering the system.
Let \(\mathcal A:=\{\varnothing,1,\dots,N\}\) denote the set of possible 
one-period assignments for an agent, where \(\varnothing\)
denotes receiving no good. A contingent lifetime assignment path is an
element \(a=(a_0,a_1,a_2,\dots)\in\mathcal A^{\mathbb N_0}\), where \(a_t\)
is the good enjoyed in the \(t\)-th period after entry, conditional on the
agent still being in the system.

Since the agent survives through period \(t\) with probability \(\delta^t\),
her expected lifetime utility from the path \(a\) is
\[
\sum_{t=0}^{\infty}\delta^t v_{a_t},
\]
with the convention \(v_{\varnothing}:=0\). If the contingent assignment
path is random, with distribution
\(Z\in\Delta(\mathcal A^{\mathbb N_0})\), her expected lifetime utility is
\[
\mathbb E_{a\sim Z}\left[
\sum_{t=0}^{\infty}\delta^t v_{a_t}
\right].
\]

\paragraph{Mechanisms.} We restrict attention to stationary mechanisms that admit a steady state. That is, the same mechanism is applied to every entering cohort, and the induced cross-sectional distribution of active agents, assignment histories, and occupied goods is constant over time. Using stationarity and the Revelation Principle, we can assume that the designer selects a stationary allocation rule $z:\mathcal V\to \Delta(\mathcal A^{\mathbb N_0})$. She chooses it to maximize steady-state welfare
\begin{equation}\label{eq:SSwelf}
    \max_z
    \int
    \mathbb E_{a\sim z(v)}\left[
        \sum_{t=0}^{\infty}\delta^t v_{a_t}
    \right]
    dF(v),\tag{O''}
\end{equation}
subject to incentive compatibility and supply constraints:
\begin{equation}\label{eq:ICww}
    \mathbb E_{a\sim z(v)}\left[
        \sum_{t=0}^{\infty}\delta^t v_{a_t}
    \right]
    \geq
    \mathbb E_{a\sim z(v')}\left[
        \sum_{t=0}^{\infty}\delta^t v_{a_t}
    \right]
    \quad \text{for all }v,v',\tag{IC''}
\end{equation}
\begin{equation}\label{eq:Sww}
    \sum_{t=0}^{\infty}
    \delta^t
    \int
    \mathbb E_{a\sim z(v)}\left[
        \mathbf 1\{a_t=i\}
    \right]
    dF(v)
    \leq s_i
    \quad \text{for every }i.\tag{S''}
\end{equation}
The supply constraint says that, in the steady-state period, the total mass of agents assigned good \(i\) cannot exceed the stock \(s_i\). To see why it takes this form, note that in any calendar period there is a unit mass of newly arrived agents, a mass \(\delta\) of agents who entered one period earlier, a mass \(\delta^2\) of agents who entered two periods earlier, and so on. The sum therefore represents the steady-state use of good \(i\).

\paragraph{Reduction to the static model.} The following result reduces the setting to the baseline model.

\begin{proposition}\label{prop:waitlist-reduction}
For any feasible stationary allocation rule
\(z:\mathcal V\to\Delta(\mathcal A^{\mathbb N_0})\), the reduced-form
allocation rule \(y:\mathcal V\to\mathbb R_+^N\) given by
\begin{equation}\label{eq:waitlist-reduced-form}
y_i(v)
=
\mathbb E_{a\sim z(v)}\left[
\sum_{t=0}^{\infty}
\delta^t\mathbf 1\{a_t=i\}
\right]
\quad\text{for every \(i\) and \(v\)}
\end{equation}
is feasible in the baseline model. Moreover, welfare from \(z\) equals
welfare from \(y\):
\begin{equation}\label{eq:waitlist-welfare-equality}
\int_{\mathcal V}
\mathbb E_{a\sim z(v)}\left[
\sum_{t=0}^{\infty}\delta^t v_{a_t}
\right]dF(v)
=
\int_{\mathcal V}v\cdot y(v)\,dF(v).
\end{equation}
Conversely, for any feasible \(y\) in the baseline model satisfying
\begin{equation}\label{eq:waitlist-individual-feasibility}
\sum_{i=1}^N y_i(v)
\leq
\frac{1}{1-\delta}
\quad\text{for every }v,
\end{equation}
the stationary allocation rule that permanently assigns an agent reporting
\(v\) to good \(i\) with probability \((1-\delta)y_i(v)\), and to no good
with the remaining probability, is feasible in the dynamic problem. Its
reduced-form allocation is \(y\), and the two allocation rules satisfy
\eqref{eq:waitlist-welfare-equality}.
\end{proposition}

Thus, the dynamic problem is equivalent to the baseline problem augmented
with \eqref{eq:waitlist-individual-feasibility}. This additional bound is
analogous to the probability constraint discussed in Subsection~\ref{rem:model-dis} and thus, by
Proposition~\ref{prop:P_redundant}, is redundant when supply is sufficiently
scarce. In the waitlist setting, it suffices that supply be too scarce
for any feasible mechanism to offer any option with zero wait.

\begin{proof}
\((\Rightarrow)\) Fix a feasible stationary allocation rule \(z\) and define
\(y\) by \eqref{eq:waitlist-reduced-form}. Since an agent can receive at most
one good in each period,
\[
\sum_{i=1}^N y_i(v)
=
\mathbb E_{a\sim z(v)}\left[
\sum_{t=0}^{\infty}\delta^t
\sum_{i=1}^N\mathbf 1\{a_t=i\}
\right] 
\leq
\sum_{t=0}^{\infty}\delta^t
=
\frac{1}{1-\delta},
\]
which gives \eqref{eq:waitlist-individual-feasibility}. Moreover, for every
true type \(v\) and report \(v'\),
\[
\mathbb E_{a\sim z(v')}\left[
\sum_{t=0}^{\infty}\delta^t v_{a_t}
\right]
=
\sum_{i=1}^N v_i y_i(v')
=
v\cdot y(v').
\]
Thus, incentive compatibility of \(z\) implies \eqref{eq:IC}, and integrating
the preceding equality at \(v'=v\) gives
\eqref{eq:waitlist-welfare-equality}. Finally, for every good \(i\),
\[
\sum_{t=0}^{\infty}\delta^t
\int_{\mathcal V}
\mathbb E_{a\sim z(v)}
\left[\mathbf 1\{a_t=i\}\right]dF(v)
=
\int_{\mathcal V}y_i(v)\,dF(v),
\]
so \eqref{eq:Sww} implies \eqref{eq:S}. Hence, \(y\) is feasible in the
baseline model.

\((\Leftarrow)\) Conversely, fix a feasible \(y\) satisfying
\eqref{eq:waitlist-individual-feasibility}. Upon entry, permanently assign an
agent reporting \(v\) to good \(i\) with probability
\((1-\delta)y_i(v)\), and to no good with the remaining probability.
Equation~\eqref{eq:waitlist-individual-feasibility} ensures that these
probabilities are well defined. The resulting reduced-form allocation
satisfies
\[
\mathbb E_{a\sim z(v)}\left[
\sum_{t=0}^{\infty}\delta^t\mathbf 1\{a_t=i\}
\right]
=
(1-\delta)y_i(v)\sum_{t=0}^{\infty}\delta^t
=
y_i(v).
\]
The identities established above therefore imply that \(z\) is feasible in
the dynamic problem and that \(z\) and \(y\) satisfy
\eqref{eq:waitlist-welfare-equality}.
\end{proof}

\begin{remark}
Under this reduction, distinct waitlist and lottery mechanisms can yield equivalent reduced-form allocations, in the spirit of \citet{arnosti2020design}. Consider, for example, a good-specific waitlist, in which each agent joins the waitlist for one good and receives it once she reaches the front, or a repeated choice-based lottery, in which each waiting agent chooses one good-specific lottery to enter in each period and receives that good if she wins. When supply is sufficiently scarce, both mechanisms implement reduced-form allocations that correspond to the pure option mechanism in the baseline model.
\end{remark}

\section{Omitted proofs: technical preliminaries}

\subsection{Strassen's theorem}\label{sec:Strassen}

\begin{definition}
Let $\succeq$ be a partial order on $\Omega$. A set $C \subseteq \Omega$ is an $\succeq$-\textbf{upper set} if $\theta\in C, \ \theta\preceq \theta'$ implies $\theta'\in C$. A function $\eta:\Omega\to \mathbb{R}$ is $\succeq$-\textbf{increasing} if $\theta\preceq \theta'$ implies $\eta(\theta')\ge \eta(\theta)$.
\end{definition}

The following is a special case of Strassen's theorem stated in \cite{fritz2018antisymmetry}:

\begin{theorem}[\cite{strassen1965existence, kellerer1984duality, edwards1978existence}]\label{thm:stoch_order_i}
Let $\rho,\tau$ be measures on some $\Omega \subset \mathbb{R}^N$ with $\rho(\Omega)=\tau(\Omega)$ and let 
$\succeq$ be a partial order on $\Omega$ such that the set
$
\left\{ (x,y) \in \Omega\times \Omega : \ x\succeq y \right\}
$
is closed in $\Omega\times \Omega$. Then $\tau$ $\succeq$-\textbf{stochastically dominates} $\rho$ if and only if any of the following conditions holds: 
\begin{enumerate}
\item $\rho(C)\le \tau(C)$ for every closed $\succeq$-upper set $C\subseteq \Omega$.
\item For every bounded, lower semicontinuous, $\succeq$-increasing $\eta:\Omega\to\mathbb{R}$, $\int_\Omega \eta\,d\rho \;\le\; \int_\Omega \eta\,d\tau.$
\item There exists a $\succeq$-monotone transport plan from $\rho$ to $\tau$.
\end{enumerate}
\end{theorem}

\subsection{Differential geometry facts}

Let $H$ denote the $(N-1)$-dimensional hyperplane containing the simplex $\Gamma$:
\[
H := \big\{ \theta \in \mathbb{R}^N: \ \sum \theta_i =1 \big\}.
\]
Note that the tangent space to $H$ at any $\theta\in H$ is
\[
T H := \bigl\{v\in\mathbb{R}^N \colon \sum v_i=0\bigr\}.
\]
Let us also define the intrinsic gradient for this surface:
\begin{definition}\label{def:intrinsic-gradient}
Let $\eta:H\to\mathbb{R}$ and fix $\theta\in H$. The \textbf{intrinsic
gradient} $\nabla_H \eta(\theta)\in TH$ is the unique vector such that
\[
D_v \eta(\theta)=\nabla_H \eta(\theta)\cdot v
\qquad\text{for all }v\in TH.
\]
\end{definition}
I use the following divergence theorem on the affine hyperplane $H$,
which follows from Green's formula in $\mathbb{R}^{N-1}$
(see, e.g., \cite{rodrigues1987obstacle}).
\begin{theorem}\label{th:divonH}
Let $\Omega\subset H$ be a bounded, open set such that $\partial\Omega$ is Lipschitz. Let $\eta:\overline{\Omega}\to\mathbb{R}$ be Lipschitz. Fix a tangent vector field $X:\Omega\to\mathbb{R}^N$, $X(\theta) \in T H$, such that $X \in H^{1}\big(\overline \Omega;\,TH\big)$. Then
\begin{equation}\label{eq:intbypartsdiv}
\int_{\Omega} \nabla_H \eta(\theta)\cdot X(\theta)\,dV_H(\theta)
+
\int_{\Omega} \eta(\theta)\,\operatorname{div} X(\theta)\,dV_H(\theta)
=
\int_{\partial\Omega} \eta(\theta)\,X(\theta)\cdot\nu(\theta) \, dS_{\partial\Omega}(\theta),
\end{equation}
where $dV_H$ denotes the $(N-1)$--dimensional surface measure on $H$, $dS_{\partial\Omega}$ denotes the $(N-2)$--dimensional surface measure on $\partial\Omega$, and $\nu$ is the outward unit conormal along $\partial\Omega$. Finally, $\operatorname{div} X(\theta)$ is the divergence taken in the $(N-1)$-dimensional subsurface $H$.
\end{theorem}
Throughout, boundary and interface integrals involving $g$ are
understood using its Sobolev trace on the relevant surface.

\subsection{General lemmas}

\begin{lemma}
\label{lem:pure_option_derivatives}
For \(q'\in\mathbb R_{++}^N\), let
\(U_{q'}(\theta):=\max_i\theta_iq_i'\). Then, at \(q\), for every
\(h\in\mathbb R^N\),
\[
\left.\frac{d}{dt}\right|_{t=0}
\int_\Gamma\lambda U_{q+th}g\,d\theta
=
A\cdot h,
\quad \text{and, for every \(j\),} \quad
\left.\frac{d}{dt}\right|_{t=0}
\int_\Gamma(x_{U_{q+th}})_j g\,d\theta
=
\sum_iJ_{ij}h_i.
\]
\end{lemma}

\begin{proof}
Since indifference boundaries have \(G\)-measure zero, Danskin's theorem and
dominated convergence give
\[
\left.\frac{d}{dt}\right|_{t=0}
\int_\Gamma\lambda U_{q+th}g\,d\theta
=
\sum_i h_i\int_{\Gamma_i}\theta_i\lambda g\,d\theta
=
A\cdot h.
\]
For the latter derivative, note that $\int_\Gamma(x_{U_{q'}})_j g\,d\theta
=
q_j'G(\Gamma_j(q')).$ The direct effect of changing \(q_j\) is \(M_j\). Along
\(\Gamma_i\cap\Gamma_j\), the separating surface is
\(q_i\theta_i=q_j\theta_j\), whose intrinsic normal has norm
\[
\left\|\nabla_H(q_i\theta_i-q_j\theta_j)\right\|
=
\sqrt{q_i^2+q_j^2-\tfrac1N(q_i-q_j)^2}.
\]
The boundary-movement formula therefore gives
\[
\frac{\partial}{\partial q_j}
\int_\Gamma(x_{U_{q'}})_j g\,d\theta
=
M_j+q_j\sum_{i\neq j}T_{ji}
=
J_{jj},
\quad \text{and, for \(i\neq j\),} \quad
\frac{\partial}{\partial q_i}
\int_\Gamma(x_{U_{q'}})_j g\,d\theta
=
-q_jT_{ij}
=
J_{ij}.
\]
\end{proof}

\section{Omitted proofs from Sections \ref{sec:model}---\ref{sec:CEEI}}

\subsection{Proof of Proposition \ref{prop:P_redundant}}
Let $M:=\sup\{v_j:v\in\mathcal V,\ j=1,\dots,N\}<\infty.$ For every \(i\), choose \(\delta_i>0\) such that $Z_i:=F(\{v\in\mathcal V:v_i\ge\delta_i\})>0,$ and set
\[
\bar\eta:=\min_i\frac{\delta_iZ_i}{4MN\sum_j s_j}.
\]
Fix \(\eta\leq\bar\eta\) and an allocation rule \(y\) feasible with supplies
\(\eta s\). For each \(i\), Markov's inequality gives
\[
F\Big(\Big\{v:\sum_jy_j(v)>
\frac{2\eta\sum_j s_j}{Z_i}\Big\}\Big)
\le\frac{Z_i}{2}.
\]
Consequently, the set
\[
S_i:=
\Big\{
v:v_i\ge\delta_i,\ 
\sum_jy_j(v)\le\frac{2\eta\sum_j s_j}{Z_i}
\Big\}
\]
has positive mass. Note also that every \(v\in S_i\) satisfies
\[
v\cdot y(v)
\le
M\sum_jy_j(v)
\le
\frac{2M\eta\sum_j s_j}{Z_i}
\le
\frac{\delta_i}{2N}.
\]
If \(y_i(v')\ge1/N\) for some \(v'\), every \(v\in S_i\) would strictly prefer
reporting \(v'\), since
\[
v\cdot y(v')
\ge v_i y_i(v')
\ge\frac{\delta_i}{N}
>
v\cdot y(v),
\]
contradicting incentive compatibility. Thus \(y_i(v)<1/N\) for every \(i\) and
\(v\), and consequently \(\sum_i y_i(v)<1\) for every \(v\). Hence
\eqref{eq:P} is slack.

\subsection{Proof of Lemma \ref{lem:renorm-wlog}}

Let \(y:\mathcal V\to\mathbb R_+^N\) be feasible in the original problem.
By the boundedness argument above, \(y\) is uniformly bounded. Define
\[
K:=\overline{\operatorname{co}}\{y(v):v\in\mathcal V\},
\qquad
\tilde U(\theta):=\max_{z\in K}\theta\cdot z,
\]
and let \(\bar x\) be a measurable version of
\(\mathbb E[y(V)\mid\Theta=\theta]\). Since \(K\) is compact, \(\tilde U\) is continuous. Moreover, because
\(z\mapsto\theta\cdot z\) is linear and continuous, passing to the convex hull
and its closure does not change the supremum, and hence, $\tilde U(\theta)
=
\sup_{v'\in\mathcal V}\theta\cdot y(v').$ Now, since \(y(V)\in K\) a.s. and \(K\) is closed and convex, its
conditional expectation satisfies
\(\bar x(\theta)=\mathbb E[y(V)\mid\Theta=\theta]\in K\) for a.e. \(\theta\).

Moreover, incentive compatibility implies that, for every realized
\(V=(\sum_iV_i)\Theta\) and every \(v'\in\mathcal V\), we have $\Theta\cdot y(V)\ge \Theta\cdot y(v').$ Hence $\Theta\cdot y(V)
=
\sup_{v'\in\mathcal V}\Theta\cdot y(v')
=
\tilde U(\Theta)$ a.s. Taking conditional expectations given \(\Theta=\theta\) therefore gives $\theta\cdot\bar x(\theta)
=
\mathbb E[\Theta\cdot y(V)\mid\Theta=\theta]
=
\tilde U(\theta)$ for a.e. $\theta$. Let us then take a Borel set of full \(G\)-measure \(B\subseteq\Gamma\) such that $\bar x(\theta)\in K$ and $\theta\cdot\bar x(\theta)=\tilde U(\theta)$ for every \(\theta\in B\), and make a measurable selection
\[
z^*(\theta)\in\argmax_{z\in K}\theta\cdot z
\quad \text{and define} \quad
x(\theta)
:=
\begin{cases}
\bar x(\theta),&\theta\in B,\\
z^*(\theta),&\theta\notin B.
\end{cases}
\]
Then, for every \(\theta\in\Gamma\), we have $x(\theta)\in K$ and $\theta\cdot x(\theta)=\tilde U(\theta)$. Therefore, for every \(\theta,\theta'\in\Gamma\),
\[
\theta\cdot x(\theta')
\le
\max_{z\in K}\theta\cdot z
=
\tilde U(\theta)
=
\theta\cdot x(\theta),
\]
so \(x\) is incentive compatible. Now, since \(x=\bar x\) \(G\)-almost everywhere, the tower property gives
\[
\int_\Gamma x(\theta)\,dG(\theta)
=
\mathbb E\!\left[\mathbb E[y(V)\mid\Theta]\right]
=
\int_{\mathcal V}y(v)\,dF(v)
\le s.
\]
Thus \(x\) is feasible in Problem~\ref{prob:renormalized}. Moreover, since
\(V=(\sum_iV_i)\Theta\),
\[
\int_{\mathcal V}v\cdot y(v)\,dF(v)
=
\mathbb E\big[\big(\sum V_i\big)\tilde U(\Theta)\big]
=
\int_\Gamma\lambda(\theta)\tilde U(\theta)\,dG(\theta)
=
\int_\Gamma
\lambda(\theta)\,\theta\cdot x(\theta)\,dG(\theta).
\]
Conversely, let \(x\) be feasible in Problem~\ref{prob:renormalized}, and
define $y(\mathbf 0):=0$ and $y(v):=x({v}/{\sum v_i})$ for $v\neq\mathbf 0.$ The supply constraint follows from the definition of \(G\). For
\(v,v'\neq\mathbf 0\), \eqref{eq:ICprime} gives
\[
v\cdot y(v)
=
\left(\sum v_i\right)
\frac{v}{\sum v_i}\cdot
x\left(\frac{v}{\sum v_i}\right)
\ge
\left(\sum v_i\right)
\frac{v}{\sum v_i}\cdot
x\left(\frac{v'}{\sum v_i'}\right)
=
v\cdot y(v').
\]
The cases involving \(\mathbf 0\) are immediate. Finally, because \(F\)
puts no mass on \(\mathbf 0\), the welfare equality follows from the
definition of \(\lambda\).

\subsection{Proof of Fact \ref{fact:puremechs}}

For \(q\in\mathbb R_{++}^N\), let
\[
x^q(\theta)\in
\argmax_{z\in\{q_1e_1,\dots,q_Ne_N\}}\theta\cdot z.
\]
Write \(y_i:=\log(1/q_i)\). Choosing option \(i\) is then equivalent to
maximizing \(\log\theta_i-y_i\). Define
\[
\Gamma_i(y)
:=
\left\{
\theta\in\Gamma:
\log\theta_i-y_i\ge\log\theta_j-y_j
\text{ for every }j
\right\},
\qquad
m_i(y):=G(\Gamma_i(y)).
\]
Aggregate demand for good \(i\) is \(e^{-y_i}m_i(y)\), so market clearing is
equivalent to
\begin{equation}\label{eq:marketcleariff}
m_i(y)=s_i e^{y_i}
\qquad\text{for every }i.
\end{equation}

Consider the potential
\begin{equation}\label{eq:Psi_def}
\Psi(y)
:=
\int_\Gamma\max_j\{\log\theta_j-y_j\}\,dG(\theta)
+
\sum_{j=1}^N s_j e^{y_j},
\end{equation}
with the convention \(\log0=-\infty\). For every \(y\), indifference between
options \(i\neq j\) occurs only on $\{\theta\in\Gamma:\theta_i=e^{y_i-y_j}\theta_j\},$ which is \(G\)-null. Hence, the maximizer is
unique \(G\)-almost everywhere. Since the integrand is \(1\)-Lipschitz in
each coordinate of \(y\), Danskin's theorem and dominated convergence give $\frac{\partial\Psi}{\partial y_i}(y)
=
-m_i(y)+s_i e^{y_i}.$ Thus, by \eqref{eq:marketcleariff}, the market-clearing vectors are precisely the critical points of \(\Psi\). It then suffices to show that \(\Psi\) has a unique critical point. I do so by proving that \(\Psi\) is strictly convex and attains a minimum.

Strict convexity follows because the first term in \(\Psi\) is convex and \(\sum_j s_je^{y_j}\) is strictly convex because \(s_j>0\) for every \(j\). Consequently, if a minimum is attained, it is unique. I now show it is attained. To that end, let $\overline y:=\max_i y_i,$ $\underline y:=\min_i y_i,$ and $\underline s:=\min_i s_i>0.$ Because \(\max_i\theta_i\ge 1/N\), we have $\Psi(y)
\ge
-\log N-\overline y+\underline s e^{\overline y},$
which tends to \(+\infty\) as \(\overline y\to+\infty\). If instead
\(\overline y\) remains bounded while \(\|y\|\to\infty\), then, along a
subsequence, \(y_i=\underline y\to-\infty\) for some fixed \(i\). Since
\(G(\{\theta_i>0\})=1\), there exists \(\delta>0\) such that $Z:=G(\{\theta_i>\delta\})>0.$ Therefore
\[
\int_\Gamma\max_j\{\log\theta_j-y_j\}\,dG
\ge
Z(\log\delta-\underline y)
+
(1-Z)(-\log N-\overline y)
\to+\infty.
\]
Thus, \(\Psi\) is coercive and attains a unique minimum.

\subsection{Proof of Fact \ref{fact:pure_implementations}}

Let $(p,x)$ be a CEEI. First, note $p_i>0$ for every $i$: if $p_i=0$, then every type with $\theta_i>0$ would demand arbitrarily large amounts of good $i$, contradicting feasibility. Now, define $q_i:=1/{p_i}$. For any bundle $z\ge 0$ satisfying $p\cdot z\le 1$,
\[
\theta\cdot z
=\textstyle\sum_{j} \theta_j z_j
\le
\Big(\max_j \frac{\theta_j}{p_j}\Big)\sum_{j=1}^N p_j z_j
\le
\max_j \frac{\theta_j}{p_j}
=
\max_j \theta_j q_j.
\]
This upper bound is attained by any pure option $q_i e_i$ with $i\in\argmax_j \theta_j q_j$. Thus, for every type, at least one pure option is utility-maximizing in the CEEI demand problem, and for all but a null set of types this pure option is unique. Since the CEEI allocation clears the market, the vector $q=(q_1,\dots,q_N)$ is a market-clearing vector of pure options. By Fact \ref{fact:puremechs}, this vector is uniquely pinned down, and the induced allocation coincides with the pure option mechanism up to tie-breaking among a null set of types. The argument for the representative endowment economy is analogous. 

Now, consider the choice-based lottery. In any pure-strategy Nash equilibrium, the mass $m_i$ of agents choosing each good $i$ must satisfy $m_i>0$; otherwise any agent with $\theta_i>0$ would deviate to good $i$ and receive an infinite allocation. Each agent who chose good $i$ receives $q_i:=s_i/m_i$ of it. In equilibrium, every type $\theta$ must be choosing a good $i$ that maximizes $\theta_i q_i$; otherwise, switching to a good $j$ with $\theta_j q_j>\theta_i q_i$ would be profitable, since the mass of a single agent is zero and such a deviation does not change $m_j$. Hence the equilibrium allocation coincides with the pure option mechanism with quantities $q=(q_1,\dots,q_N)$: by Fact~\ref{fact:puremechs}, this vector is uniquely pinned down, and the induced allocation assigns $q_i e_i$ to almost every type $\theta\in\Gamma_i$.

Conversely, let $q=(q_1,\dots,q_N)$ be the market-clearing vector of pure options. Setting $p_i=1/q_i$ implements the pure option mechanism as a CEEI. Since market clearing gives $\sum_i s_i/q_i=1$, the same prices also implement it as a Walrasian equilibrium of the representative endowment economy. Finally, in the choice-based lottery, if each type chooses a good maximizing $\theta_i q_i$, then $m_i=s_i/q_i$ and hence $s_i/m_i=q_i$. No type can profitably deviate, so this action profile is a pure-strategy Nash equilibrium.

\subsection{Proof of Fact \ref{fact:supply_cost_properties}}

Let \(D:=\operatorname{diag}(q_1,\ldots,q_N)\). A direct calculation gives $J=HD,$ where
\[
H_{ii}=\frac{M_i}{q_i}+\sum_{j\neq i}T_{ij},
\qquad
H_{ij}=-T_{ij}\leq 0
\quad (i\neq j).
\]
For every row \(i\),
\[
H_{ii}-\sum_{j\neq i}|H_{ij}|
=
\frac{M_i}{q_i}
>
0.
\]
Thus, \(H\) is a strictly row-diagonally dominant \(Z\)-matrix with
positive diagonal entries. Hence \(H\) is a nonsingular \(M\)-matrix, so
\(H^{-1}\geq 0\) entrywise. Since \(D\) is also invertible, \(J=HD\) is
invertible and $c=J^{-1}A=D^{-1}H^{-1}A.$ Moreover, \(H^{-1}A\) is strictly positive. Indeed, \(H^{-1}\) is
nonnegative and invertible, so each of its rows contains at least one
strictly positive entry. Since \(A\) is strictly positive, it follows that $H^{-1}A\gg 0.$ Finally, \(D^{-1}\) has strictly positive diagonal entries, and therefore $c=D^{-1}H^{-1}A\gg0.$

\subsection{Proof of Proposition \ref{prop:exchcondition}}

We verify the conditions of
Theorem~\ref{thm:rent-measure-characterization} by taking
\(\tilde\mu_i:=\mu_i:=\mu\restriction_{\Gamma_i}\) for every \(i\).
The rent measure \(\mu\) assigns zero mass to the indifference boundaries
between distinct pure-option regions. Moreover,
\(\theta_*=\theta_i\) \(\mu\)-almost everywhere on \(\Gamma_i\). Hence, for
every continuous, convex \(U:\Gamma\to\mathbb R_+\),
\[
\sum_{i=1}^N
\int_{\Gamma_i}\frac{U}{\theta_i}\,d\tilde\mu_i
=
\int_\Gamma\frac{U}{\theta_*}\,d\mu,
\]
so preprocessing holds with equality. Moreover, by the balance condition
\eqref{eq:gammai-balance}, \(\mu_i(\Gamma_i)=0\), so
\(\mu_i^+(\Gamma_i)=\mu_i^-(\Gamma_i)\). Since the order relation
\(\succ_i\) is closed on \(\Gamma_i\),
Theorem~\ref{thm:stoch_order_i} implies that it is enough to show
\(\mu_i(C)\ge0\) for every closed \(\succ_i\)-upper set
\(C\subseteq\Gamma_i\).

Under exchangeability and equal supplies, the pure option mechanism is symmetric, so \(q_1=...=q_N\), \(\theta^0=\frac1N\mathbf 1\), and all \(A_i\) are equal: \(A_i =:\bar A\). Moreover, the supply costs satisfy \(c=N\bar A\,\mathbf 1\) and \(\sum_j c_j=N^2\bar A\). Plugging this into \eqref{eq:signedmeas} gives, for any Borel set \(\Omega\subseteq\Gamma_i\),
\begin{equation}\label{eq:symmmeasure_new}
\mu_i(\Omega)
=
\int_{\Omega} \lambda\,\theta_i\, g \,d\theta
-N^2\bar A\int_{\Omega}\theta_i\big[\operatorname{div}\big((\theta-\theta^0) g\big)+ g \big]\,d\theta
+N^2\bar A\int_{\Omega\cap \partial \Gamma }\theta_i\,  g \,
(\theta-\theta^0)\cdot\nu\  d\sigma.
\end{equation}

I first show that for every \(\succ_i\)-upper set \(C\subseteq \Gamma_i\), we have
\begin{equation}\label{eq:upper_mean_ineq_final}
\int_C \lambda\,\theta_i\,g\,d\theta \;\ge\; N\bar A \int_C g\,d\theta.
\end{equation}
Indeed, since \(q_1=...=q_N\), the set \(\Gamma_i\) coincides up to a null set with \(\{\Theta_i\ge \Theta_j\ \forall j\neq i\}\). Also, by definition of \(\succ_i\), every \(\succ_i\)-upper set \(C\subseteq\Gamma_i\) can be written as \(C=\{\theta\in\Gamma_i:\ R_i(\theta)\in B\}\) for some \(\ge\)-lower set \(B\subseteq \mathbb R_+^N\). Since \(B^c\) is then \(\ge\)-upper, the assumed stochastic monotonicity implies that for every \(t\ge 0\),
\[
\mathbb P \left[\Theta\in C \,\middle|\, \lambda(\Theta)\Theta_i\ge t,\ \Theta\in\Gamma_i\right]
\ge
\mathbb P \left[\Theta\in C \,\middle|\, \Theta\in\Gamma_i\right].
\]
Integrating over \(t\) yields
\[
\mathbb E \left[\lambda(\Theta)\Theta_i\,\mathbf 1_{\{\Theta\in C\}} \,\middle|\, \Theta\in\Gamma_i\right]
\ge
\mathbb P \left[\Theta\in C \,\middle|\, \Theta\in\Gamma_i\right]
\mathbb E \left[\lambda(\Theta)\Theta_i \,\middle|\, \Theta\in\Gamma_i\right].
\]
Equivalently,
\[
\int_C \lambda\,\theta_i\,g\,d\theta
\ge
\frac{\int_{\Gamma_i}\lambda\,\theta_i\,g\,d\theta}{\int_{\Gamma_i}g\,d\theta}\int_C g\,d\theta.
\]
Under exchangeability, \(\int_{\Gamma_i}g\,d\theta=\frac1N\) and \(\int_{\Gamma_i}\lambda\,\theta_i\,g\,d\theta=\bar A\), giving \eqref{eq:upper_mean_ineq_final}.

Next, let \(C\subseteq\Gamma_i\) be a \(\succ_i\)-upper set with Lipschitz boundary. Applying Theorem \ref{th:divonH} to the tangent field \((\theta-\theta^0)g\) on \(C\), and using \(\nabla_H\theta_i=e_i-\frac1N\mathbf 1\), we obtain
\[
\int_C \theta_i\Big[\operatorname{div}\big((\theta-\theta^0)g\big)+g\Big]\,d\theta
=
\int_{\partial C}\theta_i\,(\theta-\theta^0)\cdot \nu_C\, g\,d\sigma
+\frac1N\int_C g\,d\theta,
\]
where \(\nu_C\) is the outward unit conormal to \(\partial C\). Substituting this into \eqref{eq:symmmeasure_new} gives
\[
\mu_i(C)
=
\int_C \lambda\,\theta_i\,g\,d\theta
-
N\bar A\int_C g\,d\theta
-
N^2\bar A\int_{\partial C\cap \Gamma^\circ}\theta_i\,(\theta-\theta^0)\cdot \nu_C\, g\,d\sigma.
\]
By \eqref{eq:upper_mean_ineq_final}, the sum of the first two terms is weakly positive. Since \(N^2\bar A>0\) and \(\theta_i g\geq 0\), it therefore suffices to show that $(\theta-\theta^0)\cdot \nu_C\leq 0$ for a.e. \(\theta\in\partial C\cap\Gamma^\circ\). Fix such a boundary point \(\theta\) and, for small \(t>0\), let
\[
\theta^t:=\theta+t(\theta-\theta^0)
=\theta^0+(1+t)(\theta-\theta^0).
\]
Since \(\theta\in\Gamma^\circ\), we have \(\theta^t\in\Gamma^\circ\) for all sufficiently small \(t>0\). Moreover, for every \(k\neq i\),
\[
\frac{\theta_k^t}{\theta_i^t}
=
\frac{\theta_k+t(\theta_k-\theta_k^0)}
{\theta_i+t(\theta_i-\theta_i^0)}
\leq
\frac{\theta_k}{\theta_i}
\quad\Longleftrightarrow\quad
\frac{\theta_k^0}{\theta_i^0}
\geq
\frac{\theta_k}{\theta_i}.
\]
The last inequality holds because \(\theta^0=\frac1N\mathbf 1\) and
\(\theta\in\Gamma_i\), so \(\theta_k/\theta_i\leq1\). Thus
\(\theta^t\succ_i\theta\), and the same inequalities imply that
\(\theta^t\in\Gamma_i\). Since \(C\) is closed and \(\succ_i\)-upper,
\(\theta\in C\) and hence \(\theta^t\in C\). Therefore
\(\theta-\theta^0\) points weakly inward at \(\theta\), which implies $(\theta-\theta^0)\cdot\nu_C\leq0$
wherever the outward conormal exists. Since \(\partial C\) is Lipschitz, it
exists almost everywhere. The boundary integral is therefore nonpositive,
and hence \(\mu_i(C)\geq0\).

I now extend the argument to all closed $\succ_i$-upper sets using the following lemma:

\begin{lemma}\label{lem:approx_upper_simple}
Fix \(i\). Let \(C\subseteq \Gamma_i\) be a closed \(\succ_i\)-upper set. Then there exists a decreasing sequence \((K_m)_{m\ge 1}\) of closed \(\succ_i\)-upper sets such that \(K_{m+1}\subseteq K_m\), \(\bigcap_{m\ge 1}K_m=C\), and each \(K_m\) is a finite union of polytopes in \(H\) defined by finitely many inequalities of the form \(\theta_k\le a\,\theta_i\) for \(k\neq i\).
\end{lemma}

\begin{proof}
Define
\[
Q_i:\Gamma_i\to\mathbb R_+^{N-1}, \qquad Q_i(\theta):=\Bigl(\frac{\theta_k}{\theta_i}\Bigr)_{k\neq i}.
\]
This map is injective on \(\Gamma_i\), and \(\theta'\succ_i\theta\) holds if and only if \(Q_i(\theta')\le Q_i(\theta)\) coordinatewise. Therefore \(C\subseteq\Gamma_i\) is \(\succ_i\)-upper if and only if \(Q_i(C)\) is a \(\ge\)-lower subset of \(\mathbb R_+^{N-1}\). Since \(C\) is closed and \(Q_i\) is continuous, \(Q_i(C)\) is compact.

Fix \(m\ge 1\). Let
\[
D_m:=\bigcup\Bigl\{[0,b+\tfrac1m\mathbf 1]:
b\in \tfrac1m\mathbb Z_+^{N-1},\
[0,b]\subseteq Q_i(C)\Bigr\},
\]
where \([0,b]:=\{r\in\mathbb R_+^{N-1}:0\le r\le b\}\). Then \(D_m\) is a closed \(\ge\)-lower set and a finite union of rectangles. Moreover, \(Q_i(C)\subseteq D_m\). To see this, fix \(r\in Q_i(C)\). Since \(Q_i(C)\) is lower, \([0,r]\subseteq Q_i(C)\). Let \(b\in \tfrac1m\mathbb Z_+^{N-1}\) be obtained by rounding each coordinate of \(r\) down to the nearest multiple of \(1/m\), so that
\(b\leq r\leq b+\tfrac1m\mathbf 1\). Hence
\([0,b]\subseteq[0,r]\subseteq Q_i(C)\), so the rectangle
\([0,b+\tfrac1m\mathbf 1]\) appears in the union defining \(D_m\). Since \(r\) belongs to this rectangle, \(r\in D_m\).

I now show that \(Q_i(C)=\bigcap_{m\ge 1}D_m\). Since \(Q_i(C)\subseteq D_m\) for all \(m\), it remains to show the reverse inclusion. Take any \(r\in \bigcap_m D_m\). For each \(m\), choose \(b_m\in \tfrac1m\mathbb Z^{N-1}\cap Q_i(C)\) such that \(r\le b_m+\tfrac1m\mathbf 1\). By compactness of \(Q_i(C)\), some subsequence of \((b_m)\) converges to \(\tilde b\in Q_i(C)\), and then \(r\le \tilde b\). Since \(Q_i(C)\) is lower, this implies \(r\in Q_i(C)\).

Now define \(C_m:=Q_i^{-1}(D_m)\) and set \(K_m:=\bigcap_{n=1}^m C_n\). Then \(K_{m+1}\subseteq K_m\), each \(K_m\) is closed and \(\succ_i\)-upper, and
\[
\bigcap_{m\ge 1}K_m
=
\bigcap_{m\ge 1}C_m
=
Q_i^{-1}\Bigl(\bigcap_{m\ge 1}D_m\Bigr)
=
Q_i^{-1}(Q_i(C))
=
C.
\]
Finally, each \(C_m\) is a finite union of sets of the form
\[
\Bigl\{\theta\in\Gamma_i:\ \frac{\theta_k}{\theta_i}\le b_k \ \text{for all }k\neq i\Bigr\}
=
\Bigl\{\theta\in\Gamma_i:\ \theta_k\le b_k\,\theta_i \ \text{for all }k\neq i\Bigr\},
\]
hence a finite union of polytopes in \(H\). The same is therefore true of each \(K_m\). 
Moreover, in ratio coordinates each \(K_m\) is a finite union of
axis-aligned rectangles anchored at the origin and sharing a common
full-dimensional rectangle. Hence it has Lipschitz boundary, and since
\(Q_i^{-1}\) is a smooth diffeomorphism on \(Q_i(\Gamma_i)\), the same is
true of \(K_m\) as a subset of \(H\).

\end{proof}

Let \(C\subseteq\Gamma_i\) be any closed \(\succ_i\)-upper set. By Lemma \ref{lem:approx_upper_simple}, there exists a decreasing sequence of closed \(\succ_i\)-upper sets \(K_m\downarrow C\), each a finite union of polytopes in \(H\), each with a Lipschitz boundary. By the above, \(\mu_i(K_m)\ge 0\) for every \(m\). Since \(\mu_i\) is finite, continuity from above gives \(\mu_i(C)=\lim_{m\to\infty}\mu_i(K_m)\ge 0\). Thus, condition~1 of Theorem \ref{thm:stoch_order_i} holds, so \(\mu_i^+\) \(\succ_i\)-stochastically dominates \(\mu_i^-\). Since this is true for every \(i\), Theorem \ref{thm:rent-measure-characterization} implies that the pure option mechanism is optimal.

\subsection{Proof of Corollary \ref{cor:iid_reverse_hazard_pure}}

The i.i.d.\ assumption implies that \(g\) and \(\lambda\) are exchangeable.
We verify directly the upper-set inequality used in the proof of
Proposition~\ref{prop:exchcondition}.

Fix \(i\), and write $E_i:=\{V_i\ge V_j\text{ for all }j\neq i\}.$ Conditional on \(V_i=k\) and \(E_i\), the coordinates \(\{V_j:j\neq i\}\) are
independent and each has distribution \(F_M\) truncated to \([0,k]\). Hence, for
\(j\neq i\) and \(t\in(0,1)\),
\[
\mathbb P\left[\frac{V_j}{V_i}\ge t\ \middle|\ V_i=k,\ E_i\right]
=
1-\frac{F_M(tk)}{F_M(k)}.
\]
It is therefore enough to show that \(F_M(tk)/F_M(k)\) is non-decreasing in
\(k\). Differentiating gives
\[
\frac{\partial}{\partial k}\frac{F_M(tk)}{F_M(k)}
=
\frac{F_M(tk)}{F_M(k)}
\left[
t\,\frac{f_M(tk)}{F_M(tk)}
-
\frac{f_M(k)}{F_M(k)}
\right]
=
\frac{F_M(tk)}{F_M(k)}
\frac{1}{k}
\left[
tk\,\frac{f_M(tk)}{F_M(tk)}
-
k\,\frac{f_M(k)}{F_M(k)}
\right].
\]
The bracketed term is nonnegative by \eqref{eq:reverse_hazard_elasticity}, since \(tk\le k\).
Thus \(V_j/V_i\) is \(\ge\)-stochastically decreasing in \(V_i\) conditional on
\(E_i\). Since the ratios are conditionally independent across \(j\neq i\), the vector
\[
\left(\frac{V_1}{V_i},\dots,\frac{V_N}{V_i}\right)
\]
is \(\ge\)-stochastically decreasing in \(V_i\) conditional on \(E_i\);
see Theorem~3.3.10 in \citet{müller2002comparison}.

Now, let \(C\subseteq\Gamma_i\) be any \(\succ_i\)-upper set. Since ${\Theta_j}/{\Theta_i}={V_j}/{V_i},$ membership in \(C\) is determined by a coordinatewise lower set of the ratio
vector \((V_1/V_i,\dots,V_N/V_i)\). The stochastic monotonicity from above implies that
\begin{equation}\label{eq:probthing}
k\longmapsto
\mathbb P[\Theta\in C\mid V_i=k,\ E_i]
\end{equation}
is non-decreasing. Moreover, notice that the conditional law of \(V_i\) given
\(\{V_i\ge t\}\cap E_i\) stochastically dominates its conditional law given
\(E_i\). Averaging \eqref{eq:probthing} under these two laws
therefore gives $\mathbb P[\Theta\in C\mid V_i\ge t,\ E_i]
\ge
\mathbb P[\Theta\in C\mid E_i]$ for every \(t\in[0,\bar v)\). Integrating over \(t\) then gives
\[
\mathbb E[V_i\mathbf 1_{\{\Theta\in C\}}\mid E_i]
\ge
\mathbb P[\Theta\in C\mid E_i]\,\mathbb E[V_i\mid E_i].
\]
Since \(E_i\) and \(C\) are determined by \(\Theta\) and $\mathbb E[V_i\mid \Theta]
=
\Theta_i\,\mathbb E\big[\sum V_j\mid \Theta\big]
=
\lambda(\Theta)\Theta_i,$ this gives
\[
\int_C \lambda(\theta)\theta_i g(\theta)\,d\theta
\ge
\frac{\int_{\Gamma_i}\lambda(\theta)\theta_i g(\theta)\,d\theta}
{\int_{\Gamma_i}g(\theta)\,d\theta}
\int_C g(\theta)\,d\theta.
\]
By exchangeability, \(\int_{\Gamma_i}g(\theta)\,d\theta=1/N\) and
\(\int_{\Gamma_i}\lambda(\theta)\theta_i g(\theta)\,d\theta=\bar A\), so
\[
\int_C \lambda(\theta)\theta_i g(\theta)\,d\theta
\ge
N\bar A\int_C g(\theta)\,d\theta,
\]
which was needed in the proof of Proposition~\ref{prop:exchcondition}.

\subsection{Proof of Proposition \ref{prop:neighbourhood-perturb}}

Write \(p^*:=(\lambda^*,g^*,s^*)\). I use a superscript \(*\) for objects
associated with \(p^*\). For any other primitive vector
\(p=(\lambda,g,s)\), I write \(\Gamma_i(p)\), \(q(p)\),
\(c(p)\), and \(\mu^p\).

Let \(\mathcal S_i(p)\) denote the collection of closed \(\succ_i\)-upper sets \(C\subseteq\Gamma_i(p)\). For \(p\) for which the vector of pure options \(q(p)\) is well defined, define the map \(\Phi_i^p:\Gamma_i^*\to\Gamma_i(p)\) by
\[
\Phi_i^p(\theta)=\theta^p,
\quad \text{where} \quad
r_{ik}^p(\theta)
:=
\frac{q_i(p)/q_k(p)}{q_i^*/q_k^*}
\frac{\theta_k}{\theta_i},
\quad k\neq i,
\qquad
\theta_k^p
=
\begin{cases}
\displaystyle
\frac{1}{1+\sum_{l\neq i}r_{il}^p(\theta)},
& k=i,\\[12pt]
\displaystyle
r_{ik}^p(\theta)\,
\frac{1}{1+\sum_{l\neq i}r_{il}^p(\theta)},
& k\neq i.
\end{cases}
\]
By construction, \(\Phi_i^{p^*}\) is the identity on \(\Gamma_i^*\). The
following fact follows from the definition of \(\Phi_i^p\) and a standard continuity argument.
\begin{fact}\label{fact:region-identification-maps}
Let \(p\) be any primitive vector for which \(q(p)\in\mathbb R^N_{++}\) is
well defined. Then $\Phi_i^p:\Gamma_i^*\to\Gamma_i(p)$ is a homeomorphism and an order isomorphism:
\[
\theta\succ_i\theta'
\quad\Longleftrightarrow\quad
\Phi_i^p(\theta)\succ_i\Phi_i^p(\theta'),
\quad \text{and thus} \quad
C\in\mathcal S_i^*
\quad\Longleftrightarrow\quad
\Phi_i^p(C)\in\mathcal S_i(p).
\]

Moreover, if \(q(p_n)\to q^*\), then $\Phi_i^{p_n}\to \mathrm{id}$ uniformly on \(\Gamma_i^*\), and the \((N-1)\)-dimensional interior Jacobians
and the \((N-2)\)-dimensional boundary Jacobians of \(\Phi_i^{p_n}\) converge
uniformly to \(1\).
\end{fact}

Given a signed measure \(\nu^p\) on \(\Gamma\), define its pasted pullback
\(\hat\nu^p\) region by region, by requiring that for each \(i\),
\[
\hat\nu^p\restriction_{\Gamma_i^*}(C)
:=
\nu^p(\Phi_i^p(C)),
\qquad 
C\subseteq\Gamma_i^*\text{ Borel}.
\]
Equivalently, whenever \(C\subseteq\Gamma_i^*\), the notation
\(\hat\nu^p(C)\) means this \(i\)-th regionwise pullback. We now show three lemmas.

\begin{lemma}
\label{lem:BUS-open-moving}
Assume that \(\mu^*\) has upper-set slack with constant \(\eta>0\). Let
\(\nu^p\) be a regionwise balanced signed measure on \(\Gamma\), and suppose
that $\mu^*\ll\alpha$ and $\hat\nu^p\ll\alpha$ region by region. Define
\[
d_{\mathrm{US}}^p(\nu^p,\mu^*)
:=
\max_i
\sup\left\{
\frac{|(\hat\nu^p-\mu^*)(C)|}
{\min\{\alpha(C),\alpha(\Gamma_i^*\setminus C)\}}
:
C\in\mathcal S_i^*,\ 
\min\{\alpha(C),\alpha(\Gamma_i^*\setminus C)\}>0
\right\}.
\]
If $d_{\mathrm{US}}^p(\nu^p,\mu^*)<\eta,$ 
then, for every \(i\) and every \(D\in\mathcal S_i(p)\), $\nu^p(D)\ge0.$
\end{lemma}

\begin{proof}
Fix \(i\) and \(D\in\mathcal S_i(p)\). By
Fact~\ref{fact:region-identification-maps}, \(C:=(\Phi_i^p)^{-1}(D)\) belongs to \(\mathcal S_i^*\), and \(\nu^p(D)=\hat\nu^p(C)\).  

If $\min\{\alpha(C),\alpha(\Gamma_i^*\setminus C)\}=0,$ then either \(\alpha(C)=0\) or
\(\alpha(\Gamma_i^*\setminus C)=0\). In the first case,
\(\hat\nu^p\ll\alpha\) gives $\hat\nu^p(C)=0.$ In the second case, regionwise balance gives $\hat\nu^p(C)
=
-\hat\nu^p(\Gamma_i^*\setminus C)
=
0$ again because \(\hat\nu^p\ll\alpha\). Hence \(\nu^p(D)\ge0\).

Now suppose $\min\{\alpha(C),\alpha(\Gamma_i^*\setminus C)\}>0.$ By the definition of \(d_{\mathrm{US}}^p\),
\[
\hat\nu^p(C)
\ge
\mu^*(C)
-
d_{\mathrm{US}}^p(\nu^p,\mu^*)
\min\{
\alpha(C),
\alpha(\Gamma_i^*\setminus C)
\}.
\]
Using \eqref{eq:BUS}, we obtain
\[
\hat\nu^p(C)
\ge
(\eta-d_{\mathrm{US}}^p(\nu^p,\mu^*))
\min\{
\alpha(C),
\alpha(\Gamma_i^*\setminus C)
\}
\ge 0.
\]
Since \(\nu^p(D)=\hat\nu^p(C)\), the claim follows.
\end{proof}

\begin{lemma}
\label{lem:pulled-back-density-control-US}
Let \(\nu^p\) be a signed measure on \(\Gamma\) that is balanced on each
\(\Gamma_i(p)\), and suppose that \(\mu^*\) is balanced on each
\(\Gamma_i^*\). For each \(i\), write
\[
m_i^*:=m\restriction_{\Gamma_i^*},
\qquad
\sigma_i^*:=\sigma\restriction_{\Gamma_i^*\cap\partial\Gamma},
\qquad
\alpha_i^*:=m_i^*+\sigma_i^*.
\]
Suppose that, region by region,
\[
\hat\nu^p-\mu^*
=
a_i\,m_i^*
+
b_i\,\sigma_i^*
\text{  on  }\Gamma_i^*,
\quad \text{with} \quad
\max_i
\left\{
\|a_i\|_{L^\infty(m_i^*)}
+
\|b_i\|_{L^\infty(\sigma_i^*)}
\right\}
\le r.
\]
Then \(d_{\mathrm{US}}^p(\nu^p,\mu^*)\le r\).
\end{lemma}

\begin{proof}
Fix \(i\) and \(C\in\mathcal S_i^*\). On \(\Gamma_i^*\), define
\(\Delta:=\hat\nu^p-\mu^*\). By the assumed density representation,
\[
|\Delta(C)|
\le
\|a_i\|_{L^\infty(m_i^*)}m_i^*(C)
+
\|b_i\|_{L^\infty(\sigma_i^*)}\sigma_i^*(C)
\le
r\alpha_i^*(C).
\]
Since \(\nu^p\) and \(\mu^*\) are regionwise balanced, $\hat\nu^p(\Gamma_i^*)=0$ and $\mu^*(\Gamma_i^*)=0.$ Hence \(\Delta(\Gamma_i^*)=0\), so $\Delta(C)=-\Delta(\Gamma_i^*\setminus C).$ Applying the same density bound to \(\Gamma_i^*\setminus C\) gives
\[
|\Delta(C)|
=
|\Delta(\Gamma_i^*\setminus C)|
\le
r\alpha_i^*(\Gamma_i^*\setminus C).
\]
Combining the two bounds, we get
\[
|\Delta(C)|
\le
r
\min\{
\alpha_i^*(C),
\alpha_i^*(\Gamma_i^*\setminus C)
\}.
\]
Since \(\alpha_i^*\) is the restriction of \(\alpha\) to \(\Gamma_i^*\), this is the bound appearing in \(d_{\mathrm{US}}^p\). Taking the supremum over
\(C\in\mathcal S_i^*\), and then the maximum over \(i\), gives
\(d_{\mathrm{US}}^p(\nu^p,\mu^*)\le r\).
\end{proof}

\begin{lemma}
\label{lem:primitive-ball-to-US-ball}
Consider primitives \(p=(\lambda,g,s)\) written as
\[
\lambda=e^l\lambda^*,
\qquad
g=\frac{e^h g^*}{\int_\Gamma e^h g^*\,dm},
\quad \text{and define} \quad
d_{\mathcal P}(p,p^*)
:=
\|l\|_{C^1(\Gamma)}
+
\|h\|_{C^2(\Gamma)}
+
\|s-s^*\|.
\] 
Then, for every \(r>0\), there exists \(\varepsilon>0\) such that, whenever
\(d_{\mathcal P}(p,p^*)<\varepsilon\), the objects
\[
q(p),\qquad c(p),\qquad \Gamma_i(p),\qquad \Phi_i^p,\qquad \mu^p
\]
are all well defined, \(q(p)\in\mathbb R^N_{++}\), and $d_{\mathrm{US}}^p(\mu^p,\mu^*)<r.$
\end{lemma}

\begin{proof}
By the implicit function theorem applied to the market-clearing equations,
there is a \(d_{\mathcal P}\)-neighborhood of \(p^*\) on which the pure-option
vector \(q(p)\) is uniquely and continuously defined. Shrinking the
neighborhood if necessary, \(q(p)\in\mathbb R^N_{++}\). Hence the regions
\(\Gamma_i(p)\) and the identification maps \(\Phi_i^p:\Gamma_i^*\to\Gamma_i(p)\)
are well defined.

Now, let \(p_n=(\lambda^n,g^n,s^n)\) be any sequence such that
\(d_{\mathcal P}(p_n,p^*)\to0\). Then
\begin{equation}\label{eq:nicecontperturb}
\lambda^n\to\lambda^*
\quad\text{in }C^1(\Gamma),
\qquad
g^n\to g^*
\quad\text{in }C^2(\Gamma),
\qquad
s^n\to s^*.
\end{equation}
By the local continuity of the market-clearing solution,
\(q(p_n)\to q^*\). Hence, by
Fact~\ref{fact:region-identification-maps}, the maps
\(\Phi_i^{p_n}\) converge uniformly to the identity on \(\Gamma_i^*\), and
their interior and boundary Jacobians converge uniformly to \(1\).

The quantities defining the shadow-cost system are continuous in the
primitives. Indeed, after pulling all moving regions and faces back to the
fixed baseline regions and faces by \(\Phi_i^{p_n}\), the domains are fixed,
the Jacobians converge uniformly to \(1\), and the integrands converge
uniformly by \eqref{eq:nicecontperturb}. Thus, the shadow-cost matrices and right-hand sides converge: $J(p_n)\to J^*$ and $A(p_n)\to A^*.$ Since \(J^*\) is nonsingular, \(J(p_n)\) is nonsingular for all sufficiently
large \(n\), and
\[
c(p_n)=J(p_n)^{-1}A(p_n)\to (J^*)^{-1}A^*=c^*.
\]

For each \(i\), write the restriction of the rent measure \(\mu^p\) to
\(\Gamma_i(p)\) in density form as
\[
d\mu^p
=
\rho_i(p)\,dm
+
\beta_i(p)\,d\sigma
\ \text{on }\Gamma_i(p),
\]
where
\[
\rho_i(p)(\theta)
=
\theta_i
\Big[
\lambda(\theta)g(\theta)
+
\operatorname{div}_{\Gamma}
\big(
(c(p)-(\sum_j c_j(p))\theta)g(\theta)
\big)
-
(\sum_j c_j(p))g(\theta)
\Big],
\]
and, on the boundary faces, $\beta_i(p)(\theta)
=
-
\theta_i
\big(c(p)-(\sum_j c_j(p))\theta\big)g(\theta)\cdot\nu(\theta).$ Write \(\hat\mu^{p_n}\) for the pasted pullback of \(\mu^{p_n}\). That is, on
each baseline region \(\Gamma_i^*\), $\hat\mu^{p_n}(C)
:=
\mu^{p_n}(\Phi_i^{p_n}(C)),$ where $C\subseteq\Gamma_i^*$. By the area formula, on \(\Gamma_i^*\), $d\hat\mu^{p_n}
=
\hat\rho_i^{\,n}\,dm_i^*
+
\hat\beta_i^{\,n}\,d\sigma_i^*,$
where
\[
\hat\rho_i^{\,n}(\theta)
=
\rho_i(p_n)(\Phi_i^{p_n}(\theta))
J_{i,\mathrm{int}}^{p_n}(\theta),
\quad \text{and, \(\sigma_i^*\)-a.e.,} \quad
\hat\beta_i^{\,n}(\theta)
=
\beta_i(p_n)(\Phi_i^{p_n}(\theta))
J_{i,\partial}^{p_n}(\theta).
\]
Using the convergence of \(\lambda^n\), \(g^n\), \(q(p_n)\), \(c(p_n)\),
\(\Phi_i^{p_n}\), and the corresponding Jacobians, we obtain $\|\hat\rho_i^{\,n}-\rho_i(p^*)\|_{L^\infty(m_i^*)}
\to0$ and $\|\hat\beta_i^{\,n}-\beta_i(p^*)\|_{L^\infty(\sigma_i^*)}
\to0.$ Hence
\[
\hat\mu^{p_n}-\mu^*
=
a_i^n\,m_i^*
+
b_i^n\,\sigma_i^*
\ \ \text{on }\Gamma_i^*,
\ \  \text{with} \ \
\max_i
\left\{
\|a_i^n\|_{L^\infty(m_i^*)}
+
\|b_i^n\|_{L^\infty(\sigma_i^*)}
\right\}
\to0.
\]
Since rent measures are regionwise balanced, Lemma
\ref{lem:pulled-back-density-control-US} implies $d_{\mathrm{US}}^{p_n}(\mu^{p_n},\mu^*)\to0.$
\end{proof}

I now complete the proof. Let us apply Lemma~\ref{lem:primitive-ball-to-US-ball} with \(r=\eta/2\). After
shrinking \(\varepsilon>0\) if necessary, every primitive vector \(p\)
satisfying the displayed perturbation bound has
\[
q(p),\quad c(p),\quad \Gamma_i(p),\quad \Phi_i^p,\quad \mu^p
\]
well defined, with \(q(p)\in\mathbb R^N_{++}\), and satisfies $d_{\mathrm{US}}^p(\mu^p,\mu^*)<\eta/2.$ 
By Lemma~\ref{lem:BUS-open-moving}, it follows that, for every \(i\) and every
closed \(\succ_i\)-upper set \(D\subseteq\Gamma_i(p)\), $\mu^p(D)\ge0.$ Now fix \(i\), and write $\mu_i^p:=\mu^p\restriction_{\Gamma_i(p)}.$ Since the rent measure is regionwise balanced, $\mu_i^p(\Gamma_i(p))=0.$ Let \((\mu_i^p)^+\) and \((\mu_i^p)^-\) denote the positive and negative parts
of \(\mu_i^p\). Then $(\mu_i^p)^+(\Gamma_i(p))
=
(\mu_i^p)^-(\Gamma_i(p)).$ Moreover, for every closed \(\succ_i\)-upper set \(D\subseteq\Gamma_i(p)\),
\[
(\mu_i^p)^+(D)-(\mu_i^p)^-(D)
=
\mu_i^p(D)
=
\mu^p(D)
\ge0;
\quad \text{equivalently,}\quad
(\mu_i^p)^-(D)\le(\mu_i^p)^+(D).
\]
By Theorem~\ref{thm:stoch_order_i}, there exists a
\(\succ_i\)-monotone transport plan from \((\mu_i^p)^-\) to
\((\mu_i^p)^+\). Since \(i\) was arbitrary, this verifies the regionwise transport condition
in Theorem~\ref{thm:rent-measure-characterization} for the tuple
\((\tilde\mu_i^p)_{i=1}^N\) defined by
\(\tilde\mu_i^p:=\mu_i^p\).

It remains only to verify preprocessing. The rent measure \(\mu^p\) assigns
zero mass to the indifference boundaries between distinct pure-option
regions, and \(\theta_*^p=\theta_i\) \(\mu^p\)-almost everywhere on
\(\Gamma_i(p)\). Therefore, for every continuous, convex
\(U:\Gamma\to\mathbb R_+\),
\[
\sum_{i=1}^N
\int_{\Gamma_i(p)}
\frac{U}{\theta_i}\,d\tilde\mu_i^p
=
\int_\Gamma
\frac{U}{\theta_*^p}\,d\mu^p.
\]
Thus preprocessing holds with equality. The tuple
\((\tilde\mu_i^p)_{i=1}^N\) satisfies both conditions of
Theorem~\ref{thm:rent-measure-characterization}. Therefore the theorem
applies to the primitives \(p\), and the associated pure option mechanism is
optimal.

\subsection{Proof of Proposition \ref{prop:PE_impl_implies_PO}}

Fix an allocation rule \(x\) that is Pareto efficient subject to
\eqref{eq:supply'}. First, note that every supply constraint must bind;
otherwise, we could give every agent a common positive amount of a good whose
supply constraint is slack, strictly raising utility for almost every type
without violating supply.

Note also that \(x\) is pure almost everywhere. Indeed, suppose instead that, for
some \(i\neq j\), the set
\[
M:=\{\theta\in\Gamma:x_i(\theta)>0,\ x_j(\theta)>0\}
\]
has positive mass. Since \(G\) admits a density, choose \(t>0\) such that
both
\[
M^-:=M\cap\{\theta_i/\theta_j<t\},
\qquad
M^+:=M\cap\{\theta_i/\theta_j>t\}
\]
have positive mass, and define $m^-:=\int_{M^-}x_i\,dG>0$ and $m^+:=\int_{M^+}x_j\,dG>0.$

For sufficiently small \(\delta>0\), replace \(\delta x_i\) units of good
\(i\) with \(t\delta x_i\) units of good \(j\) on \(M^-\), and replace
\(t\delta(m^-/m^+)x_j\) units of good \(j\) with
\(\delta(m^-/m^+)x_j\) units of good \(i\) on \(M^+\). This preserves
aggregate supply and nonnegativity and strictly benefits every affected type,
because $t\theta_j-\theta_i>0$ on $M^-$ and $\theta_i-t\theta_j>0$ on $M^+$. This modified allocation rule is a Pareto improvement; contradiction. Thus, \(x\) is pure a.e.

We now show that any such allocation rule that satisfies \eqref{eq:ICprime} must coincide with the pure option mechanism a.e. For each \(i\), let
\[
S_i
:=
\{\theta\in\Gamma:x(\theta)=x_i(\theta)e_i,\ x_i(\theta)>0\}.
\]
Since the \(i\)-th supply constraint binds and \(s_i>0\), \(S_i\) has
positive mass. Incentive compatibility implies that, for some \(q_i^*>0\), $x(\theta)=q_i^*e_i$ for a.e. $\theta\in S_i.$ Indeed, outside the null face \(\{\theta_i=0\}\), any type receiving a
smaller quantity of good \(i\) would imitate one receiving a larger quantity. Next, define
\[
C_i
:=
\left\{
\theta\in\Gamma:
\theta_iq_i^*\ge\theta_kq_k^*
\text{ for every }k
\right\}.
\]
Incentive compatibility against reports in \(S_k\) implies
\(S_i\subseteq C_i\) up to a null set. Conversely, for a.e. type satisfying
\(\theta_iq_i^*>\theta_kq_k^*\) for every \(k\neq i\), incentive
compatibility requires \(x(\theta)=q_i^*e_i\). Otherwise, that type would
strictly prefer to report a type in \(S_i\). Since different \(C_i\) overlap
only on zero-measure indifference boundaries, we get that \(x\)
coincides almost everywhere with the pure-option allocation generated by
\(q^*:=(q_1^*,\dots,q_N^*)\).

Finally, because every supply constraint binds, $s_i=q_i^*G(C_i)$ for every $i$. Thus, \(q^*\) is a market-clearing vector of pure options. By the uniqueness
of the market-clearing pure-option vector established in
Fact~\ref{fact:puremechs}, \(q^*=q\). Consequently, all \(C_i\)
coincide a.e. with the pure-option regions \(\Gamma_i\).

\section{Omitted proofs from Section \ref{sec:2good}}

\subsection{Proof of Theorem \ref{th:2good_asym_menu}}

For each \(z\in[0,1]\), let \(x^z\) be the threshold rule in
\eqref{eq:thresholdrulee}.

\begin{fact}\label{fact:two-good-IC-representation}
Every incentive-compatible allocation rule \(x\) can be modified on a
\(G\)-null set, while preserving incentive compatibility, so that there
exists a finite positive measure \(\nu\) on \([0,1]\) satisfying
\begin{equation}\label{eq:measure_rep_xz_short}
x(1-t,t)
=
\int_{[0,1]}x^z(1-t,t)\,\nu(dz)
\qquad\text{for a.e. }t.
\end{equation}
Conversely, for every finite positive measure \(\nu\), the allocation rule
defined by \eqref{eq:measure_rep_xz_short} is
nonnegative and incentive compatible.
\end{fact}

\begin{proof}
Suppose first that \(x\) is incentive compatible. Write
\[
U(1-t,t):=(1-t)x_1(1-t,t)+tx_2(1-t,t),
\qquad
\Delta(t):=x_2(1-t,t)-x_1(1-t,t).
\]
By the one-dimensional IC characterization of \cite{myerson}, after modifying
\(x\) on a \(G\)-null set, we may take \(\Delta\) to be right-continuous and
non-decreasing, with
\[
U(1-t,t)=U(1,0)+\int_0^t\Delta(s)\,ds.
\]
Let \(\mu\) be the finite positive measure satisfying $\mu((0,t])=\Delta(t)-\Delta(0),$ with $\mu(\{0\})=0.$ The envelope formula and Tonelli's theorem then give
\[
x_1(1-t,t)
=
x_1(0,1)+\int_{(t,1]}z\,\mu(dz),
\qquad
x_2(1-t,t)
=
x_2(1,0)+\int_{[0,t]}(1-z)\,\mu(dz).
\]
Since \(x\ge0\), the endpoint terms are nonnegative. Hence, defining $\nu:=\mu+x_2(1,0)\delta_0+x_1(0,1)\delta_1,$ we obtain, almost everywhere,
\begin{equation}\label{eq:measure_rep_x1x2_short}
x_1(1-t,t)=\int_{(t,1]}z\,\nu(dz),
\qquad
x_2(1-t,t)=\int_{[0,t]}(1-z)\,\nu(dz).
\end{equation}
This is equivalent to \eqref{eq:measure_rep_xz_short}.

Conversely, suppose that \(x\) has the form
\eqref{eq:measure_rep_xz_short} for some finite positive measure \(\nu\).
Then \(x\ge0\), and
\[
\Delta(t)
=
\nu([0,t])-\int_{[0,1]}z\,\nu(dz) \quad \text{is non-decreasing.}
\]
Moreover,
\[
U(1-t,t)
=
\int_{[0,1]}
\max\{z(1-t),(1-z)t\}\,\nu(dz)
=
U(1,0)+\int_0^t\Delta(s)\,ds.
\]
The one-dimensional IC characterization therefore implies that \(x\) is
incentive compatible.
\end{proof}

Fact \ref{fact:two-good-IC-representation} and Fubini's theorem let us rewrite the designer's
problem as 
\begin{equation}\label{eq:primalLP}
W^*
=
\sup_{\nu\ge0}
\left\{
\int_{[0,1]}w(z)\,\nu(dz):
\int_{[0,1]}\varphi_1(z)\,\nu(dz)\le s_1,\quad
\int_{[0,1]}\varphi_2(z)\,\nu(dz)\le s_2
\right\}.
\end{equation}
We now show that \eqref{eq:primalLP} admits an optimizer, derive its dual,
and establish that both supply constraints bind at the optimum. To that end, let
\[
\kappa:=\min_{z\in[0,1]}
\{\varphi_1(z)+\varphi_2(z)\}>0.
\]
Every feasible measure \(\nu\) in \eqref{eq:primalLP} satisfies
\[
\kappa\nu([0,1])
\le
\int_{[0,1]}(\varphi_1+\varphi_2)\,d\nu
\le
s_1+s_2,
\]
so the feasible measures have uniformly bounded total mass. Since
\([0,1]\) is compact, measures with uniformly bounded total mass form a
weak-* compact set. Moreover, because \(\varphi_1\) and \(\varphi_2\) are
continuous, nonnegativity and the two supply constraints define a weak-*
closed subset of this set. The feasible set of \eqref{eq:primalLP} is
therefore weak-* compact. Since \(w\) is continuous, the objective
\(\nu\mapsto\int w\,d\nu\) is weak-* continuous, so a primal optimizer
exists.

To derive the dual, attach multipliers \(c_1,c_2\ge0\) to the two supply
constraints. The Lagrangian is
\[
c_1s_1+c_2s_2
+
\int_{[0,1]}
\bigl[
w(z)-c_1\varphi_1(z)-c_2\varphi_2(z)
\bigr]\,\nu(dz).
\]
Its supremum over all positive measures \(\nu\) is finite if and only if $c_1\varphi_1(z)+c_2\varphi_2(z)\ge w(z)$ for every $z\in[0,1].$ Indeed, if this inequality failed at some \(z\), placing arbitrarily large
mass at \(z\) would make the Lagrangian arbitrarily large. The dual is
therefore
\[
\inf_{c_1,c_2\ge0}
\left\{
c_1s_1+c_2s_2:
c_1\varphi_1(z)+c_2\varphi_2(z)\ge w(z)
\quad\text{for every }z\in[0,1]
\right\}.
\]
Since the zero measure is strictly feasible, conic duality gives strong duality and attainment of the dual optimum \citep{cones}. Also, both primal constraints bind at any optimum: otherwise, adding
mass at \(z=1\) or \(z=0\) would use only the slack good and strictly
increase welfare.

Now, let \((c_1^*,c_2^*)\) be dual optimal. Evaluating the dual constraint at
\(z=1\), where \(\varphi(1)=(1,0)\), and at \(z=0\), where
\(\varphi(0)=(0,1)\), gives $c_1^*\ge w(1)>0$ and $c_2^*\ge w(0)>0.$ We may therefore define \(\rho:=c_1^*/c_2^*>0\). Define the corresponding
contact set by
\[
E
:=
\left\{
z\in[0,1]:
c_1^*\varphi_1(z)+c_2^*\varphi_2(z)=w(z)
\right\}.
\]
Every primal optimizer is supported on \(E\). Indeed, if \(\nu^*\) is primal
optimal, then strong duality and the fact that both supply constraints bind
give
\[
0
=
c_1^*s_1+c_2^*s_2-\int_{[0,1]}w(z)\,\nu^*(dz)
=
\int_{[0,1]}
\bigl[
c_1^*\varphi_1(z)+c_2^*\varphi_2(z)-w(z)
\bigr]\,\nu^*(dz).
\]
Since the integrand is nonnegative by dual feasibility, it must vanish
\(\nu^*\)-a.e., so \(\nu^*\) is supported on \(E\).

I now show that \(E\subset(0,1)\). Fix
\(0<z<1/(1+\rho)\). Then
\[
\rho\varphi_1(z)+\varphi_2(z)
=
(1-z)+\bigl(\rho z-(1-z)\bigr)
\mathbb P(\Theta_2\le z)
<
1-z.
\]
Moreover, \(x^z\) gives every type at least the utility generated by
\((1-z)x^0\), and strictly more to a positive mass of types. Hence $w(z)>(1-z)w(0)
>
\bigl[\rho\varphi_1(z)+\varphi_2(z)\bigr]w(0).$ On the other hand, dual feasibility and \(c_1^*=\rho c_2^*\) imply $c_2^*
\bigl[\rho\varphi_1(z)+\varphi_2(z)\bigr]
\ge w(z).$
Since \(\rho\varphi_1(z)+\varphi_2(z)>0\), the preceding inequalities give
\(c_2^*>w(0)\). The dual constraint is therefore strict at \(z=0\), so
\(0\notin E\). Similarly, fix \(1/(1+\rho)<z<1\). Then $\rho\varphi_1(z)+\varphi_2(z)
=
\rho z+\bigl((1-z)-\rho z\bigr)
\mathbb P(\Theta_2\ge z)
<
\rho z.$ Moreover, \(x^z\) gives every type at least the utility generated by
\(zx^1\), and strictly more to a positive mass of types. Therefore, $w(z)>zw(1)
>
\frac{w(1)}{\rho}
\bigl[\rho\varphi_1(z)+\varphi_2(z)\bigr].$ Combining this with dual feasibility gives
\(c_2^*>w(1)/\rho\), and hence \(c_1^*>w(1)\). The dual constraint is
therefore strict at \(z=1\), so \(1\notin E\), which tells us \(E\subset(0,1)\).

Finally, if \(\nu^*\) is any primal optimizer, its support is contained in
\(E\), and the binding supply constraints give
\[
s
=
\int_E\varphi(z)\,\nu^*(dz)
\in
\operatorname{cone}\{\varphi(z):z\in E\}.
\]
By Carath\'eodory's theorem for cones in \(\mathbb R^2\), there exist
\(a,b\in E\), \(a\le b\), and \(m_a,m_b\ge0\) such that $m_a\varphi(a)+m_b\varphi(b)=s.$ The measure \(m_a\delta_a+m_b\delta_b\) is therefore optimal. Its allocation
rule is
\[
x(1-t,t)
=
m_ax^a(1-t,t)+m_bx^b(1-t,t)
=
\begin{cases}
(r_1,0), & t<a,\\
(b_1,b_2), & a\le t<b,\\
(0,r_2), & t\ge b,
\end{cases}
\]
where
\begin{equation}\label{eq:quantitiesinmenu}
r_1:=am_a+bm_b,\qquad
b_1:=bm_b,\qquad
b_2:=(1-a)m_a,\qquad
r_2:=(1-a)m_a+(1-b)m_b.
\end{equation}
Thus, an optimal mechanism offers at most two pure options and one bundle.

\paragraph{Characterizing optimality of the pure option mechanism.}
In the threshold-rule representation, the pure option mechanism is uniquely
represented by $\frac{s_1}{\varphi_1(\theta_2^0)}\delta_{\theta_2^0},$ and therefore generates welfare
\begin{equation}\label{eq:pure_welfare_short}
\frac{s_1}{\varphi_1(\theta_2^0)}w(\theta_2^0).
\end{equation}
By the preceding argument, it is optimal if and only if it weakly dominates
every nondegenerate two-threshold mechanism. Such a mechanism is generated by
\(a<b\) and \(m_a,m_b>0\) satisfying
\[
m_a\varphi(a)+m_b\varphi(b)=s.
\]
Since \(\varphi_1(z)/\varphi_2(z)\) is strictly increasing, these weights
exist uniquely if and only if
\[
\frac{\varphi_1(a)}{\varphi_2(a)}
<
\frac{s_1}{s_2}
<
\frac{\varphi_1(b)}{\varphi_2(b)}.
\]
Moreover, the pure-option cutoff satisfies $\frac{s_1}{s_2}
=
\frac{\varphi_1(\theta_2^0)}{\varphi_2(\theta_2^0)},$ so the preceding inequalities are equivalent to
\(a<\theta_2^0<b\). Such a mechanism generates welfare
\(m_aw(a)+m_bw(b)\). Hence, the pure option mechanism is optimal exactly when
\eqref{eq:asym_no_profitable_bundle} holds for every such pair \(a,b\).

\subsection{Proof of Corollary \ref{cor:2good_symmetric_two_options_theta}}

Symmetrizing any feasible measure \(\nu\) under \(z\mapsto1-z\) preserves
feasibility and welfare. Among symmetric measures, the two supply constraints
coincide. The resulting one-constraint linear program is therefore solved by
either the cutoff \(1/2\) or a symmetric pair \(z,1-z\), with \(z<1/2\).
Hence, the pure option mechanism is optimal if and only if
\[
\frac{w(1-z)}{\varphi_1(1-z)+\varphi_2(1-z)}
\le
\frac{w(1/2)}{\varphi_1(1/2)+\varphi_2(1/2)}
\qquad
\text{for every }z\in[0,1/2).
\]
By exchangeability,
\[
w(1-z)
=
\frac12\mathbb E\left[
\lambda(\Theta)
\left(
(1-z)\mathbf 1_{\{\min_i\Theta_i\ge z\}}
+
\max_i\Theta_i\mathbf 1_{\{\min_i\Theta_i<z\}}
\right)
\right],
\]
while
\[
\varphi_1(1-z)+\varphi_2(1-z)
=
\frac12\left[
1+(1-2z)\mathbb P(\min_i\Theta_i\ge z)
\right].
\]
Moreover,
\[
\frac{w(1/2)}{\varphi_1(1/2)+\varphi_2(1/2)}
=
\mathbb E[\lambda(\Theta)\max_i\Theta_i].
\]
Substitution and rearrangement show that the ratio inequality is equivalent
to
\[
\mathbb E\left[
\lambda(\Theta)(\min_i\Theta_i-z)
\mathbf 1_{\{\min_i\Theta_i\ge z\}}
\right]
\le
(1-2z)\mathbb E[\lambda(\Theta)\max_i\Theta_i]\,
\mathbb P(\min_i\Theta_i\ge z).
\]
Dividing by \(\mathbb P(\min_i\Theta_i\ge z)>0\) then gives
\eqref{eq:foralltau_cor}.

Finally, suppose \(t\lambda(1-t,t)\) is non-decreasing on
\([1/2,1]\). On \(\{\min_i\Theta_i\ge z\}\),
\[
\min_i\Theta_i-z
=
1-z-\max_i\Theta_i
\le
(1-2z)\max_i\Theta_i.
\]
Furthermore,
\[
\lambda(\Theta)\max_i\Theta_i
=
\max_i\Theta_i\,
\lambda(1-\max_i\Theta_i,\max_i\Theta_i)
\]
is non-decreasing in \(\max_i\Theta_i\). Conditioning on
\(\min_i\Theta_i\ge z\), equivalently
\(\max_i\Theta_i\le1-z\), therefore weakly lowers its expectation. Hence
\[
\mathbb E\left[
\lambda(\Theta)(\min_i\Theta_i-z)
\ \middle|\ 
\min_i\Theta_i\ge z
\right]
\le
(1-2z)\mathbb E[\lambda(\Theta)\max_i\Theta_i],
\]
which is \eqref{eq:foralltau_cor}.

\section{Omitted proofs from Section \ref{sec:observable-groups}}

We first establish a lemma that will be useful throughout.

\begin{lemma}
\label{lem:Wr-properties}
For each group \(r\), the value function
\(W_r:\mathbb R_+^N\to\mathbb R_+\) is finite, increasing, concave, and
homogeneous of degree one.
\end{lemma}

\begin{proof}
Fix a group \(r\). For every allocation rule \(x\) feasible under supply
\(s\),
\[
0
\leq
\pi_r\int_\Gamma
\lambda_r(\theta)\theta\cdot x(\theta)\,dG_r(\theta)
\leq
\|\lambda_r\|_\infty\sum_{i=1}^N s_i,
\]
so \(W_r\) is finite. It is increasing because a larger supply vector
relaxes the feasibility constraint. For concavity, fix \(s,t\in\mathbb R_+^N\) and \(\alpha\in[0,1]\), and let
\(x^s\) and \(x^t\) be optimal under \(s\) and \(t\), respectively. The
allocation $x^\alpha:=\alpha x^s+(1-\alpha)x^t$ is incentive compatible and feasible under
\(\alpha s+(1-\alpha)t\). By linearity of welfare,
\[
W_r\bigl(\alpha s+(1-\alpha)t\bigr)
\geq
\alpha W_r(s)+(1-\alpha)W_r(t),
\]
which proves concavity.

Finally, if \(x\) is feasible under \(s\), then \(\beta x\) is feasible
under \(\beta s\) and generates \(\beta\) times its welfare. Thus,
\(W_r(\beta s)\geq\beta W_r(s)\) for every \(\beta\geq0\). For
\(\beta>0\), scaling any allocation feasible under \(\beta s\) by
\(1/\beta\) gives the reverse inequality. Hence \(W_r(\beta s)=\beta W_r(s)\) for every \(\beta\geq0\), where the case \(\beta=0\) follows from \(W_r(0)=0\).
\end{proof}

\subsection{Proof of Theorem \ref{thm:3}}

\begin{lemma}
\label{lem:infinite-marginal-missing-good}
Fix \(s\in\mathbb R_+^N\setminus\{0\}\) such that \(s_j=0\). Then
\[
\lim_{\varepsilon\downarrow0}
\frac{W_r(s+\varepsilon e_j)-W_r(s)}{\varepsilon}
=
+\infty.
\]
\end{lemma}

\begin{proof}
Suppress the group index and write \(\pi,G,g,\lambda\), and \(W\) for the
group-specific objects. Let \(U\) be the indirect utility induced by an
optimal allocation under supply \(s\), and let \(x\) be its bounded,
pointwise incentive-compatible implementation from the proof of
Proposition~\ref{thm:icchar}.

Let $B:=\{k:s_k>0\}.$ Then every supply constraint corresponding to \(k\in B\) binds: otherwise, adding a
small common amount of good \(k\) to every option would preserve incentive
compatibility and strictly increase welfare. Hence $\pi\int_\Gamma x_k(\theta)\,dG(\theta)=s_k>0$ for every $k\in B$, while \(x_l=0\) almost everywhere for \(l\notin B\).

By Corollary~\ref{lem:utility-space-corr}, there exists
\(\ool M<\infty\) such that \(\sum_kx_k(\theta)\leq\ool M\) a.e. Therefore
\begin{equation}
U(\theta)
=
\sum_{k\in B}\theta_kx_k(\theta)
\leq
\ool M\sum_{k\in B}\theta_k
\leq
\ool M(1-\theta_j)
\quad\text{a.e.}
\label{eq:direct-upper-bound-U}
\end{equation}
Now, fix \(q>0\) and add the option \(qe_j\) to the menu corresponding to the
optimal mechanism under supply \(s\). Let $D_q
:=
\{\theta\in\Gamma:q\theta_j>U(\theta)\}$ denote the set of types switching to this new option. By
\eqref{eq:direct-upper-bound-U}, \(D_q\) contains, up to a null set, $\{
\theta\in\Gamma:
\theta_j>{\ool M}/{\ool M+q}
\}.$
This is a nonempty, relatively open subset of \(\Gamma\), so full support
implies
\begin{equation}
G(D_q)>0
\ \ \text{for every }q>0.
\label{eq:dqbound}
\end{equation}

We next derive a lower bound on the old allocations of the switching types. For every
\(k\in B\), positive aggregate use of good \(k\) implies that there exist
\(\hat\theta^k\in\Gamma\) and \(a_k>0\) such that
\(x_k(\hat\theta^k)\geq a_k\). Incentive compatibility then gives $U(\theta)
\geq
\theta\cdot x(\hat\theta^k)
\geq
a_k\theta_k.$ Consequently, defining
\[
\uul M
:=
\frac{1}{|B|}
\min_{k\in B}a_k
>0,
\]
we obtain
\begin{equation}
U(\theta)
\geq
\uul M\sum_{k\in B}\theta_k.
\label{eq:direct-lower-bound-U}
\end{equation}
On the other hand, $U(\theta)
\leq
(\sum_{k\in B}\theta_k)
(\sum_{k\in B}x_k(\theta)).$ Since \(\sum_{k\in B}\theta_k=0\) only on a null set,
\eqref{eq:direct-lower-bound-U} implies
\begin{equation}
\sum_{k\in B}x_k(\theta)\geq\uul M
\ \ \text{for almost every }\theta\in D_q.
\label{eq:direct-lower-bound-allocation}
\end{equation}

We now use $z_q:=\int_{D_q}x(\theta)\,dG(\theta)$ to denote the supply freed up by types in \(D_q\) switching to the new option \(qe_j\). Consider the following allocation rule rebating this supply uniformly:
\[
\tilde x_q(\theta)
:=
\begin{cases}
z_q+qe_j, & \theta\in D_q,\\
z_q+x(\theta), & \theta\notin D_q.
\end{cases}
\]
This rule is incentive compatible and its
aggregate supply use is
\[
\pi\int_\Gamma\tilde x_q(\theta)\,dG(\theta)
=
\pi\int_\Gamma x(\theta)\,dG(\theta)
+
\varepsilon_qe_j,
\qquad
\varepsilon_q:=\pi qG(D_q)>0.
\]
It is therefore feasible under \(s+\varepsilon_qe_j\).

We now consider the welfare consequences of adding this option. The switching types weakly gain from the new option, so the welfare gain is at least the value of adding \(z_q\) to every allocation. Thus,
\[
W(s+\varepsilon_qe_j)-W(s)
\geq
\pi\int_\Gamma\lambda(\theta)\theta\cdot z_q\,dG(\theta).
\]
For \(k\in B\), define
\[
\beta_k
:=
\int_\Gamma\lambda(\theta)\theta_k\,dG(\theta)>0,
\qquad
\underline\beta^B:=\min_{k\in B}\beta_k>0.
\]
Since \(z_q\) contains only goods in \(B\), equation
\eqref{eq:direct-lower-bound-allocation} gives
\[
\pi\int_\Gamma\lambda(\theta)\theta\cdot z_q\,dG(\theta)
\geq
\pi\underline\beta^B\sum_{k\in B}(z_q)_k
\geq
\pi\underline\beta^B\uul M\,G(D_q).
\]
Using \(\varepsilon_q=\pi qG(D_q)\), we obtain $\frac{W(s+\varepsilon_qe_j)-W(s)}{\varepsilon_q}
\geq
\frac{\underline\beta^B\uul M}{q}.$ Now, choose any sequence \(q_n\downarrow0\), and define
\(\varepsilon_n:=\varepsilon_{q_n}\). Since \(G(D_{q_n})\leq1\), we have $0<\varepsilon_n
=
\pi q_nG(D_{q_n})
\leq
\pi q_n
\to 0 .$ At the same time, $\frac{W(s+\varepsilon_ne_j)-W(s)}{\varepsilon_n}
\geq
\frac{\underline\beta^B\uul M}{q_n}
\to+\infty,$ so the difference quotient diverges to \(+\infty\) along the sequence
\(\varepsilon_n\to 0\). It remains to show that the same holds for every sufficiently small
\(\varepsilon>0\). Since \(W\) is
concave, $h(\varepsilon)
:=
\frac{W(s+\varepsilon e_j)-W(s)}{\varepsilon}$ is non-increasing in \(\varepsilon>0\). Hence, whenever
\(0<\varepsilon\leq\varepsilon_n\), we have $h(\varepsilon)\geq h(\varepsilon_n).$ Fix any \(K>0\). Since \(h(\varepsilon_n)\to+\infty\), there exists \(n\)
such that \(h(\varepsilon_n)>K\). Therefore, for every
\(0<\varepsilon\leq\varepsilon_n\), we have $h(\varepsilon)\geq h(\varepsilon_n)>K.$ Since \(K>0\) was arbitrary, it follows that $\lim_{\varepsilon\downarrow0}
\frac{W(s+\varepsilon e_j)-W(s)}{\varepsilon}
=
+\infty.$
\end{proof}

\paragraph{Every group is either active or excluded.}
Suppose toward a contradiction that group \(r\) receives none of good \(j\) but a strictly positive amount of other goods: \(s_j^r=0\) and \(s^r\neq0\). At an optimum, the whole supply is allocated, so there is some group \(r'\neq r\) for which \(s_j^{r'}>0\).

By Lemma~\ref{lem:Wr-properties}, the function
\(t\mapsto W_{r'}(s^{r'}+te_j)\) is finite and concave. Since
\(s_j^{r'}>0\), it is locally Lipschitz at zero. Thus, for some
\(L<\infty\) and all sufficiently small \(\varepsilon>0\), we have $W_{r'}(s^{r'})
-
W_{r'}(s^{r'}-\varepsilon e_j)
\leq
L\varepsilon.$ By Lemma~\ref{lem:infinite-marginal-missing-good}, for all sufficiently
small \(\varepsilon>0\), $W_r(s^r+\varepsilon e_j)-W_r(s^r)>L\varepsilon.$
Then taking \(\varepsilon<s_j^{r'}\) and transferring \(\varepsilon\) units of
good \(j\) from group \(r'\) to group \(r\) strictly increases welfare; contradiction.

\paragraph{Sufficient condition for exclusion.}
For any supply vector \(s^r\),
\begin{equation}
W_r(s^r)
\leq
\sum_{i=1}^N\overline v_{ri}s_i^r,
\label{eq:group-welfare-upper-bound}
\end{equation}
because every allocation feasible under \(s^r\) satisfies
\[
\pi_r\int_\Gamma
\lambda_r(\theta)\theta\cdot x_r(\theta)\,dG_r(\theta)
\leq
\sum_{i=1}^N
\overline v_{ri}
\pi_r\int_\Gamma x_{ri}(\theta)\,dG_r(\theta)
\leq
\sum_{i=1}^N\overline v_{ri}s_i^r.
\]

Suppose, toward a contradiction, that group \(r\) is active. Then, by the above, \(s^r\gg0\). For each good \(i\), choose
\(k(i)\neq r\) such that $\eta_{k(i)i}=\max_{k\neq r}\eta_{ki}.$ Remove \(s^r\) from group \(r\) and give every agent in group \(k(i)\) the
common amount \(s_i^r/\pi_{k(i)}\) of good \(i\). This is incentive
compatible and generates welfare
\[
\sum_{i=1}^N
\left(\max_{k\neq r}\eta_{ki}\right)s_i^r
>
\sum_{i=1}^N\overline v_{ri}s_i^r
\geq
W_r(s^r).
\]
The reallocation therefore strictly increases welfare; contradiction. Hence \(s^r=0\).

\paragraph{Sufficient condition for activity.}
Fix a nonempty set \(B\) and consider the following allocation: assign every good \(i\in B\) entirely to group \(r\); for each \(i\notin B\), assign good \(i\) entirely to a group
\(k(i)\in\argmax_{k\neq r}\eta_{ki}\). Distributing each group's assigned
bundle uniformly among its members is incentive compatible and generates
welfare
\begin{equation}
\hat W
=
\sum_{i\in B}s_i\eta_{ri}
+
\sum_{i\notin B}s_i\max_{k\neq r}\eta_{ki}.
\label{eq:comparative-advantage-test-welfare}
\end{equation}

By contrast, the welfare of any allocation excluding group \(r\) is at most $\sum_{i=1}^N
s_i\max_{l\neq r}\overline v_{li}.$ Indeed, if \(s^l\) is the supply assigned to group \(l\neq r\), then its
welfare is at most
\[
\sum_{l\neq r}\sum_{i=1}^N\overline v_{li}s_i^l
\leq
\sum_{i=1}^N
\left(\max_{l\neq r}\overline v_{li}\right)
\sum_{l\neq r}s_i^l
\leq
\sum_{i=1}^N
s_i\max_{l\neq r}\overline v_{li}.
\]
Condition~\eqref{eq:comparative-advantage-participation-difference} is
equivalent to $\hat W
>
\sum_{i=1}^N
s_i\max_{l\neq r}\overline v_{li}.$ Thus, some feasible allocation serving group \(r\) generates strictly more
welfare than every allocation that excludes it. Group \(r\) must therefore
be active in every optimal supply split.

\subsection{Proof of Proposition \ref{prop:generic-common-pseudomarket}}

\begin{lemma}
\label{lem:Wr-differentiable-at-pure}
Fix a group \(r\), and let \(s\gg0\). Suppose that a pure option mechanism
with option vector \(q\in\mathbb R_{++}^N\) clears supply \(s\) and is
optimal for \(W_r(s)\). Then \(W_r\) is differentiable at \(s\), and
\[
\nabla W_r(s)
=
c^r(q)
=
J^r(q)^{-1}A^r(q),
\]
where \(J^r(q)\) and \(A^r(q)\) are computed using group-\(r\)
primitives.
\end{lemma}

\begin{proof}
By Lemma~\ref{lem:Wr-properties}, \(W_r\) is finite and concave. Fix any \(p\in\partial^+W_r(s)\) and, for any within-group incentive-compatible allocation rule \(x\), let
\[
t:=\pi_r\int_\Gamma x(\theta)\,dG_r(\theta).
\]
Then
\[
\pi_r\int_\Gamma
\lambda_r(\theta)\theta\cdot x(\theta)\,dG_r(\theta)
\leq
W_r(t)
\leq
W_r(s)+p\cdot(t-s).
\]
The optimal pure option mechanism attains equality and therefore maximizes
the corresponding Lagrangian. Differentiating along \(q+\varepsilon h\)
and applying Lemma~\ref{lem:pure_option_derivatives} gives $0
=
\pi_r
\bigl(A^r(q)-J^r(q)p\bigr)\cdot h$ for every  $h\in\mathbb R^N.$ Hence, by Fact~\ref{fact:supply_cost_properties}, we have $p
=
J^r(q)^{-1}A^r(q)
=
c^r(q).$ Since \(p\in\partial^+W_r(s)\) was arbitrary,
\(\partial^+W_r(s)=\{c^r(q)\}\). Thus, \(W_r\) is differentiable at \(s\),
with $\nabla W_r(s)=c^r(q).$
\end{proof}

Let
\[
\mathcal Q
:=
\left\{
q\in\mathbb R_{++}^N:
\sum  q_i=1
\right\},
\qquad
\mathcal Q_\delta
:=
\left\{
q\in\mathcal Q:
q_i\geq\delta\text{ for every }i
\right\} \ \ \  \text{for }\delta>0,
\]
and, for \(q\in\mathcal Q\), define
\[
\Gamma_i(q)
:=
\left\{
\theta\in\Gamma:
\theta_iq_i\geq\theta_jq_j
\text{ for every }j
\right\}.
\]
Fix baseline primitives $((\lambda_r^*,g_r^*)_{r=1}^R,s^*)$ and parameterize nearby primitives by
\[
\lambda_r=e^{l_r}\lambda_r^*,
\qquad
g_r=
\frac{e^{h_r}g_r^*}
{\int_\Gamma e^{h_r}g_r^*\,dm},
\qquad
r=1,\ldots,R,
\]
together with a nearby supply vector \(s\). Suppose that a CEUI is optimal
at such primitives. After normalizing the common inverse-price vector, we can apply the argument from the
proof of Fact~\ref{fact:pure_implementations} separately to each group. It follows that there exist
\(q\in\mathcal Q\) and budgets \(b_1,\ldots,b_R\geq0\) such that group \(r\)
is offered the pure option vector \(b_rq\). The supply of good \(i\) assigned
to group \(r\) is
\begin{equation}
\label{eq:common-price-group-supply}
s_i^r(q,b_r)
:=
\pi_rb_rq_i
\int_{\Gamma_i(q)}g_r(\theta)\,d\theta,
\qquad
i=1,\ldots,N,
\end{equation}
and market clearing requires
\begin{equation}
\label{eq:common-price-supply-clearing}
\sum_{r=1}^R s^r(q,b_r)=s.
\end{equation}

Suppose that the active set of this CEUI is
\(\mathcal A\subseteq\{1,\ldots,R\}\), with \(|\mathcal A|\geq2\); that is,
\(b_r>0\) if and only if \(r\in\mathcal A\). Then
\(s^r(q,b_r)\gg0\) for every \(r\in\mathcal A\). Because the CEUI is optimal
for the overall supply-splitting problem, its within-group mechanism must be
optimal at its assigned supply. By the preceding lemma, $\nabla W_r\bigl(s^r(q,b_r)\bigr)
=
c^r(b_rq)
=
c^r(q)$ for $r\in\mathcal A,$ where the second equality follows because rescaling a pure option vector
does not change its supply-cost vector.

Since the supplies assigned to all active groups are interior, the
first-order conditions for the outer problem require their gradients to
coincide. Fixing \(r_0\in\mathcal A\), we obtain
\begin{equation}
\label{eq:common-price-cost-equalities}
c^{r_0}(q)=c^r(q),
\qquad
r\in\mathcal A\setminus\{r_0\}.
\end{equation}
Moreover, \(b_r=0\) for every \(r\notin\mathcal A\), so
\eqref{eq:common-price-supply-clearing} reduces to $\sum_{r\in\mathcal A}s^r(q,b_r)=s.$ We therefore define
\[
\mathcal F_{\mathcal A}
\left(
(l_r,h_r)_{r=1}^R,
s,
q,
(b_r)_{r\in\mathcal A}
\right)
:=
\begin{pmatrix}
\bigl(c^{r_0}(q)-c^r(q)\bigr)_
{r\in\mathcal A\setminus\{r_0\}}
\\[0.6em]
\displaystyle
\sum_{r\in\mathcal A}s^r(q,b_r)-s
\end{pmatrix}.
\]
The map \(\mathcal F_{\mathcal A}\) takes values in
\(\mathbb R^{|\mathcal A|N}\). Note that a CEUI with active set
\(\mathcal A\) can be optimal only at a zero of \(\mathcal F_{\mathcal A}\). I now show that the existence of such zeros is nongeneric.

\begin{lemma}
\label{lem:common-pseudomarket-smoothness}
Fix \(\delta>0\). On a sufficiently small neighborhood of the baseline
primitives, \(\mathcal F_{\mathcal A}\) is continuously differentiable
jointly in the primitive perturbations, \(q\in\mathcal Q_\delta\), and
\((b_r)_{r\in\mathcal A}\).
\end{lemma}

\begin{proof}
For each \(i\), parameterize \(\Gamma_i(q)\) by
\(z\in[0,1]^{N-1}\), with coordinates indexed by \(j\neq i\), through
\[
\Phi_i(q,z)_i
=
\frac{1}
{1+\sum_{j\neq i}(q_i/q_j)z_j},
\qquad
\Phi_i(q,z)_j
=
\frac{(q_i/q_j)z_j}
{1+\sum_{k\neq i}(q_i/q_k)z_k}.
\]
This is a bijective parameterization of \(\Gamma_i(q)\), and the face
\(z_j=1\) parameterizes \(\Gamma_i(q)\cap\Gamma_j(q)\). On
\(\mathcal Q_\delta\), the map \(\Phi_i\), together with its interior and
face Jacobians, depends smoothly on \((q,z)\).

Pulling the integrals defining \(M_i^r(q)\), \(A_i^r(q)\), and
\(T_{ij}^r(q)\) back to the fixed cubes \([0,1]^{N-1}\) and their faces, and differentiating under the integral sign, shows that these objects are continuously differentiable jointly in \((l_r,h_r,q)\).

Thus \(A^r(q)\), \(J^r(q)\), and \(M_i^r(q)\) are jointly \(C^1\).
By Fact~\ref{fact:supply_cost_properties}, \(J^r(q)\) is invertible, so $c^r(q)=J^r(q)^{-1}A^r(q)$ and $s_i^r(q,b_r)=\pi_rb_rq_iM_i^r(q)$ are jointly \(C^1\). Hence \(\mathcal F_{\mathcal A}\) is jointly \(C^1\).
\end{proof}

\paragraph{Uniform bounds on the CEUI parameters.}
After shrinking the primitive neighborhood if necessary, there exist
\(\underline s,\overline s>0\) such that $s_i\geq\underline s$ and $\sum_{i=1}^Ns_i\leq\overline s$ for every nearby supply vector \(s\). We first obtain a uniform upper bound on the budgets. Fix \(r\) and
\(q\in\mathcal Q\), and choose \(i\) such that \(q_i=\max_jq_j\). Then
\(q_i\geq1/N\) and $\left\{
\theta\in\Gamma:
\theta_i\geq\theta_j
\text{ for every }j
\right\}
\subseteq
\Gamma_i(q).$ The mass of the set on the left is uniformly bounded away from zero under
sufficiently small primitive perturbations. Hence, there exists
\(\kappa>0\) such that $\sum_{j=1}^N
q_j\int_{\Gamma_j(q)}g_r(\theta)\,d\theta
\geq\kappa$ for every \(r\), every \(q\in\mathcal Q\), and every nearby primitive
vector. Therefore, for each \(r\in\mathcal A\), market clearing gives
\[
\overline s
\geq
\sum_{i=1}^Ns_i
\geq
\sum_{i=1}^Ns_i^r(q,b_r)
=
\pi_rb_r
\sum_{i=1}^N
q_i\int_{\Gamma_i(q)}g_r(\theta)\,d\theta
\geq
\pi_rb_r\kappa;
 \qquad \text{thus} \qquad
b_r\leq\frac{\overline s}{\pi_r\kappa}
=:\overline b_r,
\quad r\in\mathcal A.
\]

We next obtain a uniform lower bound on the components of \(q\). For every
\(i\), market clearing and the preceding budget bounds imply $s_i
=
q_i\sum_{r\in\mathcal A}
\pi_rb_r
\int_{\Gamma_i(q)}g_r(\theta)\,d\theta
\leq
q_i\sum_{r\in\mathcal A}\pi_r\overline b_r.$ Since $s_i\geq\underline s$, we have $q_i
\geq
\frac{\underline s}
{\sum_{r\in\mathcal A}\pi_r\overline b_r}.$ Set $\underline q
:=
\min\{
{\underline s}/
{\sum_{r\in\mathcal A}\pi_r\overline b_r}, \ 
{1}/{2N}
\}
>0.$ Consequently, every zero of \(\mathcal F_{\mathcal A}\) associated with
nearby primitives lies in the compact set
\[
K_{\mathcal A}
:=
\left\{
\left(q,(b_r)_{r\in\mathcal A}\right):
q\in\mathcal Q,\quad
q_i\geq\underline q,\quad
0\leq b_r\leq\overline b_r
\right\}.
\]

\paragraph{Transversality.}
After shrinking the primitive neighborhood if necessary, we can apply
Lemma~\ref{lem:common-pseudomarket-smoothness} with
\(\delta=\underline q/2\) to conclude that \(\mathcal F_{\mathcal A}\) is
jointly \(C^1\) in the primitive perturbations and the endogenous parameters.
Since every \(q\) appearing in \(K_{\mathcal A}\) satisfies
\(q_i\geq\underline q>\underline q/2\), this smoothness holds on a
neighborhood of \(K_{\mathcal A}\).

Fix a zero of \(\mathcal F_{\mathcal A}\). We first show that variations in
the primitives span the entire codomain of \(\mathcal F_{\mathcal A}\).
For each \(r\in\mathcal A\setminus\{r_0\}\), a variation
\(k_r\in C^1(\Gamma)\) of \(l_r\) gives
\[
D_{l_r}c^r(q)[k_r]
=
J^r(q)^{-1}
\left(
\int_{\Gamma_i(q)}
\theta_i\lambda_r(\theta)
k_r(\theta)g_r(\theta)\,d\theta
\right)_{i=1}^N.
\]
For each \(i\), choose a \(C^1\) function supported in
\(\Gamma_i(q)^\circ\) for which the corresponding integral is nonzero.
Scaling these functions independently shows that the vector of integrals
can take any value in \(\mathbb R^N\). Since \(J^r(q)\) is invertible,
\(D_{l_r}c^r(q)\) is therefore surjective onto \(\mathbb R^N\).

A variation in \(l_r\) affects only the \(r\)-th cost-difference block
\(c^{r_0}(q)-c^r(q)\). In particular, it does not affect the other
cost-difference blocks or the supply-clearing block. Hence independent
variations in \((l_r)_{r\in\mathcal A\setminus\{r_0\}}\) span all
\((|\mathcal A|-1)N\) cost-difference coordinates of
\(\mathcal F_{\mathcal A}\). Variations in \(s\) span the remaining \(N\)
coordinates, since $D_s\mathcal F_{\mathcal A}[\dot s]
=
(0,-\dot s)^T.$ Thus the derivative of \(\mathcal F_{\mathcal A}\) with respect to the
primitive variables is surjective at every zero, and so \(\mathcal F_{\mathcal A}\) is transverse to \(\{0\}\).

Now, stratify \(K_{\mathcal A}\) by the relative interiors of its finitely many
faces. Let \(\mathcal P\) denote the sufficiently small open neighborhood
of primitive perturbations under consideration. For every stratum
\(\Sigma\) of \(K_{\mathcal A}\), consider the restriction $\mathcal F_{\mathcal A}^{\Sigma}:
\mathcal P\times\Sigma
\to
\mathbb R^{|\mathcal A|N}.$ By the preceding argument, this map is \(C^1\) and transverse to \(\{0\}\).
The Banach-manifold parametric transversality theorem
\citep[see, e.g.,][]{abraham1967transversal} therefore implies that, for a
residual subset of \(\mathcal P\), the primitive-specific restriction to
\(\Sigma\) is transverse to \(\{0\}\). Since $\dim\Sigma
\leq N+|\mathcal A|-1
<
|\mathcal A|N,$ such a transverse restriction cannot have a zero. Intersecting over the
finitely many strata therefore gives a residual, and hence dense, set of
primitive perturbations for which \(\mathcal F_{\mathcal A}\) has no zero
in \(K_{\mathcal A}\). Next, define
\[
\mathcal U_{\mathcal A}
:=
\left\{
\xi\in\mathcal P:
\min_{(q,b)\in K_{\mathcal A}}
\left\|\mathcal F_{\mathcal A}(\xi,q,b)\right\|>0
\right\}.
\]
The preceding dense set is contained in \(\mathcal U_{\mathcal A}\), so
\(\mathcal U_{\mathcal A}\) is dense. Compactness of \(K_{\mathcal A}\)
and joint continuity of \(\mathcal F_{\mathcal A}\) imply that
\(\mathcal U_{\mathcal A}\) is also open. Therefore, for each fixed
\(\mathcal A\subseteq\{1,\ldots,R\}\) with \(|\mathcal A|\geq2\), there is
an open and dense set of nearby primitives for which no optimal CEUI has
active set \(\mathcal A\). There are only finitely many such active sets, so
intersecting these open and dense sets shows that optimality of a CEUI with
at least two active groups is nongeneric at
\(((\lambda_r^*,g_r^*)_{r=1}^R,s^*)\).

\section{Omitted proofs from Section \ref{sec:proofdetails}}

\subsection{Proof of Proposition \ref{thm:icchar} and Corollary
\ref{lem:utility-space-corr}}

Necessity was shown in the main text, so it suffices to prove sufficiency and
Corollary~\ref{lem:utility-space-corr}. The corollary is implied by the following lemma.
\begin{lemma}\label{lem:utility-space-formulation}
Let \(U:\Gamma\to\mathbb R_+\) be convex and satisfy
\eqref{eq:pseudomono}. Then \(U\) is Lipschitz continuous. Moreover, there
exists a nonempty compact convex set \(X\subseteq\mathbb R_+^N\) such that
\begin{equation}\label{eq:support-representation-U}
U(\theta)
=
\max_{z\in X}\theta\cdot z
\qquad
\text{for every }\theta\in\Gamma.
\end{equation}
At every \(\theta\in\Gamma^\circ\) where \(\nabla_HU(\theta)\) exists, the
maximizer in \eqref{eq:support-representation-U} is unique and equals
\[
x_U(\theta)
=
\nabla_H U(\theta)
-
\mathbf 1\bigl(\nabla_H U(\theta)\cdot\theta-U(\theta)\bigr).
\]
Consequently, \(x_U(\theta)\in\mathbb R_+^N\), \(x_U\) is uniformly bounded,
and \(U(\theta)=\theta\cdot x_U(\theta)\) a.e. Furthermore,
every allocation rule satisfying \eqref{eq:ICprime} and implementing \(U\)
coincides with \(x_U\) a.e.
\end{lemma}

\begin{proof}
Extend \(U\) homogeneously to \(\mathbb R_+^N\) by
\[
\tilde U(z)
:=
\begin{cases}
(\sum_lz_l) U\bigl(z/\sum_lz_l\bigr),&z\neq0,\\[3pt]
0,&z=0.
\end{cases}
\]
Since \(U\) is convex and nonnegative, \(\tilde U\) is nonnegative,
convex, positively homogeneous, and hence also sublinear.

I first show that \(\tilde U\) is coordinatewise non-decreasing. It is
enough to prove that $\tilde U(r+te_j)\ge\tilde U(r)$ for every $r\in\mathbb R_+^N,\ t\ge0.$ If \(r=0\), this follows from nonnegativity. If \(r\) has a positive
coordinate \(r_i\) with \(i\neq j\), define
\[
\alpha:=\frac{r}{\sum_kr_k},
\qquad
\beta:=\frac{r+te_j}{\sum_kr_k+t}.
\]
Then \(\alpha\succ_i\beta\), so \eqref{eq:pseudomono} gives $\frac{U(\beta)}{\beta_i}
\ge
\frac{U(\alpha)}{\alpha_i}.$ Since $\alpha_i=\frac{r_i}{\sum_kr_k}$ and $\beta_i=\frac{r_i}{\sum_kr_k+t},$ this is equivalent to
\(\tilde U(r+te_j)\ge\tilde U(r)\). The remaining case, in which
\(r\) is supported only on coordinate \(j\), follows from positive
homogeneity.

I now show that $U$ is Lipschitz on $\Gamma$. Let \(M:=\max_iU(e_i)\). For every \(v\in\mathbb R_+^N\), sublinearity gives
\[
0\le\tilde U(v)
\le
\sum v_i\tilde U(e_i)
\le
M\|v\|_1.
\]
Moreover, coordinatewise monotonicity and sublinearity imply
\[
\tilde U(v)
\le
\tilde U\bigl(w+(v-w)_+\bigr)
\le
\tilde U(w)+M\|v-w\|_1.
\]
Interchanging \(v\) and \(w\) yields $|\tilde U(v)-\tilde U(w)|
\le
M\|v-w\|_1.$

Now, since \(\tilde U\) is finite, sublinear, and coordinatewise
non-decreasing, it is the support function of a nonempty compact convex set
\(X\subseteq[0,M]^N\).\footnote{This follows from the support-function
representation of closed sublinear functions; see
\citet[Theorem~13.2]{rockafellar1997convex}. Apply the theorem to the
sublinear extension \(v\mapsto\tilde U(v_+)\).}
Thus $\tilde U(v)=\max_{z\in X}v\cdot z$ for every $v\in\mathbb R_+^N.$ Restricting this representation to \(\Gamma\) gives
\eqref{eq:support-representation-U}.

Finally, fix a point \(\theta\in\Gamma^\circ\) at which
\(\nabla_HU(\theta)\) exists, and let \(z\in X\) maximize
\(\theta\cdot z\). For every \(\theta'\in\Gamma\),
\[
U(\theta')
\ge
\theta'\cdot z
=
U(\theta)+z\cdot(\theta'-\theta).
\]
Thus, the projection of \(z\) onto the tangent space of \(\Gamma\) equals
\(\nabla_HU(\theta)\). Together with \(\theta\cdot z=U(\theta)\), this
uniquely determines \(z\) and gives
\[
z
=
\nabla_H U(\theta)
-
\mathbf 1\bigl(\nabla_H U(\theta)\cdot\theta-U(\theta)\bigr)
=
x_U(\theta).
\]
Hence \(x_U(\theta)\in X\subseteq\mathbb R_+^N\), \(x_U\) is uniformly
bounded, and \(U(\theta)=\theta\cdot x_U(\theta)\) almost everywhere.

Now, let \(x:\Gamma\to\mathbb R_+^N\) be any allocation rule satisfying
\eqref{eq:ICprime} and implementing \(U\). Incentive compatibility implies
that, for every \(\theta,\theta'\in\Gamma\), we have $U(\theta')
\ge
U(\theta)+x(\theta)\cdot(\theta'-\theta).$ At every \(\theta\in\Gamma^\circ\) where \(\nabla_HU(\theta)\) exists, this
implies that the projection of \(x(\theta)\) onto the tangent space of
\(\Gamma\) equals \(\nabla_HU(\theta)\). Together with
\(\theta\cdot x(\theta)=U(\theta)\), this uniquely determines \(x(\theta)\)
and gives
\[
x(\theta)
=
\nabla_H U(\theta)
-
\mathbf 1\bigl(\nabla_H U(\theta)\cdot\theta-U(\theta)\bigr)
=
x_U(\theta).
\]
Since \(U\) is Lipschitz, Rademacher's theorem implies that it is
differentiable almost everywhere. Therefore, \(x=x_U\) almost everywhere.
\end{proof}
It remains to show sufficiency. To that end, for each \(\theta\in\Gamma\), choose $x(\theta)\in\argmax_{z\in X}\theta\cdot z.$ The measurable maximum theorem gives a measurable selection. By
\eqref{eq:support-representation-U}, $U(\theta)=\theta\cdot x(\theta)$ for every $\theta\in\Gamma$. Moreover, for every \(\theta,\theta'\in\Gamma\), we have $\theta\cdot x(\theta)
=
\max_{z\in X}\theta\cdot z
\ge
\theta\cdot x(\theta').$ Thus \(x\) satisfies \eqref{eq:ICprime} and implements \(U\), proving
sufficiency.

\subsection{Proof of Proposition \ref{prop:supply_cost_characterization}}

For \(U\in\mathcal U\), write
\[
\Phi(U):=\int_\Gamma\lambda U g\,d\theta,
\qquad
\Pi(U):=\int_\Gamma x_U g\,d\theta.
\]
Recall that \(\Pi(U_{\mathrm{pure}})=s\) and, by
Fact~\ref{fact:supply_cost_properties}, \(c=J^{-1}A\gg0\).

Suppose first that \(U_{\mathrm{pure}}\) solves
\eqref{eq:problem-supplycost}. If \(U\) is feasible in
\eqref{eq:U-objective}, then \(\Pi(U)\le s\), and hence
\[
\Phi(U)
\le
\Phi(U)-c\cdot(\Pi(U)-s)
\le
\Phi(U_{\mathrm{pure}}).
\]
Thus \(U_{\mathrm{pure}}\) solves \eqref{eq:U-objective}.

Conversely, suppose that \(U_{\mathrm{pure}}\) solves
\eqref{eq:U-objective}, and define
\[
P(r):=\sup\{\Phi(U):U\in\mathcal U,\ \Pi(U)\le r\},
\qquad r\in\mathbb R_+^N.
\]
This function is finite, increasing, and concave. Indeed, if
\(\Pi(U)\le r\), then, writing \(\bar\lambda\) for an upper bound on
\(\lambda\),
\[
\Phi(U)
\le
\bar\lambda\int_\Gamma\sum_i(x_U)_i\,dG
\le
\bar\lambda\sum_i r_i.
\]
Monotonicity is immediate, while concavity follows from the convexity of
\(\mathcal U\) and the fact that $x_{tU+(1-t)V}=tx_U+(1-t)x_V$ a.e. for $t\in[0,1].$ Now, since \(s\gg0\), \(P\) has a supergradient \(c^*\) at \(s\). Because \(P\) is
increasing, \(c^*\ge0\), and $P(r)\le P(s)+c^*\cdot(r-s)$ for every $r\in\mathbb R_+^N.$ Moreover, \(P(s)=\Phi(U_{\mathrm{pure}})\), so, for every \(U\in\mathcal U\),
\[
\Phi(U)-c^*\cdot(\Pi(U)-s)
\le
P(\Pi(U))-c^*\cdot(\Pi(U)-s)
\le
\Phi(U_{\mathrm{pure}}).
\]
Thus \(U_{\mathrm{pure}}\) maximizes the Lagrangian with multiplier \(c^*\). It remains to identify this multiplier.

Note that for every \(h\in\mathbb R^N\), \(U_{q+th}\in\mathcal U\) for all sufficiently
small \(t\). Since \(U_{\mathrm{pure}}=U_q\) maximizes the Lagrangian,
Lemma~\ref{lem:pure_option_derivatives} implies
\[
0
=
A\cdot h
-
\sum c_j^*\sum_iJ_{ij}h_i
=
\bigl(A-Jc^*\bigr)\cdot h.
\]
Since \(h\) was arbitrary, \(Jc^*=A\). Therefore
\(c^*=J^{-1}A=c\), so \(U_{\mathrm{pure}}\) solves
\eqref{eq:problem-supplycost}.

\subsection{Proof of Lemma \ref{lemma:measure-objective}}

\emph{Step 1: Measure formulation.}
Dropping the constant \(c\cdot s\), the objective is
\[
\int_\Gamma \lambda U g\,d\theta
-
\int_\Gamma c\cdot x_U\,g\,d\theta.
\]
By the definition of \(x_U\),
\[
\int_\Gamma c\cdot x_U\,g\,d\theta
=
\int_\Gamma
\big(c-(\sum c_j)\theta\big)\cdot\nabla_HU\,g\,d\theta
+
\Bigl(\sum c_j\Bigr)\int_\Gamma Ug\,d\theta.
\]
The vector field $\big(c-(\sum c_j)\theta\big)g$ is tangent to \(H\), and \(U\) is Lipschitz by
Corollary~\ref{lem:utility-space-corr}. Hence, by
Theorem~\ref{th:divonH} and integration by parts,
$\int_\Gamma \lambda Ug\,d\theta
-
\int_\Gamma c\cdot x_U\,g\,d\theta$ is equal to
\[
\int_\Gamma U
\big[
(\lambda-\sum c_j)g
+
\operatorname{div}
\big(
\big(c-(\sum c_j)\theta\big)g
\big)
\big]d\theta-
\int_{\partial\Gamma}
U\big(c-(\sum c_j)\theta\big)g\cdot\nu\,d\sigma\\
=
\int_\Gamma\frac{U}{\theta_*}\,d\mu.
\]
The equality follows from the definition of \(\mu\) in
\eqref{eq:signedmeas}.

\emph{Step 2: Balance on each region \(\Gamma_i\).}
We now apply Theorem~\ref{th:divonH} on \(\Gamma_i\) with $\eta(\theta)=\theta_i$ and $X(\theta)
=
\bigl(c-(\textstyle\sum c_j)\theta\bigr)g(\theta).$ Since \(\nabla_H\theta_i=e_i-\tfrac1N\mathbf1\) and \(X(\theta)\in TH\),
substituting the resulting integration-by-parts identity into the definition
of \(\mu(\Gamma_i)\) gives
\[
\mu(\Gamma_i)
=
A_i-c_iM_i
+
\int_{\partial\Gamma_i\setminus\partial\Gamma}
\theta_i
\bigl(c-(\textstyle\sum c_j)\theta\bigr)g\cdot\nu\,d\sigma.
\]
Up to lower-dimensional intersections, the interior boundary of
\(\Gamma_i\) is the union of the interfaces
\(\Gamma_i\cap\Gamma_k\), \(k\neq i\). Along such an interface,
\[
q_i\theta_i=q_k\theta_k,
\qquad
\nu_{ik}^{(i)}
=
-\frac{\nabla_H(q_i\theta_i-q_k\theta_k)}
{\|\nabla_H(q_i\theta_i-q_k\theta_k)\|}.
\]
Hence
\[
\theta_i
\bigl(c-(\textstyle\sum c_j)\theta\bigr)g
\cdot\nu_{ik}^{(i)}
=
(q_kc_k-q_ic_i)g
\frac{\theta_i}
{\|\nabla_H(q_i\theta_i-q_k\theta_k)\|},
\]
where the terms involving \(\sum c_j\) cancel because
\(q_i\theta_i=q_k\theta_k\). Moreover,
\[
\|\nabla_H(q_i\theta_i-q_k\theta_k)\|^2
=
q_i^2+q_k^2-\frac{(q_i-q_k)^2}{N}.
\]
Therefore, by the definition of \(T_{ik}\),
\[
\mu(\Gamma_i)
=
A_i-c_iM_i
+
\sum_{k\neq i}(q_kc_k-q_ic_i)T_{ik}.
\]
The right-hand side is zero by the \(i\)-th row of \(Jc=A\). Thus
\(\mu(\Gamma_i)=0\) for every \(i\).

\subsection{Proof of Proposition \ref{prop:dualcert}}

$(\Leftarrow)$. Suppose \((\gamma_1,\dots,\gamma_N)\) is a dual
certificate. For every \(U\) feasible in \eqref{eq:primal},
\[
\int_\Gamma\frac{U}{\theta_*}\,d\mu
\le
\sum_{i=1}^N\int_{R_i}B_iU\,d\gamma_i
\le0,
\]
because \(B_iU\le0\) on \(R_i\) and \(\gamma_i\ge0\). Since $\int_\Gamma\frac{U_{\mathrm{pure}}}{\theta_*}\,d\mu=0$ by \eqref{eq:primalzeropure}, \(U_{\mathrm{pure}}\) is primal optimal.

$(\Rightarrow)$. We introduce the
ordered Banach space $Y:=\prod_{i=1}^N C( R_i)$ with the product sup norm and the pointwise order. The primal constraints are
\( B_iU\le0\) on \( R_i\), for every \(i\). Note also that
$\mathcal U$, defined in \eqref{eq:UL}, can be written as
\[
\mathcal{U}
=
\left\{
U\in K:
 B_iU\le0\text{ on } R_i
\text{ for every }i
\right\}.
\]
We first show the following lemma:

\begin{lemma}
\label{lem:R-repair}
There exists \(C<\infty\) such that the following holds. Let
\(\bar U\in K\) and \(\delta\ge0\). If $B_i\bar U\le\delta$ on $R_i$ for every $i$, then there exists \(U\in\mathcal U\) such that
\(\|\bar U-U\|_\infty\le C\delta\).
\end{lemma}
\begin{proof}
Extend \(\bar U\) homogeneously to \(\mathbb R_+^N\) by
\[
W(z)
:=
\begin{cases}
(\sum_lz_l)\bar U\bigl(z/\sum_lz_l\bigr),&z\neq0,\\[3pt]
0,&z=0.
\end{cases}
\]
Since \(\bar U\) is convex and nonnegative, \(W\) is nonnegative and
sublinear.

I first show that \(W\) is approximately coordinatewise increasing. Fix
\(j\), let \(r\in\mathbb R_+^N\) satisfy \(r_j=0\), and take
\(0\le s\le t\) with \(\sum_lr_l+t\le1\). First, note that
\begin{equation}\label{eq:approx-coordinate-monotonicity}
W(r+se_j)
\le
W(r+te_j)+(N-1)\delta.
\end{equation}
Indeed, if \(r=0\), this follows from homogeneity and nonnegativity. Otherwise, choose
\(i\neq j\) such that
\(r_i\ge\frac{\sum_lr_l}{N-1}\), and set
\[
\alpha:=\frac{r}{\sum_lr_l},
\qquad
\beta:=\frac{r+te_j}{\sum_lr_l+t}.
\]
Then \((\beta,\alpha)\in R_i\), so
\(\beta_i\bar U(\alpha)-\alpha_i\bar U(\beta)\le\delta\). Using the
definitions of \(\alpha\), \(\beta\), and \(W\), we obtain
\[
W(r)-W(r+te_j)
\le
\delta\frac{(\sum_lr_l)(\sum_lr_l+t)}{r_i}
\le
(N-1)\delta.
\]
Convexity of \(a\mapsto W(r+ae_j)\) on \([0,t]\) then gives
\eqref{eq:approx-coordinate-monotonicity}.

Now, applying \eqref{eq:approx-coordinate-monotonicity} successively to each
coordinate shows that, whenever \(0\le x\le y\) and \(\sum_ly_l\le1\),
\begin{equation}\label{eq:approx-monotone-cone}
W(x)\le W(y)+N(N-1)\delta.
\end{equation}
We now construct the desired \(U\). To that end, define the coordinatewise
monotone hull
\[
W^\uparrow(y)
:=
\sup\{W(x):0\le x\le y\}.
\]
Then, for every \(y\in\Gamma\), \eqref{eq:approx-monotone-cone} gives
\begin{equation}\label{eq:hull-close}
0
\le
W^\uparrow(y)-W(y)
\le
N(N-1)\delta.
\end{equation}
The function \(W^\uparrow\) is coordinatewise increasing and positively
homogeneous. It is also subadditive. Indeed, if \(0\le z\le y+y'\), define $z_k^1:=\min\{z_k,y_k\}$ and $z^2:=z-z^1.$ Then \(0\le z^1\le y\), \(0\le z^2\le y'\), and sublinearity of \(W\) gives
\[
W(z)
\le
W(z^1)+W(z^2)
\le
W^\uparrow(y)+W^\uparrow(y').
\]
Taking the supremum over \(z\le y+y'\) proves subadditivity.

Now, set \(U:=W^\uparrow|_\Gamma\). Since \(W^\uparrow\) is nonnegative and
sublinear, \(U\) is nonnegative and convex. It is also continuous: if
\(M:=\max_iW^\uparrow(e_i)\), then sublinearity and coordinatewise
monotonicity imply $|U(\theta)-U(\theta')|
\le
M\|\theta-\theta'\|_1$, so \(U\in K\). Moreover, \eqref{eq:hull-close} gives $\|\bar U-U\|_\infty
\le
N(N-1)\delta.$ 

It therefore remains to verify the constraints defining \(\mathcal U\). Fix
\((\theta,\theta')\in R_i\). If \(\theta_i=0\), then $B_iU(\theta,\theta')
=
-\theta_i'U(\theta)
\le0.$ If \(\theta_i>0\), let
\(z:=(\theta_i'/\theta_i)\theta\). By the definition of \(R_i\),
\(z\ge\theta'\) coordinatewise. Therefore,
\[
U(\theta')
\le
W^\uparrow(z)
=
\frac{\theta_i'}{\theta_i}U(\theta),
\]
so \(B_iU(\theta,\theta')\le0\). Thus \(U\in\mathcal U\), and the result
follows with \(C=N(N-1)\).
\end{proof}
We now prove existence of a certificate. By
\eqref{eq:primalzeropure} and the optimality of
\(U_{\mathrm{pure}}\),
\begin{equation}
\label{eq:KR-nonpositive-no-lip}
\int_\Gamma\frac{U}{\theta_*}\,d\mu\le0
\qquad\text{for every }U\in\mathcal U.
\end{equation}
Let \(Y_+\) denote the positive cone of \(Y\), and define the convex cone
\[
\mathcal C
:=
\left\{
\left(
(-B_iU-w_i)_{i=1}^N,\,
\int_\Gamma\frac{U}{\theta_*}\,d\mu-r
\right):
U\in K,\ (w_i)_{i=1}^N\in Y_+,\ r\ge0
\right\}.
\]
I now show that \((0,1)\notin\overline{\mathcal C}\). Otherwise, there would
exist \(U_n\in K\), \(w_n\in Y_+\), and \(r_n\ge0\) such that
\[
(-B_iU_n-w_{n,i})_{i=1}^N\to0,
\qquad
\int_\Gamma\frac{U_n}{\theta_*}\,d\mu-r_n\to1.
\]
Since $\|(-B_iU_n-w_{n,i})_{i=1}^N\|_Y\longrightarrow0$ and \(w_{n,i}\geq0\), we have $B_iU_n
\le
\|(-B_jU_n-w_{n,j})_{j=1}^N\|_Y$ on $R_i$ for every $i$. By Lemma~\ref{lem:R-repair}, there is
\(\hat U_n\in\mathcal U\) such that
\[
\|U_n-\hat U_n\|_\infty
\le
C\|(-B_iU_n-w_{n,i})_{i=1}^N\|_Y
\longrightarrow0.
\]
Hence, \eqref{eq:KR-nonpositive-no-lip} gives
\[
\limsup_{n\to\infty}
\int_\Gamma\frac{U_n}{\theta_*}\,d\mu
\le
\limsup_{n\to\infty}
\left\|\frac{\mu}{\theta_*}\right\|_{TV}
\|U_n-\hat U_n\|_\infty
=
0.
\]

This contradicts \(r_n\ge0\) and
\(\int_\Gamma U_n/\theta_*\,d\mu-r_n\to1\).

Since \(\overline{\mathcal C}\) is a closed convex cone, Hahn--Banach
separation yields \((\Lambda,\lambda)\in Y^*\times\mathbb R\) that strictly
separates \((0,1)\) from \(\overline{\mathcal C}\). Thus $\Lambda(z)+\lambda a\le0$ for every $(z,a)\in\overline{\mathcal C}$ and $\lambda>0.$ Normalize \(\lambda=1\). Since \((-w,0)\in\mathcal C\) for every
\(w\in Y_+\), the functional \(\Lambda\) is positive. Moreover, taking
\(w=0\) and \(r=0\) gives
\begin{equation}
\label{eq:functional-dual-no-lip}
\int_\Gamma\frac{U}{\theta_*}\,d\mu
\le
\Lambda\bigl((B_iU)_{i=1}^N\bigr)
\qquad\text{for every }U\in K.
\end{equation}
By the Riesz representation theorem, there exist finite positive Borel
measures \(\gamma_i\in\mathcal M_+(R_i)\) such that
\[
\Lambda\bigl((B_iU)_{i=1}^N\bigr)
=
\sum_{i=1}^N\int_{R_i}B_iU\,d\gamma_i.
\]
Substituting this identity into
\eqref{eq:functional-dual-no-lip} shows that
\((\gamma_1,\dots,\gamma_N)\) is a dual certificate.

\subsection{Proof of Lemma \ref{lem:regionwise_exact_certificate}}

Let \((\gamma_1,\dots,\gamma_N)\) be a dual certificate. Applying
\eqref{eq:dualcert} to \(U_{\mathrm{pure}}\), and using
\eqref{eq:primalzeropure} and \(B_iU_{\mathrm{pure}}\le0\), gives
\[
0
\le
\sum_{i=1}^N\int_{R_i}B_iU_{\mathrm{pure}}\,d\gamma_i
\le0.
\]
Hence each \(\gamma_i\) is supported on $E_i
:=
\{(\theta,\theta')\in R_i:
B_iU_{\mathrm{pure}}(\theta,\theta')=0\}.$ Since every \(U\) satisfies \(B_iU=0\) on $Z_i
:=
\{(\theta,\theta')\in R_i:\theta_i=\theta_i'=0\},$ we may discard the mass on \(Z_i\).

Fix \((\theta,\theta')\in E_i\setminus Z_i\). Then
\(\theta_i,\theta_i'>0\), since otherwise
\(B_iU_{\mathrm{pure}}(\theta,\theta')=0\) would contradict
\(U_{\mathrm{pure}}>0\). Since \((\theta,\theta')\in R_i\), we get $\frac{q_k\theta_k}{\theta_i}
\ge
\frac{q_k\theta_k'}{\theta_i'}$ for every $k$. Moreover, \(B_iU_{\mathrm{pure}}(\theta,\theta')=0\) implies
\[
\max_k\frac{q_k\theta_k}{\theta_i}
=
\frac{U_{\mathrm{pure}}(\theta)}{\theta_i}
=
\frac{U_{\mathrm{pure}}(\theta')}{\theta_i'}
=
\max_k\frac{q_k\theta_k'}{\theta_i'}.
\]
If \(j\) maximizes the last expression, the corresponding coordinate
inequality must bind. Hence
\[
q_j\theta_j=U_{\mathrm{pure}}(\theta),
\qquad
q_j\theta_j'=U_{\mathrm{pure}}(\theta'),
\]
so \(\theta,\theta'\in\Gamma_j\). Therefore $E_i\setminus Z_i
\subseteq
\bigcup_{j=1}^N(\Gamma_j\times\Gamma_j).$ To obtain a disjoint partition, define
\[
S_{i1}
:=
(E_i\setminus Z_i)\cap(\Gamma_1\times\Gamma_1),
\quad \text{and, for \(j\ge2\),} \quad
S_{ij}
:=
(E_i\setminus Z_i)\cap(\Gamma_j\times\Gamma_j)
\setminus
\bigcup_{l<j}(\Gamma_l\times\Gamma_l).
\]
Let \(\eta_{ij}:=\gamma_i|_{S_{ij}}\), so that
\(\gamma_i=\sum_j\eta_{ij}\). We now reassign \(\eta_{ij}\) from constraint \(i\) to constraint \(j\).
On \(S_{ij}\), we have $U_{\mathrm{pure}}(\theta)=q_j\theta_j$ and $U_{\mathrm{pure}}(\theta')=q_j\theta_j'$. \(B_iU_{\mathrm{pure}}=0\) therefore gives
\begin{equation}\label{eq:commratio}
\frac{\theta_i}{\theta_j}
=
\frac{\theta_i'}{\theta_j'}.
\end{equation}
Together with \((\theta,\theta')\in R_i\), this implies
\((\theta,\theta')\in R_j\). Define $d\tilde\eta_{ij}(\theta,\theta')
:=
\frac{\theta_i}{\theta_j}\,
d\eta_{ij}(\theta,\theta').$ This measure is finite because, on \(\Gamma_j\),
\(q_j\theta_j=U_{\mathrm{pure}}(\theta)\ge \min_kq_k/N\), so
\(\theta_j\) is bounded away from zero. Moreover, by \eqref{eq:commratio},
\[
\frac{\theta_i}{\theta_j}
B_jU(\theta,\theta')
=
B_iU(\theta,\theta'),
\]
and hence
\begin{equation}\label{eq:measuretransform}
\int_{R_j}B_jU\,d\tilde\eta_{ij}
=
\int_{S_{ij}}B_iU\,d\eta_{ij}
\qquad
\text{for every }U\in K.
\end{equation}

Finally, set $\bar\gamma_j:=\sum_{i=1}^N\tilde\eta_{ij}.$ Then $\bar\gamma_j\in\mathcal M_+(R_j),$ with $\operatorname{supp}(\bar\gamma_j)
\subseteq
R_j\cap(\Gamma_j\times\Gamma_j)$. Summing \eqref{eq:measuretransform} gives
\[
\sum_{j=1}^N\int_{R_j}B_jU\,d\bar\gamma_j
=
\sum_{i=1}^N\int_{R_i}B_iU\,d\gamma_i
\qquad
\text{for every }U\in K.
\]
Thus the original certificate inequality continues to hold with
\((\bar\gamma_1,\dots,\bar\gamma_N)\), which is therefore a regionwise dual
certificate.

\subsection{Proof of Lemma \ref{lem:certificate-transport-equivalence}}

Suppose \((\gamma_i)_{i=1}^N\) is a regionwise dual certificate. For each
\(i\), define
\(d\pi_i(\theta,\theta'):=\theta_i\theta_i'\,
d\gamma_i(\theta,\theta')\), and let \(\rho_i\) and \(\tau_i\) denote the
first and second marginals of \(\pi_i\). Since
\(\operatorname{supp}(\gamma_i)\subseteq
R_i\cap(\Gamma_i\times\Gamma_i)\), the measure \(\pi_i\) is a
\(\succ_i\)-monotone transport plan from \(\rho_i\) to \(\tau_i\). Define the finite signed measure
\(\tilde\mu_i:=\tau_i-\rho_i\) on \(\Gamma_i\). Then, for every \(U\in K\),
\[
\int_{\Gamma_i}\frac{U(\theta)}{\theta_i}\,
d\tilde\mu_i(\theta)
=
\int_{R_i}B_iU(\theta,\theta')\,
d\gamma_i(\theta,\theta').
\]
Summing over \(i\) and applying the certificate inequality gives
\[
\int_\Gamma\frac{U}{\theta_*}\,d\mu
\le
\sum_{i=1}^N
\int_{\Gamma_i}\frac{U}{\theta_i}\,d\tilde\mu_i,
\]
which is the preprocessing condition. It remains to verify the regionwise transport condition. Since \(\pi_i\)
moves mass only upward in the \(\succ_i\)-order, for every closed
\(\succ_i\)-upper set \(C\subseteq\Gamma_i\) we have
\(\rho_i(C)\le\tau_i(C)\), and therefore
\(\tilde\mu_i(C)\ge0\). Moreover,
\(\tilde\mu_i(\Gamma_i)=0\), because \(\rho_i\) and \(\tau_i\) have the same
total mass. Hence
\(\tilde\mu_i^+(\Gamma_i)=\tilde\mu_i^-(\Gamma_i)\). By
Theorem~\ref{thm:stoch_order_i}, there exists a
\(\succ_i\)-monotone transport plan from \(\tilde\mu_i^-\) to
\(\tilde\mu_i^+\).

Conversely, suppose
\((\tilde\mu_1,\dots,\tilde\mu_N)\) satisfies preprocessing and regionwise
transport. For each \(i\), let \(\pi_i\) be a
\(\succ_i\)-monotone transport plan from \(\tilde\mu_i^-\) to
\(\tilde\mu_i^+\). Since \(\Gamma_i\) is compact and
\(\theta_i>0\) everywhere on \(\Gamma_i\), the coordinate \(\theta_i\) is
bounded away from zero on \(\Gamma_i\). Hence the finite nonnegative measure
\(\gamma_i\) defined by
\(d\gamma_i(\theta,\theta')
:=(\theta_i\theta_i')^{-1}d\pi_i(\theta,\theta')\)
is well defined. Its support is contained in
\(R_i\cap(\Gamma_i\times\Gamma_i)\). For every \(U\in K\),
\[
\int_{R_i}B_iU\,d\gamma_i
=
\int_{\Gamma_i}\frac{U}{\theta_i}\,d\tilde\mu_i.
\]
Therefore, summing over \(i\) and using preprocessing,
\[
\int_\Gamma\frac{U}{\theta_*}\,d\mu
\le
\sum_{i=1}^N
\int_{\Gamma_i}\frac{U}{\theta_i}\,d\tilde\mu_i
=
\sum_{i=1}^N
\int_{R_i}B_iU\,d\gamma_i.
\]
Thus \((\gamma_i)_{i=1}^N\) is a regionwise dual certificate.

\section{Deriving examples}\label{sec:verifyex3}

\subsection{Optimal mechanism in Example~\ref{example:ex3}}

Let us verify that the optimal mechanism in Example~\ref{example:ex3} offers only pure options. For \(\rho>0\), define the welfare-cost ratio
\[
R_\rho(z):=\frac{w(z)}{\rho\varphi_1(z)+\varphi_2(z)}.
\]
For \(z\in[0,\tfrac12]\), direct calculation gives
\[
\mathbb P(\Theta_2\le z)=\frac{z}{2(1-z)},
\quad
\mathbb P(\Theta_2\ge z)=\frac{2-3z}{2(1-z)}, 
\quad 
\text{and}
\quad
\mathbb E \left[\max\{zV_1,(1-z)V_2\}\right]
=\frac{4z^2-6z+3}{6(1-z)}.
\]
For \(z\in[\tfrac12,1]\),
\[
\mathbb P(\Theta_2\le z)=\frac{3z-1}{2z},
\quad
\mathbb P(\Theta_2\ge z)=\frac{1-z}{2z}, 
\quad 
\text{and}
\quad
\mathbb E \left[\max\{zV_1,(1-z)V_2\}\right]
=\frac{4z^2-2z+1}{6z}.
\]
Since $\mathbb E \left[\lambda(\Theta)\max\{z\Theta_1,(1-z)\Theta_2\}\right]
=
\mathbb E \left[\max\{zV_1,(1-z)V_2\}\right],$ we obtain
\[
R_\rho(z)=
\begin{cases}
\displaystyle
\frac13\,
\frac{4z^2-6z+3}{(\rho+3)z^2-5z+2},
& z\in[0,\tfrac12], \\[10pt]
\displaystyle
\frac13\,
\frac{4z^2-2z+1}{(3\rho+1)z^2-(\rho+2)z+1},
& z\in[\tfrac12,1].
\end{cases}
\]
Note that $R_\rho(z)=\frac{1}{\rho}R_{1/\rho}(1-z).$ Thus \(R_\rho\) has a unique maximizer if and only if \(R_{1/\rho}\) does,
so it is enough to consider \(\rho\le1\).

Now, on \([0,\tfrac12]\), differentiating the first expression gives
\[
R_\rho'(z)
=
\frac{(6\rho-2)z^2-(6\rho+2)z+3}
{3\bigl((\rho+3)z^2-5z+2\bigr)^2}.
\]
The numerator is strictly decreasing on this interval and equals
\(\tfrac32(1-\rho)\ge0\) at \(z=\tfrac12\). Hence \(R_\rho\) is strictly
increasing on \([0,\tfrac12]\). On \([\tfrac12,1]\), differentiating the second expression gives
\[
R_\rho'(z)
=
\frac{2\rho z^2-6\rho z+\rho-6z^2+6z}
{3\bigl((3\rho+1)z^2-(\rho+2)z+1\bigr)^2}.
\]
The numerator is strictly decreasing in \(z\), is nonnegative at
\(z=\tfrac12\) (with equality only when \(\rho=1\)), and is negative at
\(z=1\). Combining the two intervals shows that \(R_\rho\) has a unique
maximizer on \([0,1]\). By symmetry, the same conclusion
holds for every \(\rho>0\).

Let \((c_1^*,c_2^*)\gg0\) be an optimal dual vector for
\eqref{eq:primalLP}, and set \(\rho^*:=c_1^*/c_2^*\). Dual feasibility gives
\[
R_{\rho^*}(z)
=
\frac{w(z)}
{\rho^*\varphi_1(z)+\varphi_2(z)}
\le c_2^*
\qquad\text{for every }z\in[0,1].
\]
Complementary slackness then implies that equality holds
\(\nu^*\)-a.e. for every optimal primal measure \(\nu^*\).
Since an optimal measure is nonzero, it follows that $c_2^*=\max_{z\in[0,1]}R_{\rho^*}(z)$ and that every optimal measure is supported on $\arg\max_{z\in[0,1]}R_{\rho^*}(z).$ Since this set is a singleton in \((0,1)\), every optimal measure is supported on one interior cutoff.

\subsection{Optimal mechanism in Example~\ref{example:ex2}}

By symmetry, the pure option mechanism offers \(2s\) units
of either good. Its welfare is
\[
W_{\mathrm{pure}}
=
2s\,\mathbb E[\max\{V_1,V_2\}].
\]
Fix \(\tau\in(1/2,1)\) and consider instead the following menu:
\[
\left\{
q_\tau e_1,\,
q_\tau e_2,\,
\tau q_\tau\mathbf 1
\right\},
\ \ \ \text{where}
\ \ \ p_\tau
:=
\mathbb P\bigl(\max\{V_1,V_2\}\leq \tau(V_1+V_2)\bigr),
\quad
q_\tau
:=
\frac{2s}{1+p_\tau(2\tau-1)}.
\]
Under this menu, type \(V\) obtains utility $q_\tau\max\{V_1,V_2\}$ from her preferred pure option and utility $\tau q_\tau(V_1+V_2)$ from the bundle. She therefore chooses the bundle exactly when $\max\{V_1,V_2\}\leq\tau(V_1+V_2).$ Note also that the corresponding mechanism satisfies supply constraints: by symmetry, the aggregate use of either
good is
\[
q_\tau\frac{1-p_\tau}{2}
+
\tau q_\tau p_\tau
=
q_\tau\frac{1+p_\tau(2\tau-1)}{2}
=
s.
\]
The welfare from this menu is given by
\[
W_\tau
=
q_\tau
\mathbb E\left[
\max\left\{
\max\{V_1,V_2\},
\tau(V_1+V_2)
\right\}
\right].
\]
We now verify it is higher than that under the pure option mechanism. Since
\[
\max\left\{
\max\{V_1,V_2\},
\tau(V_1+V_2)
\right\}
=
\max\{V_1,V_2\}
+
\bigl(\tau(V_1+V_2)-\max\{V_1,V_2\}\bigr)^+,
\]
we have \(W_\tau>W_{\mathrm{pure}}\) if and only if
\begin{equation}
\label{eq:ex2-sufficient}
\mathbb E\left[
\bigl(\tau(V_1+V_2)-\max\{V_1,V_2\}\bigr)^+
\right]
>
p_\tau(2\tau-1)
\mathbb E[\max\{V_1,V_2\}].
\end{equation}
Consider first the case of \(\varepsilon=0\). Then, common-need agents have $V_1=V_2=1$; since \(\tau>1/2\), they strictly prefer the bundle:
\[
\max\{V_1,V_2\}
=
1
<
2\tau
=
\tau(V_1+V_2).
\]
Almost all specialized-need agents strictly prefer the corresponding pure option because
\[
\max\{V_1,V_2\}
=
V_1+V_2
>
\tau(V_1+V_2).
\]
Consequently, \(p_\tau=1/2\), and
\[
\mathbb E\left[
\bigl(\tau(V_1+V_2)-\max\{V_1,V_2\}\bigr)^+
\right]
=
\frac{2\tau-1}{2},
\qquad
\mathbb E[\max\{V_1,V_2\}]
=
\frac34.
\]
Thus, the difference between the two sides of
\eqref{eq:ex2-sufficient} equals
\[
\frac{2\tau-1}{2}
-
\frac12(2\tau-1)\frac34
=
\frac{2\tau-1}{8}
>
0.
\]
Now, note that all the objects appearing in \eqref{eq:ex2-sufficient} vary
continuously as \(\varepsilon\downarrow0\). Since values are uniformly bounded
and the limiting distribution assigns probability zero to the indifference
event $\max\{V_1,V_2\}=\tau(V_1+V_2)$, the strict inequality continues to hold for all sufficiently small
\(\varepsilon>0\). We then get $W_\tau>W_{\mathrm{pure}},$ so the pure option mechanism is not optimal. Finally, by Theorem~\ref{th:2good_asym_menu}, the optimal mechanism offers two pure options and one bundle.

\subsection{Example \ref{ex:iid_reverse_hazard_examples}}

For the power family, $x\frac{f_M(x)}{F_M(x)}=\alpha,$ which is constant; the uniform distribution is the special case with \(\alpha=1\). For the truncated exponential family, $x\frac{f_M(x)}{F_M(x)}
=
\frac{\beta x}{e^{\beta x}-1}.$ This is decreasing in \(x\), since \(z/(e^z-1)\) is decreasing on
\((0,\infty)\). Thus, each family satisfies
\eqref{eq:reverse_hazard_elasticity}, and the result follows from
Corollary~\ref{cor:iid_reverse_hazard_pure}.

\subsection{Example \ref{ex:uniform-cube-BUS}}

We first change variables within each choice region, then compute the rent measure in these coordinates, and finally establish the uniform lower bound required by \eqref{eq:BUS}.

\paragraph{Ratio coordinates.}
The market-clearing pure-option vector is symmetric. Hence, up to null
indifference sets,
\[
\Gamma_i
=
\{\theta\in\Gamma:\theta_i\ge\theta_k\text{ for all }k\}.
\]
Fix \(i\), let \(Q:=[0,1]^{N-1}\), and identify \(\Gamma_i\) with \(Q\)
through $R_i(\theta)
:=
\left(\frac{\theta_k}{\theta_i}\right)_{k\neq i}.$ For \(r=(r_k)_{k\neq i}\in Q\), the inverse map is
\begin{equation}\label{eq:invratiomap}
\theta_i(r)=\frac{1}{1+\sum_{k\neq i}r_k},
\qquad
\theta_k(r)=\frac{r_k}{1+\sum_{l\neq i}r_l}
\quad(k\neq i).
\end{equation}
These coordinates are useful because they turn the relevant order into the
coordinatewise order: if \(C\subseteq\Gamma_i\) is a closed
\(\succ_i\)-upper set, then \(D:=R_i(C)\) is a closed downset in \(Q\).

We first use this transformation to compute the primitives and supply
costs. Since \(V\) is uniform on \([0,1]^N\), the ray
\(\{t\theta:t\ge0\}\) remains in the cube exactly when
\(0\le t\le1/\max_j\theta_j\). Thus, for some normalizing constant
\(K_N>0\),
\[
g(\theta)=K_N\left(\max_j\theta_j\right)^{-N},
\qquad
\lambda(\theta)
=
\frac{N}{N+1}\frac{1}{\max_j\theta_j}.
\]
On \(\Gamma_i\), this becomes
\[
g(\theta)=K_N\theta_i^{-N},
\qquad
\lambda(\theta)\theta_i=\frac{N}{N+1}.
\]
The Jacobian of \(R_i^{-1}\) is proportional to \(\theta_i^N\), which
cancels the variation in \(g\). Hence the pushforward of \(g\,dm\) is
constant, and symmetry gives
\[
(R_i)_\#(g\,dm)=\frac1N\,dr,
\qquad
M_i=\frac1N,
\qquad
A_i=\frac1{N+1}.
\]
By symmetry, the supply costs are equal. Since each row of \(J\) sums to
\(M_i=1/N\), the equation \(Jc=A\) gives $c_i=c_0:=\frac{N}{N+1}$ for every $i$.

\paragraph{The rent measure.}
We next compute \(\mu_i\) in ratio coordinates. 

First, the vector field $(c-(\sum_jc_j)\theta)g
=
c_0(1-N\theta)g$ has \(r_k\)-component density
\[
b_k(r)
=
\frac{c_0}{N}(1-r_k)
\left(1+\sum_{l\neq i}r_l\right);
\quad \text{therefore}\quad 
\theta_i(r)\sum_{k\neq i}\partial_{r_k}b_k(r)
=
-\frac{c_0\sum_{k\neq i}r_k}
{1+\sum_{k\neq i}r_k}.
\]
The remaining interior terms contribute
\[
\frac{c_0}{N}
\left(
1-\frac{N}{1+\sum_{k\neq i}r_k}
\right)dr,
\]
so the total interior density is \(-\frac{N-1}{N+1}\). Each lower face
\(\{r_k=0\}\) contributes
\(\frac1{N+1}\mathcal H^{N-2}\). Hence
\[
(R_i)_\#\mu_i
=
\frac1{N+1}
\left[
\sum_{k\neq i}
\mathcal H^{N-2}\restriction_{\{r_k=0\}}
-
(N-1)\mathcal L^{N-1}\restriction_Q
\right].
\]
Thus, for every closed \(\succ_i\)-upper set
\(C\subseteq\Gamma_i\),
\begin{equation}\label{eq:uniform-cube-mu-D}
\mu_i(C)
=
\frac1{N+1}
\left[
\sum_{k\neq i}
\mathcal H^{N-2}(D\cap\{r_k=0\})
-
(N-1)\mathcal L^{N-1}(D)
\right],\quad \text{where } D:=R_i(C).
\end{equation}

\paragraph{Uniform slack.}
By \eqref{eq:uniform-cube-mu-D}, it remains to show that the positive
lower-face mass uniformly dominates the negative interior mass for every
downset. For Borel \(B\subseteq Q\), define
\[
\bar\alpha(B)
:=
\mathcal L^{N-1}(B)
+
\sum_{k\neq i}
\mathcal H^{N-2}(B\cap\{r_k=0\}).
\]
We now show that there exists \(\kappa_N>0\) such that every closed downset
\(D\subseteq Q\) satisfies
\begin{equation}\label{eq:downset-gap}
\sum_{k\neq i}
\mathcal H^{N-2}(D\cap\{r_k=0\})
-
(N-1)\mathcal L^{N-1}(D)
\ge
\kappa_N
\min\{\bar\alpha(D),\bar\alpha(Q\setminus D)\}.
\end{equation}
Write
\[
V:=\mathcal L^{N-1}(D),
\qquad
P_k:=\mathcal H^{N-2}(D\cap\{r_k=0\}),
\qquad
S:=\sum_{k\neq i}P_k.
\]
If \(V=0\), the left-hand side of \eqref{eq:downset-gap} equals
\(S=\bar\alpha(D)\). If \(V=1\), closedness gives \(D=Q\), so both sides
vanish. If \(N=2\), every nonempty proper closed downset is \(D=[0,V]\),
and the left-hand side equals \(1-V=\bar\alpha(Q\setminus D)\). Hence
\eqref{eq:downset-gap} holds with \(\kappa_N=1\).

Suppose now that \(N\ge3\). Because \(D\) is a downset,
\(D\cap\{r_k=0\}\) is its coordinate projection onto that face.
The Loomis--Whitney and arithmetic--geometric mean inequalities therefore
give
\[
V^{N-2}
\le
\prod_{k\neq i}P_k
\le
\left(\frac{S}{N-1}\right)^{N-1},
\]
and hence
\begin{equation}\label{eq:bounddex4}
S\ge (N-1)V^{(N-2)/(N-1)}.
\end{equation}

Now, consider two cases. If \(V\le1/2\), the ratio $\frac{S-(N-1)V}{S+V}$ is increasing in \(S\). Thus, by \eqref{eq:bounddex4},
\[
\frac{S-(N-1)V}{S+V}
\ge
\frac{(N-1)\bigl(1-V^{1/(N-1)}\bigr)}
{N-1+V^{1/(N-1)}}
\ge
\frac{(N-1)\bigl(1-2^{-1/(N-1)}\bigr)}
{N-1+2^{-1/(N-1)}}
=:\kappa_1>0.
\]
Since \(S+V=\bar\alpha(D)\), we obtain $S-(N-1)V\ge\kappa_1\bar\alpha(D).$

If \(V\ge1/2\), let \(t:=V^{1/(N-1)}\). Then \eqref{eq:bounddex4} gives
\[
\frac{S-(N-1)V}{1-V}
\ge
\frac{(N-1)t^{N-2}}{1+t+\cdots+t^{N-2}}
\ge
2^{-(N-2)/(N-1)}
=:\kappa_2>0.
\]
Moreover,
\[
\bar\alpha(Q\setminus D)
=
N(1-V)-\bigl(S-(N-1)V\bigr)
\le N(1-V),
\quad \text{and hence}
S-(N-1)V
\ge
\frac{\kappa_2}{N}\bar\alpha(Q\setminus D).
\]
This proves \eqref{eq:downset-gap} with $\kappa_N
:=
\min\left\{\kappa_1,\frac{\kappa_2}{N}\right\}>0.$

It remains to translate this geometric bound from \(Q\) back to
\(\Gamma_i\). By \eqref{eq:invratiomap}, the Jacobians of \(R_i^{-1}\)
and its restrictions to the lower faces are bounded above by some
\(B_N<\infty\). The area formula therefore gives
\[
\alpha(R_i^{-1}(B))
\le B_N\bar\alpha(B)
\qquad
\text{for every Borel }B\subseteq Q.
\]
Applying this inequality to \(D\) and \(Q\setminus D\), and combining it
with \eqref{eq:uniform-cube-mu-D} and \eqref{eq:downset-gap}, yields
\[
\mu_i(C)
\ge
\frac{\kappa_N}{(N+1)B_N}
\min\{\alpha(C),\alpha(\Gamma_i\setminus C)\}.
\]
Therefore, \eqref{eq:BUS} holds with \(\eta=\frac{\kappa_N}{(N+1)B_N}>0\).

\subsection{Example \ref{ex:specialists-exclude-generalist}}

For every good \(i\),
\[
\max_{k\neq0}\eta_{ki}
\ge \eta_{ii}
=
(1-\varepsilon)\left(H+\frac{\delta}{2}\right)
+\varepsilon\frac{L}{2}
=
H+\frac{\delta}{2}+o(1),
\]
whereas \(\overline v_{0i}\le L+\delta\). Since
\(H-L>\delta/2\), it follows that
\(\max_{k\neq0}\eta_{ki}>\overline v_{0i}\) for every \(i\) and all
sufficiently small \(\varepsilon>0\). Part~$(i)$ of
Theorem~\ref{thm:3} therefore implies that group \(0\) is excluded.

\subsection{Example \ref{ex:parametric-comparative-advantage}}

Fix \(r\) and apply part~$(ii)$ of Theorem~\ref{thm:3} with \(B=\{r\}\). As
\(\varepsilon\downarrow0\),
\[
\eta_{rr}-\max_{l\neq r}\overline v_{lr}
\ge
H-L-\frac{\delta}{2}+o(1),
\]
as \(\eta_{rr}=H+\delta/2+o(1)\), while
\(\overline v_{lr}\le L+\delta\) for \(l\neq r\). For every \(i\neq r\),
group \(i\) remains available, so
\[
\max_{l\neq r}\overline v_{li}
-
\max_{k\neq r}\eta_{ki}
\le
H+\delta-\eta_{ii}
=
\frac{\delta}{2}+o(1).
\]
Consequently, the difference between the left- and right-hand sides of
\eqref{eq:comparative-advantage-participation-difference} is at least
\[
s_r(H-L)
-
\frac{\delta}{2}\sum_{i=1}^Ns_i
+
o(1).
\]
This is strictly positive for all sufficiently small \(\varepsilon>0\) by
\eqref{eq:parametric-comparative-advantage-gap}. Thus, group \(r\) is active
in every optimal supply split. Since \(r\) was arbitrary, every group is
active.

\end{document}